\documentclass[11pt,a4paper]{amsart}
\usepackage[foot]{amsaddr}
\usepackage{ifxetex}
\ifxetex
  \usepackage[no-math]{fontspec}
\else
\fi
\usepackage{amsmath}
\usepackage{amsfonts}
\usepackage{amssymb}
\usepackage{amsthm}
\usepackage{lmodern}
\usepackage{fullpage}
\usepackage{microtype}
\ifxetex 
  \usepackage[libertine]{newtxmath}
\else
  \usepackage{newtxmath}
\fi
\usepackage[tt=false]{libertine} 
\usepackage{caption}
\usepackage{tcolorbox}
\usepackage{bbm}
\usepackage{hyperref, color}
\hypersetup{colorlinks=true,citecolor=blue, linkcolor=blue, urlcolor=blue}
\usepackage[linesnumbered,boxed,ruled,vlined]{algorithm2e}
\usepackage{bm}
\usepackage{bbm}
\usepackage[numbers]{natbib}
\usepackage{xcolor}
\usepackage{enumerate} 
\usepackage{enumitem}
\usepackage{ragged2e}
\usepackage{tabularx}
\usepackage{array}
\usepackage{longtable}
\usepackage{hhline}
\usepackage{multirow}
\usepackage{cleveref}
\newcolumntype{L}[1]{>{\raggedright\arraybackslash}p{#1}}
\newcolumntype{C}[1]{>{\centering\arraybackslash}m{#1}}
\newcolumntype{R}[1]{>{\raggedleft\arraybackslash}p{#1}}

\usepackage{makecell}
\usepackage{footnote}
\usepackage{float}
\usepackage{xifthen}
\makesavenoteenv{tabular}

\newcommand{\TODO}[1]{\typeout{TODO: \the\inputlineno: #1}\textbf{{\color{red}[[[ #1 ]]]}}}

\newcommand{\Couple}{\mathsf{Couple}}
\newcommand{\unpin}[1]{{#1}^{\+O}}
\newcommand{\cp}{\textnormal{cp}}
\newcommand{\mg}{\textnormal{mg}}
\newcommand{\samp}{\textnormal{samp}}
\newcommand{\lp}{\textnormal{lp}}
\newcommand{\lpone}{\textnormal{lp1}}
\newcommand{\lptwo}{\textnormal{lp2}}
\newcommand{\coup}{\textnormal{coup}}
\newcommand{\trun}{\textnormal{trun}}
\newcommand{\vio}[1]{{\textnormal{\texttt{False}}(#1)}}
\DeclareMathOperator\supp{supp}

\newcommand{\vbl}{\textnormal{\textsf{vbl}}}

\newcommand{\dtv}{d_{\rm TV}}
\def\P{\mathop{\mathbf{Pr}}\nolimits}

\newcommand{\True}{\mathtt{True}}
\newcommand{\False}{\mathtt{False}}

 \renewcommand{\P}{\mathbf{P}}

\newcommand{\abs}[1]{\left\vert#1\right\vert}

\SetKwRepeat{Do}{do}{while}

\newtheorem{theorem}{Theorem}[section]

\newtheorem{claim}[theorem]{Claim}
\newtheorem*{claim*}{Claim}
\newtheorem{condition}[theorem]{Condition}

\newtheorem{lemma}[theorem]{Lemma}
\newtheorem{proposition}[theorem]{Proposition}

\newtheorem{corollary}[theorem]{Corollary}
\theoremstyle{definition}
\newtheorem{Condition}{Condition}
\newtheorem{definition}[theorem]{Definition}
\newtheorem{remark}[theorem]{Remark}
\newtheorem*{remark*}{Remark}

\renewcommand{\Pr}[2][]{ \ifthenelse{\isempty{#1}}
  {\mathop{\mathbf{Pr}}\left[#2\right]} {\mathop{\mathbf{Pr}}\limits_{#1}\left[#2\right]} }
  \renewcommand{\P}[2][]{ \ifthenelse{\isempty{#1}}
  {\mathop{\+P}\left[#2\right]} {\mathop{\+P}\limits_{#1}\left[#2\right]} }
\newcommand{\E}[2][]{ \ifthenelse{\isempty{#1}}
  {\mathop{\mathbf{E}}\left[#2\right]}
  {\mathop{\mathbf{E}}\limits_{#1}\left[#2\right]} }
\newcommand{\poly}{{\rm poly}}  

\usepackage[textsize=tiny]{todonotes}

\def\^#1{\mathbb{#1}} 
\def\*#1{\mathfrak{#1}} 
\def\+#1{\mathcal{#1}} 
\def\-#1{\mathrm{#1}} 
\def\=#1{\boldsymbol{#1}} 

\newcommand{\defeq}{\triangleq}

\usepackage{amssymb}
\usepackage{stackengine}
\usepackage{scalerel}
\usepackage{xcolor}
\usepackage{graphicx}

\newbool{doubleblind}
\setbool{doubleblind}{false}

\title{A Sampling Lov\'{a}sz Local Lemma for Large Domain Sizes}

\ifdoubleblind
  \author{Author(s)}
\else
\author{Chunyang Wang, Yitong Yin}
\address[Chunyang Wang, Yitong Yin]{State Key Laboratory for Novel Software Technology, New Cornerstone Science Laboratory, Nanjing University, 163 Xianlin Avenue, Nanjing, Jiangsu Province, 210023, China. \textnormal{E-mail: \url{wcysai@smail.nju.edu.cn}, \url{yinyt@nju.edu.cn}}}
\fi

\begin{document}


\allowdisplaybreaks

\begin{abstract}
We present polynomial-time algorithms for approximate counting and sampling solutions to constraint satisfaction problems (CSPs) with atomic constraints within the local lemma regime: 
\[
pD^{2+o_q(1)}\lesssim 1.
\]
When the domain size $q$ of each variable becomes sufficiently large, this almost matches the known lower bound  $pD^2\gtrsim 1$ for approximate counting and sampling solutions to atomic CSPs~\cite{BGG19,galanis2021inapproximability}, 
thus establishing an almost tight sampling Lov\'{a}sz local lemma for large domain sizes.

\end{abstract}

\maketitle
\setcounter{tocdepth}{1}
\tableofcontents

\newpage
\section{Introduction}\label{sec:intro}

Constraint satisfaction problems (CSPs) are ubiquitous in Computer Science,
and their solution spaces have been a subject of great research interest. 
A CSP is represented by a collection of constraints defined on a set of variables, where a solution is an assignment of variables such that all constraints are satisfied. 
A cornerstone tool in studying CSP solution spaces is the  \emph{Lov\'{a}sz local lemma}~\cite{LocalLemma}, which establishes the following sufficient condition for the existence of a CSP solution by interpreting the space of assignments as a product probability space and the violation of each constraint as a bad event:
\begin{equation}\label{eq:LLL-condition}
\mathrm{e}p(D+1)\leq 1,
\end{equation}
where $p$ stands for the maximum \emph{violation probability} of each constraint, 
and $D$ stands for the \emph{dependency degree}, given by the maximum number of other constraints that a constraint can share variables with. 
This condition~\eqref{eq:LLL-condition} was later shown to be essentially tight~\cite{She85}. Subsequent work on the \emph{algorithmic Lov\'{a}sz local lemma} seeks to constructively find a CSP solution by efficient algorithms. 
This has led to a long line of research~\cite{beck1991algorithmic,Alon91,MR98,CS00,Sri08,moser2009constructive,moser2010constructive}, culminating in algorithms for efficiently finding a CSP solution up to the condition in \eqref{eq:LLL-condition}. 
Together, these contributions establish a sharp threshold for the existence/construction of CSP solutions.

On the other hand, a considerable amount of work has been focused on the \emph{counting/sampling Lov\'{a}sz local lemma} ~\cite{BGG19,HSZ19,Moi19,guo2019counting,FGYZ20,feng2021sampling,Vishesh21sampling,Vishesh21towards,HSW21,galanis2021inapproximability,qiu2022perfect,feng2022improved,he2022sampling,he2022counting,qiu2024inapproximability}, 
which aims to characterize a local lemma type regime under which the problem of (approximately) counting or (almost uniformly) sampling CSP solutions is tractable. 
Hardness results in ~\cite{BGG19,galanis2021inapproximability} have shown that this counting/sampling variant of LLL requires a strictly stronger condition  $pD^2\lesssim 1$, where $\lesssim$ hides lower-order factors and constants.
This holds true even when restricted to some canonical sub-classes of CSPs, such as $k$-CNFs or hypergraph colorings.
Regarding upper bounds, the current state-of-the-art~\cite{he2022counting} shows that counting/sampling CSP solutions is efficiently solvable under the condition $pD^5\lesssim 1$. 
However, the correct threshold for the counting/sampling LLL is not yet clear.
The following question is fundamental to our understanding of the critical phenomenon for counting and sampling CSP solutions:

\begin{center}
     \emph{   Is $pD^2\lesssim 1$ the correct threshold for the counting/sampling Lov\'{a}sz local lemma?}
\end{center}


Despite numerous works on the topic and successive improvements on the algorithmic threshold, current techniques have encountered barriers towards closing such gaps due to their reliance on a ``freezing" paradigm.
Originally introduced by Beck~\cite{beck1991algorithmic} to handle the non-self-reducibility of the local lemma regime, the ``freezing" paradigm ensures that a local lemma type condition is invariantly satisfied under arbitrary pinnings produced in the process of incrementally constructing a satisfying CSP solution, which imposes slackness in the local lemma regime.
As a result, the best upper bound obtained using this paradigm was $pD^4\lesssim 1$~\cite{Sri08,Alon91,MR98}, 
where the extra $D^3$ factor came from the use of a structure called $\{2,3\}$-tree.
%

In the realm of the counting and sampling Lov\'{a}sz local lemma, Beck's technique is highly prevalent. 
Moitra's seminal work~\cite{Moi19} on counting $k$-SAT solutions introduced the method of ``mark/unmark'',
where a local lemma type condition is preserved under arbitrary pinning of marked variables.
%
This ``mark/unmark'' method also introduced slackness in the local lemma regime, and can be viewed as a static version of the ``freezing'' paradigm.
Subsequently, this static version of ``freezing'' was generalized to non-static settings~\cite{guo2019counting,Vishesh21towards,he2022sampling,he2022counting,he2023improved},
and also enabled the development of several Markov chain Monte Carlo approaches~\cite{FGYZ20,feng2021sampling,Vishesh21sampling,HSW21,galanis2022fast,chen2023from} for fast sampling.
To this day, the current best regime $pD^5\lesssim 1$ for sampling general CSPs still relies on the idea of ``freezing''. 
The slackness in this case arises from the use of a structure called generalized $\{2,3\}$-trees~\cite{he2022sampling,he2022counting}, 
which is similar to the source of slackness in~\cite{Sri08}. 

These advancements suggest a barrier in the current approach to counting and sampling LLL. 
Novel techniques are required to establish a tight counting and sampling Lov\'asz local lemma.


\subsection{Our results}

In this work, we establish the tractability of approximate counting and sampling solutions to atomic constraint satisfaction problems within an improved local lemma regime. 
As the domain size of each variable increases, this condition approaches $pD^2\lesssim 1$, 
providing a positive answer the major open question on counting and sampling LLL, 
particularly for large domain sizes.

A CSP instance (or a CSP formula) is denoted by $\Phi=(V,\+Q,\+C)$. Here, $V$ is a set of $n=|V|$ variables; each variable $v\in V$ is endowed with 
a  finite domain $Q_v$, altogether $\+Q\triangleq\bigotimes_{v\in V}Q_v$; and $\+C$ is a set of constraints, 
where each constraint $c\in \+C$ is a function $c:\bigotimes_{v\in \vbl(c)}Q_v\to\{\True,\False\}$ defined on a subset of variables, denoted by $\vbl(c)\subseteq V$.
A CSP formula $\Phi=(V,\+Q,\+C)$ is called \emph{atomic} if each constraint $c\in C$ is violated by exactly one assignment in the domain of $c$, i.e., $|c^{-1}(\False)|=1$. 
The atomicity of constraints is typical for many classical constraint satisfaction problems, including $k$-CNF and hypergraph colorings\footnote{For hypergraph $q$-colorings, each hyperedge may correspond to $q$ atomic constraints, each of which forbids the edge to be colored monochromatically with one particular color.}, and is a natural assumption in previous studies of LLL~\cite{achlioptas2014random,HV15,kolmogorov2016commutativity,harris2017algorithmic,HS17,HS19,AIS19,harris2021oblivious,feng2021sampling,Vishesh21sampling,HSW21}.

Given a CSP formula $\Phi=(V,\+Q,\+C)$, we use $\+P$ to denote the uniform (product) distribution over all possible assignments in $\+Q$.
The following parameters $\Phi$ can be defined:
\begin{itemize}
    \item \emph{width} $k\triangleq \max\limits_{c\in \+C}\abs{ {\vbl}(c)}$;
    \item \emph{minimum domain size} $q_{\min}\triangleq\min\limits_{v\in V}\abs{Q_v}$;
     \item \emph{maximum domain size} $q_{\max}\triangleq\max\limits_{v\in V}\abs{Q_v}$;
    \item \emph{dependency degree} $D\triangleq\max_{c\in C}\left|\{c'\in \+C\setminus \{c\}\mid \vbl(c)\cap \vbl(c')\neq \emptyset\}\right|$;
    \item \emph{maximum violation probability} $p\triangleq\max\limits_{c\in \+C}\+P[\neg c]$.
\end{itemize}

The following condition characterizes the local lemma regime achieved in this paper.

\begin{Condition}\label{condition:main-condition}
 $\Phi=(V,\+Q,\+C)$ is an atomic CSP formula satisfying
    \begin{equation}\label{eq:local-lemma-condition}
    (8\mathrm{e})^{3}\cdot p\cdot (D+1)^{2+\zeta}\leq 1,
\end{equation}
where
\[
\zeta=\frac{2\ln(2-1/q_{\min})}{\ln q_{\min}-\ln(2-1/q_{\min})}.
\]
\end{Condition}


Our main results are the following algorithms for counting/sampling LLL under \Cref{condition:main-condition}.

\begin{theorem}[counting LLL]\label{theorem:main-counting}
Assume \Cref{condition:main-condition} for $\Phi=(V,\+Q,\+C)$.
There exists a deterministic algorithm that, 
given any $\Phi$ and any $\varepsilon\in (0,1)$, 
outputs an estimate $\hat{Z}$ such that
\[
(1-\varepsilon)Z_{\Phi}\leq \hat{Z}\leq (1+\varepsilon)Z_{\Phi}
\]
within time $O\left(\left(\frac{n}{\varepsilon}\right)^{\poly(k,D,\log q_{\max})}\right)$, where $Z_{\Phi}$ represents the number of solutions to $\Phi$.
\end{theorem}

\begin{theorem}[sampling LLL]\label{theorem:main-sampling}
Assume \Cref{condition:main-condition} for $\Phi=(V,\+Q,\+C)$.
There exists an algorithm that, given any $\Phi$ and any $\varepsilon\in (0,1)$, 
outputs a random assignment $X\in \+Q$ distributed as $\hat{\mu}$ such that
    \[
    \dtv(\hat{\mu},\mu_{\Phi})\leq \varepsilon
    \]
within time $O\left(\left(\frac{n}{\varepsilon}\right)^{\poly(k,D,\log q_{\max})}\right)$, where $\mu_{\Phi}$ represents the uniform distribution over all solutions to $\Phi$.
\end{theorem}

\begin{remark}[local lemma regime]
The $\zeta=\zeta(q)$ in \Cref{condition:main-condition} decreases monotonically as $q$ increases.
In particular, as $q$ tends to infinity, $\zeta$ approaches 0, leading \Cref{condition:main-condition} to approach $pD^{2}\lesssim 1$, which matches the $pD^{2}\gtrsim 1$ lower bound~\cite{BGG19,galanis2021inapproximability}.
Conversely, with $q=2$ (the Boolean domain), $\zeta$~achieves its maximum value of $\log_{4/3}(9/4)\approx 2.82$.
In this case, \Cref{condition:main-condition} becomes $pD^{4.82}\lesssim 1$, improving upon the previous best upper bound of $pD^5\lesssim 1$ for counting/sampling LLL~\cite{he2022counting}.
%
%
%
\end{remark}

\begin{remark}[non-uniform width]
\Cref{condition:main-condition} does not involve the width $k$, 
making it applicable to CSP formulas with non-uniform widths.
This aspect is particularly desirable from the LLL perspective.
Previously, such generality was only attained in  \cite{feng2021sampling} under the local lemma condition $pD^{350}\lesssim 1$.
\end{remark}

\begin{remark}[time complexity]
    The $O\left(\left(\frac{n}{\varepsilon}\right)^{\poly(k,D,\log q)}\right)$ time complexities of \Cref{theorem:main-counting,theorem:main-sampling} align with  previous results for deterministic counting LLL~\cite{Moi19,guo2019counting,Vishesh21towards,he2022counting,feng2023towards}.
In fact, such complexity bounds with $k,D$ in the exponents appear to be intrinsic to deterministic approximate counting to this day, which relies on exhaustive enumerations of local structures.
\end{remark}

We then apply our results to two typical classes of atomic CSPs: hypergraph $q$-colorings and $k$-CNF.
\subsubsection{Application to hypergraph colorings}

Let $H=(V,\+E)$ be a $k$-uniform hypergraph, where $|e|=k$ for all $e\in \+E$. A proper hypergraph $q$-coloring $X\in [q]^V$ assigns one of the $q$ colors to each $v\in V$, ensuring no edge is monochromatic. We denote $\Delta$ as the maximum degree of the hypergraph, i.e., each vertex belongs to at most $\Delta$ hyperedges. Hypergraph colorings are foundational combinatorial objects that have been a key to the discovery of the Lov\'{a}sz Local lemma~\cite{LocalLemma}. The local lemma was initially devised to show that a proper $q$-coloring exists if $q\geq C\Delta^{\frac{1}{k}}$, where $C$ is a sufficiently large constant.

For the problem of approximate counting/sampling proper hypergraph $q$-colorings, the standard local Markov chains are efficient for sampling when $q\geq C\Delta$ for some constant $C$~\cite{Bordewich06stoppingtimes, Bordewich08pathcoupling}, which is also necessary for the irreducibility of local Markov chains~\cite{frieze2011randomly}. 
For \emph{simple} hypergraphs, where any two distinct hyperedges share at most one vertex, it has been shown that local Markov chains rapidly mix when $q\geq \max\{C_k\log n,500k^3\Delta^{\frac{1}{k-3}}\}$~\cite{frieze2011randomly,frieze2017randomly}, for hypergraphs with $n$ vertices. 
The first result for efficiently approximate counting/sampling hypergraph colorings under LLL condition is~\cite{guo2019counting}. By adapting Moitra's technique to a non-static ``freezing'', they obtained polynomial-time algorithms in the local lemma regime $q\geq 357\Delta^\frac{14}{k-14}$ for $k\ge 28$. 
This regime has been improved in several subsequent works~\cite{feng2021sampling,Vishesh21towards,Vishesh21sampling,HSW21}, resulting in the current best  upper bound $q\geq 310\Delta^{\frac{3}{k-3}}$ for $k\ge 24$ in ~\cite{Vishesh21sampling,HSW21}. 
On the lower bound side, it was proved in \cite{galanis2021inapproximability} that no efficient algorithm exists for the problem when  $q\gtrsim \Delta^{\frac{2}{k}}$ unless \textbf{NP}$=$\textbf{RP}.

Our results produce the following corollary on counting/sampling proper hypergraph $q$-colorings, 
closing the gap between the current upper and lower bounds for the problem.

\begin{corollary}[counting/sampling hypergraph $q$-colorings]\label{corollary:hypergraph-coloring}
For any $\zeta>0$, there exists a finite $q_0$ such that 
given any $k$-uniform hypergraph $H=(V,\+E)$ on $n$ vertices with maximum degree $\Delta$ and $q\ge q_0$ colors, if $k\geq 8$ and
    \[
     q\geq 70\Delta^{\frac{2+\zeta}{k-2-\zeta}},
    \]
    then for any $\varepsilon\in (0,1)$:
    \begin{itemize}
\item (approximate counting) 
the total number of proper $q$-colorings of $H$ can be estimated within relative error $1\pm \varepsilon$ deterministically within time $O\left(\left(\frac{n}{\varepsilon}\right)^{\poly(k,D,\log q)}\right)$;
\item (almost uniform sampling) 
an almost uniform random $q$-coloring $X\in [q]^V$ can be generated within time $O\left(\left(\frac{n}{\varepsilon}\right)^{\poly(k,D,\log q)}\right)$, such that the distribution of $X$ is $\varepsilon$-close (in total variation distance) to the uniform distribution over all proper $q$-colorings of $H$.
    \end{itemize}
\end{corollary}

%
The $q_0=q_0(\zeta)$ in \Cref{corollary:hypergraph-coloring} grows as $q_0=\exp(O(1/\zeta))$ in $\zeta$.

\subsubsection{Application to $k$-CNF}
A standard setting for LLL is the $k$-CNF (conjunctive normal form) with uniform width $k$ and maximum variable-degree $d$. 
A CNF is called a $(k,d)$-formula if each clause has exactly $k$ distinct variables and each variable is contained in at most $d$ distinct clauses.

For approximate counting/sampling $k$-CNF solutions,
the breakthrough of Moitra~\cite{Moi19} was the first to give an efficient algorithm in the LLL regime $k\geq 60\log k+60\log d+300$.
This bound was then improved in a series of works~\cite{FGYZ20,Vishesh21towards,Vishesh21sampling,HSW21,he2022counting}, 
with the current best bound given in ~\cite{he2022counting} as $k\geq 5\log d+5\log k+O(1)$. 
The lower bound proved in ~\cite{BGG19} says that unless \textbf{NP}$=$\textbf{RP}, $k\gtrsim 2\log d+2\log k$ is necessary for efficient algorithms to exist.

Our results improve the current best bound for approximate counting/sampling $k$-CNF solutions.

\begin{corollary}[counting/sampling $k$-CNF solutions]\label{corollary:k-CNF}
    Given any $(k,d)$-formula $\Phi$ on $n$ variables, if
    \[
      k\geq 4.82 \log k+4.82 \log d+5, 
    \]
    then for any $\varepsilon\in (0,1)$:
    \begin{itemize}
\item (approximate counting) 
the total number of satisfying assignments for $\Phi$ can be estimated within relative error $1\pm \varepsilon$ deterministically within time $O\left(\left(\frac{n}{\varepsilon}\right)^{\poly(k,D,\log q)}\right)$;
%
\item (almost uniform sampling) 
an almost uniform random satisfying assignment $X\in \{\textnormal{\texttt{True}, \texttt{False}}\}^V$ can be generated within time $O\left(\left(\frac{n}{\varepsilon}\right)^{\poly(k,D,\log q)}\right)$, such that the distribution of $X$ is $\varepsilon$-close (in total variation distance) to the uniform distribution over all satisfying assignments for $\Phi$.
%
    \end{itemize}
\end{corollary}

\subsection{Technique overview}\label{sec:technique-overview}
We then outline our general approach for the main results (\Cref{theorem:main-counting,theorem:main-sampling}) and provide a brief overview of the techniques.

Our high-level approach falls into the following two-step framework:
\begin{enumerate}
    \item 
    In the first step, we construct a coupling between the joint distributions defined by two CSP instances that differ locally.
    We manage to show that the discrepancy in this coupling decays at an exponential rate in the local lemma regime characterized by \Cref{condition:main-condition}.
    \item 
    Next, 
    we leverage the coupling constructed in the first step to devise a linear program. 
    By mimicking the transitions of the coupling procedure, this linear program can efficiently bootstrap the marginal ratios 
    as long as the original coupling procedure converges efficiently.
\end{enumerate}
This LP-based framework for counting LLL was initiated in the seminal work of Moitra~\cite{Moi19} and was subsequently generalized and refined in~\cite{guo2019counting,galanis2019counting,Vishesh21towards}. 
%
%

Despite this high-level framework,
our method differs fundamentally from all previous works and circumvents the barrier encountered by the previous approaches.

Our method completely dispenses with the ``freezing'' paradigm relied upon in previous works to establish correlation decay properties in the local lemma regime.
In addition, it also eliminates the need for ensuring a ``factorization property'' that after properly pinning the marked variables, the residual formula breaks down into logarithmic-sized connected components, 
which was required in previous works to achieve counting/sampling using marginal estimators through an exhaustive enumeration.\footnote{See the last paragraph of Section 1.2 in ~\cite{Vishesh21towards} for a discussion on the specific barriers of this paradigm.}



Specifically, we introduce a novel constraint-wise coupling,
whose exponential decay of correlation is due to an average-case percolation-style analysis,
thus avoiding to preserve a local-lemma-type condition with respect to a worst-case pinning.
Also, we propose a ``constraint-wise self-reducibility'' for local lemma regime to replace the ``factorization property''  to obtain counting/sampling algorithms, providing a new perspective for sampling/counting LLL.
Next, we provide a more in-depth discussion.

\subsubsection{Constraint-wise coupling with exponential decay of correlation}

The correlation decay properties have been a key to approximate counting and sampling in many classic and contemporary works \cite{dobrushin70prescribing,goldberg2005strong,weitz06counting,bayati2007simple,sinclair12approximation,li2013correlation, lu2013improved,yin2013approximate,gamarnik2015strong,patel2017deterministic,liu2019correlation,anari2020spectral,chen2020rapid,chen2021rapid,feng2021rapid}.
For sampling CSP solutions,
the correlation decay properties have also been used explicitly or implicitly in various previous algorithmic results~\cite{HSZ19, Moi19,BGG19,guo2019counting,galanis2019counting, FGYZ20,feng2021sampling, Vishesh21sampling, Vishesh21towards, HSW21,qiu2022perfect,feng2022improved,he2022sampling,galanis2022fast,he2022counting,feng2023towards,chen2023from,he2023improved}.

 Fix the variable set $V$ with domain $\+Q$, and for any set of constraints $\+C$ use $\mu_{\+C}$ to denote the uniform distribution over solutions to $\Phi=(V,\+Q,\+C)$. Our work establishes a correlation decay property within the local lemma regime of \Cref{condition:main-condition}. Such a correlation decay property is captured by the decay of discrepancy in a novel constraint-wise coupling between $\mu_{\+C}$ and the distribution with one fewer constraint $\mu_{\+C\setminus \{c_0\}}$ (formally stated in \Cref{lemma:correlation-decay}).

We expose the idea of this coupling. Decompose $\mu_{\+C\setminus \{c_0\}}$ using the law of total probability:
\[
\mu_{\+C\setminus \{c_0\}}=\mu_{\+C\setminus \{c_0\}}(c_0)\cdot \mu_{\+C}+\mu_{\+C\setminus \{c_0\}}(\neg c_0)\cdot \mu_{\+C\setminus \{c_0\}}(\cdot \mid  \neg c_0).
\]
Inspired from this, we can establish a coupling between $\mu_{\+C}$ and $\mu_{\+C\setminus \{c_0\}}$ as follows:
\begin{itemize}
    \item with probability $\mu_{\+C\setminus \{c_0\}}(c_0)$, establish a coupling between between identical distributions $\mu_{\+C}$ and $\mu_{\+C}$, which can already be perfectly coupled;
    \item with probability $\mu_{\+C\setminus \{c_0\}}(\neg c_0)$, establish a coupling between between $\mu_{\+C\setminus \{c_0\}}(\cdot \mid  \neg c_0)$ and $\mu_{\+C}$.
\end{itemize}
Here, in the latter case, a negated constraint $\neg c_0$ was introduced into the distribution. 
For atomic CSP, this will force the assignment on $\vbl(c_0)$ to be the unique violating assignment $\pi=\vio{c_0}$. 

To address this issue, we further decompose $\mu_{\+C}$ as follows:
\[
\mu_{\+C}=\sum\limits_{\rho\in \+Q_{\vbl(c_0)}}\mu_{\+C}(\rho)\cdot \mu_{\+C}(\cdot \mid \rho),
\]
which inspire us to achieve the coupling in this case as: first, sample $\rho\sim \mu_{\+C,\vbl(c_0)}$ according to the marginal distribution $\mu_{\+C,\vbl(c_0)}$ induced by $\mu_{\+C}$ on $\vbl(c_0)$, and then, establish the coupling between the conditional distributions $\mu_{\+C\setminus \{c_0\}}(\cdot \mid \pi)$ and $\mu_{\+C}(\cdot \mid \rho)$. 
In this way, we eliminated the negated constraint $\neg c_0$ from the distribution, albeit potentially introducing a discrepancy in the assignment over  $\vbl(c_0)$. It is important to note that the coupling $\mu_{\+C\setminus \{c_0\}}(\cdot \mid \pi)$ and $\mu_{\+C}(\cdot \mid \rho)$, after some simplification, is equivalent to the coupling of two distributions specified by two sets of constraints on the same variable set, thereby can be resolved recursively. 

The recursively constructed coupling in our approach may remind us of the recursive coupling technique introduced by Goldberg, Martin, and Mike to prove strong spatial mixing~\cite{goldberg2005strong}. 
However, there are significant differences between our approach and theirs. 
Firstly, our coupling proceeds in a constraint-wise manner. 
Secondly, the main distinction arises from the fact that the LLL regime is not self-reducible.
In their work~\cite{goldberg2005strong}, self-reducibility plays a crucial role in the recursive coupling for establishing strong spatial mixing. 
The one-step worst-case contraction bounded there allows for the application of path coupling~\cite{Bubley97pathcoupling:}. 
However, in local lemma regimes, self-reducibility is not guaranteed, and the LLL condition degrades after pinning too many variables. 
Consequently, we must keep track of the entire evolution of the coupling process and apply an average-case percolation-style analysis.
An intriguing finding is that the random choices involved in our constraint-wise coupling procedure can be explicitly identified using the principle of deferred decision. 
This simplifies the analysis considerably and may be of independent interest.

\subsubsection{Linear programs from constraint-wise coupling}
We will transform the exponentially decayed coupling process constructed above into a linear program that mimics the transition probabilities of this coupling procedure. This linear program will be used to bootstrap the marginal ratios between the total numbers of CSP solutions to $\Phi=(V,\+Q,\+C)$ and $(V,\+Q,\+C\setminus\{c_0\})$.

Since our coupling proceeds in a constraint-wise manner and its analysis is considerably non-local, 
it is significantly different from previous works within the same general framework. 
Therefore, several new ideas are required for designing the linear program as well as the final counting/sampling algorithms:
\begin{itemize}
\item 
Unlike in all previous LP-based methods for counting LLL~\cite{Moi19,guo2019counting, galanis2019counting,Vishesh21towards},
a key property known as ``local uniformity,'' which used to be enforced by freezing, is no longer guaranteed in our setting.
Instead of relying on ``local uniformity,'' which introduces slackness in the local lemma regimes,
we utilize a 2-tree structure that provides certificates for overflow in the coupling procedure, as well as a new class of constraints in the linear program.
This 2-tree structure, critical to our algorithm, closely aligns with similar structures that have appeared in the lower bounds~\cite{BGG19,galanis2021inapproximability}, shedding light on why our algorithms approach the critical threshold.

\item 
In the process of converting a marginal estimator to an approximate counting or sampling algorithm within local lemma regimes 
where variable-wise self-reducibility does not hold, 
previous works relied on the ``factorization'' to complete partial assignments. 
However, this approach introduces slackness in local lemma regimes.
Our finding shows that such additional slackness is unnecessary for efficient counting/sampling LLL.
%
%
Instead, we utilize ``constraint-wise self-reducibility'' for local lemma regimes and show that the followings are efficient:
\begin{itemize}
\item estimate the constraint-wise marginal probability $\mu_{\+C}(c)$; 
\item incrementally update some sample $X\sim \mu_{\+C\setminus \{c\}}$ to $Y\sim \mu_{\+C}$.
\end{itemize}
\end{itemize}

\section{Preliminaries and notations}\label{sec:prelim}

\subsection{CSP defined by atomic constraints}

\subsubsection{Basic concepts for CSPs}
A CSP is described by a collection of constraints defined on a set of variables. 
Formally, an instance of a constraint satisfaction problem, 
called a \emph{CSP formula}, 
is denoted by $\Phi=(V,\+Q(=\bigotimes_{v\in V}Q_v),\+C)$.
Here, $V$ is a set of $n=|V|$ random variables, where each random variable $v\in V$ is endowed with 
 a  finite domain $Q_v$ of size $q_v\triangleq\abs{Q_v}\ge 2$;
and $\+C$ gives a collection of local constraints, 
such that each  $c\in \+C$ is a constraint function $c:\bigotimes_{v\in \vbl(c)}Q_v\to\{\True,\False\}$ defined on a subset of variables, denoted by $\vbl(c)\subseteq V$.
An assignment $\sigma\in \+Q$ is called \emph{satisfying} for $\Phi$ if 
\[
\Phi(\sigma)\triangleq\bigwedge\limits_{c\in \+C} c\left(\sigma_{\vbl(c)}\right)=\True.
\]
%
%
%
%

In the context of LLL, each constraint $c$ can be interpreted as a bad event $A_c$, which happens when the assignment on $\vbl(c)$ violates $c$. For any subset of of constraints $\+E\subseteq \+C$, denote $\vbl(\+E)\triangleq\bigcup_{c\in \+E}\vbl(c)$. For any subset of variables $\Lambda\subseteq V$, denote $\+Q_{\Lambda}\defeq \bigotimes_{v\in \Lambda}Q_v$.  Also, we say $\Phi$ is \emph{satisfiable} is at least one satisfying assignment to $\Phi$ exists.

We say some constraint $c\in \+C$ is defined by \emph{atomic bad events}, or simply, \emph{atomic}, if it is violated by exactly one configuration in $\+Q_{\vbl(c)}$. For atomic constraints $c$, we denote $\vio{c}$ as the only violating configuration of $c$ in $\+Q_{\vbl(c)}$. Moreover, when all constraints in $\+C$ are atomic, we say $\Phi$ is atomic. 

\subsubsection{Notations for (partial) assignments}

For a partial assignment $\sigma\in \+Q_{\Lambda}$ specified over a subset of variables $\Lambda\subseteq V$, we use $\Lambda(\sigma)=\Lambda$ to denote the set of assigned variables in $\sigma$. For any partial assignment $\sigma$ and any $S\subseteq \Lambda(\sigma)$, we use $\sigma_S$ to denote $\bigotimes_{v\in S}\sigma(v)$. We further write $\sigma_v=\sigma_{\{v\}}$ for $v\in V$.

 For any two assignments $\sigma,\tau$ such that $\Lambda(\sigma)\cap \Lambda(\tau)=\emptyset$, we define $\sigma\land \tau\in \+Q_{\Lambda(\sigma)\cup \Lambda(\tau)}$ as the concatenation of $\sigma$ and $\tau$ such that for any $v\in \Lambda(\sigma)\cup \Lambda(\tau)$,
\[
(\sigma\land \tau)(v)=\begin{cases}\sigma(v) & v\in \Lambda(\sigma),\\ \tau(v)& v\in \Lambda(\tau).\end{cases}
\]

We will use $\varnothing$ to specifically denote an empty assignment,  distinguishing from the empty set $\emptyset$.

\subsubsection{Notations for events and probability measures }

We then specify some notations for events and probability measures related to the CSP.

\begin{definition}[simple notations for events]\label{definition:notation-events}
For the simplicity of notations, we will use:
\begin{itemize}
    \item any constraint $c\in \+C$ to denote the event that this constraint is satisfied;
    \item any subset of constraints $\+E\subseteq \+C$ to denote the event that all constraints in $\+E$ are satisfied;
    \item any partial assignment $\sigma$  to denote the event that the assignment on $\Lambda(\sigma)$ is preciesely $\sigma$. 
\end{itemize}

\end{definition}

Note that under this definition, the notation $\sigma\land \tau$ as a concatention of assignment is consistent with the same notation where $\sigma$ and $\tau$ are
considered as events, and $\sigma\land \tau$ is the logical and of the two events. 

 We use $\+P$ to denote the uniform product distribution over the space $\+Q$. For any subset of variables $\Lambda\subseteq V$.
We use $\mu=\mu_{\Phi}$ denote the distribution over all satisfying assignments of $\Phi$ induced  by $\+P$, i.e.
\[
    \mu_{\Phi}\defeq \+P\left(\cdot \mid \+C\right).
\]
$\mu_{\Phi}$ is well-defined only when $\Phi$ is satisfiable.

When the variable set $V$ and the domain $\+Q$ is clear, for some set of constraints $\+E$ defined over $V$, and some assignment $\sigma$ defined over $\Lambda(\sigma)$, we stipulate the following notations for (conditional) distributions:
\[
\mu_{\+E}\defeq \+P(\cdot \mid \+E),\quad \mu^{\sigma}_{\+E}\defeq \+P(\cdot \mid \+E\land \sigma).
\]

For some probability distribution $\mu$ and some subset of variables $\Lambda\subseteq V$, we use $\mu_{\Lambda}$ to denote the marginal distribution induced by $\mu$ on $\Lambda$. We use commas to separate multiple subscripts, for example, we use $\mu^{\sigma}_{\+E,\Lambda}$ to denote the marginal distribution induced by $\mu^{\sigma}_{\+E}$ on $\Lambda$.

\subsubsection{Pinned formula and pinned constraints}

For a subset of variables $\Lambda\subseteq V$ and a partial assignment $\sigma\in \+Q_{\Lambda}$ specified on $\Lambda$, the pinned formula $\Phi=(V,\+Q,\+C)$ under $\sigma$, denoted by $\Phi^\sigma=(V^\sigma,\+Q^{\sigma},\+C^\sigma)$, 
is a new CSP formula such that $V^\sigma=V\setminus \Lambda(\sigma)$, $\+Q^{\sigma}=\+Q_{V\setminus \Lambda(\sigma)}$
and the $\+C^\sigma$ is obtained from $\+C$ by: 
\begin{enumerate}
\item
replacing each original constraint $c\in\+C$ with the corresponding pinned constraint $c^\sigma$, where $\vbl(c^{\sigma})=\vbl(c)\setminus\Lambda(\sigma)$ and $c^{\sigma}(\tau)=c(\tau\land \sigma_{\Lambda(\sigma)\cap\vbl(c)})$ for any $\tau\in\+Q_{\vbl(c^{\sigma})}$;
\item
removing all the resulting constraints that have already been satisfied.
\end{enumerate}

%
Whenever a pinning $\sigma$  is applied to a CSP formula $\Phi=(V,\+Q,\+C)$, we always assume that $\sigma$ does not violate any constraint in $\+C$.
Under such assumption,  $\Phi^\sigma$ is always well-defined
and $\mu_{\Phi^\sigma}=\mu^{\sigma}_{V\setminus\Lambda(\sigma)}$. 
Moreover, if $\Phi$ is atomic, then so is the pinned formula $\Phi^\sigma$. 
We use $\+C^*$ to denote the set of all possible constraints obtained from pinning some constraint in $\+C$ with some non-violating $\sigma$, including the unpinned constraints in $\+C$.
Finally, for each (possibly) pinned constraint $c\in \+C^{*}$, we use $\unpin{c}$ to denote its original unpinned constraint in $\+C$.

\subsection{Lov\'{a}sz Local lemma}
  The celebrated Lov\'{a}sz local lemma gives a sufficient criterion for a CSP solution to exist:

\begin{theorem}[\cite{LocalLemma}]\label{locallemma}
    Given a CSP formula $\Phi=(V,\+Q,\+C)$, if the following holds
    \begin{align}\label{llleq}
    \exists x\in (0, 1)^{\+C}\quad \text{ s.t.}\quad \forall c \in \+C:\quad
        {\+P[\neg c]\leq x(c)\prod_{\substack{c'\in \+C\setminus \{c\}\\ \vbl(c)\cap\vbl(c')\neq \emptyset}}(1-x(c'))},
    \end{align}
    then  
    $$
        {\+P[\+C]\geq \prod\limits_{c\in C}(1-x(c))>0}.
    $$
\end{theorem}

When the condition in \eqref{llleq} is satisfied, 
the probability of any event in the uniform distribution $\mu$ over all satisfying assignments can be well approximated by the probability of the event in the product distribution.

\begin{theorem}[\text{\cite[Theorem 2.1]{haeupler2011new}}]\label{HSS}
Given a CSP formula $\Phi=(V,\+Q,\+C)$, if $\eqref{llleq}$ holds, 
then for any event $\+A$ that is determined by the assignment on a subset of variables $\vbl(\+A)\subseteq V$,
\[
   \+P[\+A\mid \+C]\leq \+P[\+A]\prod_{\substack{c\in \+C\\ \vbl(c)\cap\vbl(\+A)\neq \emptyset}}(1-x(c))^{-1}.
\]
\end{theorem}

\subsection{Dependency graph and $2$-trees}
The dependency graph is a key notion for the Lov\'{a}sz local lemma.
\begin{definition}[dependency graph]\label{def:dependency-graph}
Let $\Phi =(V,\+Q,\+C)$ be a CSP formula and let $\+E \subseteq \+C$ be a subset of constraints.
\begin{itemize}
\item $G_{\Phi}(\+E)$ denotes the dependency graph induced by $\+E$, 
which is a graph with vertex set $\+E$ such that there is an edge between distinct $c, c'\in \+E$
if and only if $\vbl(c) \cap \vbl(c') \neq \emptyset$.
\item $G^2_{\Phi}(\+E)$ denotes the square graph  of $G_{\Phi}(\+E)$, in which there is an edge between distinct $c, c'\in \+E$ if and only if  $\text{dist}_{G_{\Phi}}(c,c')\leq 2$.
\end{itemize}
\end{definition}

The notion of $2$-tree is an important combinatorial structure introduced by Alon~\cite{Alon91}, which has played key roles in algorithmic and sampling LLL.
\begin{definition}[$2$-tree]\label{definition:2-tree}
Let $G=(V,E)$ be a graph and $\text{dist}_G(\cdot,\cdot)$ denote the shortest path distance in~$G$.
A \emph{$2$-tree} in $G$ is a subset of vertices $T\subseteq V$ such that: 
\begin{itemize}
    \item  for any $u,v\in T$, $\text{dist}_G(u,v)\geq2$; 
    \item  $T$ is connected if an edge is added between each $u,v\in T$ such that $\text{dist}_G(u,v)=2$.
\end{itemize}
\end{definition}

The following lemma bounds the number of $2$-trees of a certain size in bounded-degree graphs.

\begin{lemma}[\text{\cite[Corollary 5.7]{FGYZ20}}]\label{lemma:num-2-tree}
Let $G = (V, E)$ be a graph with maximum degree $d$ and $v \in V$ be a vertex. Then the
number of $2$-trees in $G$ of size $\ell$ containing $v$ is at most $\frac{\left(\mathrm{e}d^2\right)^{\ell-1}}{2}$.
\end{lemma}

\subsection{Total variation distance and coupling}

 Let $\mu$ and $\nu$ be two probability distributions
over the same (finite) state space $\Omega$. Their total variation distance is defined by

\[
\dtv(\mu,\nu)=\frac{1}{2}\sum\limits_{x\in \Omega}|\mu(x)-\nu(x)|=\max\limits_{A\subseteq \Omega}\left(\mu(A)-\nu(A)\right).
\]

A coupling $\+D$ of two distributions $\mu$ and $\nu$ is a joint distribution over $\Omega\times \Omega$ whose projection on the first (or second)
coordinate is $\mu$ (or $\nu$). The well-known coupling lemma is often used to bound total variation distances, given as follows.

\begin{lemma}[\text{\cite[Proposition 4.7]{levin2017markov}}]\label{lemma:coupling-lemma}
   Let $\+D:(X,Y)$ be any coupling of $\mu$ and $\nu$, then
   \[
   \dtv(\mu,\nu)\leq \Pr[\+D]{X\neq Y}.
   \]
\end{lemma}

\section{Constraint-wise coupling with exponential decay of correlation}\label{sec:correlation-decay} 

We prove a correlation decay property for the uniform distribution $\mu$ over all satisfying assignments 
for the atomic CSP formula $\Phi$ satisfying \Cref{condition:main-condition}.
%
%
%
This correlation decay property is captured by the decay of discrepancy in a coupling between $\mu$ and the distribution with one fewer constraint.

\begin{theorem}\label{lemma:correlation-decay}
Let the CSP formula $\Phi=(V,\+Q,\+C)$ satisfy \Cref{condition:main-condition}.
Let  $c_0\in\+C$ be an arbitrary constraint.
%
There exists a coupling $(X,Y)$ of $\mu_{\+C}$ and $\mu_{\+C\setminus\{c_0\}}$, such that for any integer $K\geq 1$, 
\[
\Pr{d_{\mathrm{Ham}}(X,Y)\geq k\cdot (D+1)\cdot K}\leq 2^{-K},
\]
where 
$d_{\mathrm{Ham}}(X,Y)\triangleq\sum_{v\in V}I[X(v)\neq Y(v)]$ denotes the Hamming distance between $X$ and $Y$.
\end{theorem}

\noindent
\Cref{lemma:correlation-decay} says that in the local lemma regime characterized by \Cref{condition:main-condition},
the total influence of any particular constraint decays exponentially.
This is the first time that such an exponential decay of correlation has been found in a local lemma regime as \Cref{condition:main-condition}.
This is achieved by a new constraint-wise coupling process. 
The construction and analysis of this novel process completely bypass the use of Beck's freezing approach,
which introduced slackness and was used in almost every previous work.

The rest of this section is devoted to proving \Cref{lemma:correlation-decay}.

\subsection{A constraint-wise coupling}
First, we present the construction of the coupling in \Cref{lemma:correlation-decay}. 
Fix an instance of atomic CSP formula $\Phi=(V,\+Q,\+C)$.
The coupling is defined by a recursive procedure $\Couple(\+E,\+F,\sigma,\tau)$, 
which takes as input two sets $\+E,\+F\subseteq\+C^*$ of pinned constraints, 
along with two partial assignments $\sigma,\tau\in\+Q_{\Lambda}$ specified on some subset $\Lambda\subseteq V$ of variables, 
with the promise that both $\+E^\sigma$ and $\+F^\tau$ are satisfiable.
%
The procedure tries to produce a random pair of assignments $(X,Y)\in \+Q\times \+Q$ distributed according to a coupling of the joint distributions $\mu_{\+E}^{\sigma}$ and $\mu_{\+F}^{\tau}$.
%

This recursive procedure $\Couple(\+E,\+F,\sigma,\tau)$ is given in \Cref{Alg:couple}. 
An arbitrary total ordering is assumed on $\+C^*$, which contains all original constraints in $\+C$ and all their possible pinnings.

\begin{algorithm}
\caption{$\Couple(\+E,\+F,\sigma,\tau)$} \label{Alg:couple}
\SetKwInOut{Instance}{Instance}
\SetKwInOut{Input}{Input}
\SetKwInOut{Output}{Output}
\SetKwIF{WP}{ElseIf}{Else}{with probability}{do}{else if}{else}{endif}
\Instance{atomic CSP formula $\Phi=(V,\+Q,\+C)$;}
\Input{two subsets of pinned constraints $\+E,\+F\subseteq \+C^*$, and two partial assignments $\sigma,\tau\in\+Q_\Lambda$ specified on the same subset $\Lambda\subseteq V$ of variables;} 
\Output{a pair of assignments $(X,Y)\in \+Q \times \+Q$;}
\If{$\+E^\sigma=\+F^\tau$\label{Line:couple-return-cond}}{
let $(X,Y)$ be drawn according to the coupling of $\mu_{\+E}^{\sigma}$ and $\mu_{\+F}^{\tau}$ that always satisfies $X_{V\setminus \Lambda}=Y_{V\setminus \Lambda}$\;\label{Line:couple-perfect}
    \Return $(X,Y)$\;\label{Line:couple-return}
}
\eIf{$\+F^\tau\not\subseteq \+E^\sigma$ \label{Line:couple-swap-cond}}{
   choose the smallest $c\in \+F^\tau\setminus \+E^\sigma$\; \label{Line:couple-choose}
\eWP{$\mu_{\+E}^{\sigma}(c)$ \label{Line:couple-psat-cond}}{
    \Return $\Couple\left(\+E\cup\{c\},\+F,\sigma,\tau\right)$\; \label{Line:couple-psat-return}
}
{\label{Line:couple-else}
   let $\pi=\vio{c}$ and draw a random $\rho\sim \mu_{\+F,\vbl(c)}^\tau$\;\label{Line:couple-sample}
    \Return $\Couple\left(\+E,\+F, \sigma\land\pi, \tau\land\rho\right)$\; \label{Line:couple-assign-return}
}}
{\label{Line:couple-else-1}
choose the smallest $c\in \+E^\sigma\setminus \+F^\tau$\; \label{Line:couple-choose-1}
\eWP{$\mu_{\+F}^{\tau}(c)$ \label{Line:couple-psat-cond-1}}{
    \Return $\Couple\left(\+E,\+F\cup \{c\},\sigma,\tau\right)$\; \label{Line:couple-psat-return-1}
}
{
    draw a random $\pi\sim \mu_{\+E,\vbl(c)}^\sigma$ and let $\rho=\vio{c}$ \; \label{Line:couple-sample-1}
    \Return $\Couple\left(\+E,\+F, \sigma\land\pi, \tau\land\rho\right)$\; \label{Line:couple-assign-return-1}
}
}
\end{algorithm}


\begin{remark}[well-definedness of \Cref{Alg:couple}]
It is easy to see that for atomic CSP, if $\+E^\sigma=\+F^\tau$, the joint distributions $\mu_{\+E}^{\sigma}$ and $\mu_{\+F}^{\tau}$ are identically defined over ${V\setminus \Lambda}$. 
Thus the coupling in \Cref{Line:couple-perfect} must exist.
%

Also, note that the constraints $c$ chosen in \Cref{Line:couple-choose} and \Cref{Line:couple-choose-1} (which are used later) are pinned constraints, and are chosen according to the total ordering assumed over all pinned constraints in $\+C^*$.
\end{remark}

\subsubsection*{{Idea of the coupling.}}
\Cref{Alg:couple} formalizes the ideas outlined in \Cref{sec:technique-overview}.  
The coupling procedure maintains a pair of ``sub-instances with boundaries'' $(\+E,\sigma)$ and $(\+F,\tau)$, 
which initially are $(\+C\setminus\{c_0\},\varnothing)$ and $(\+C,\varnothing)$. 
%
The procedure tries to couple the joint distributions $\mu_\+E^{\sigma}$ and $\mu_\+F^{\tau}$. 
When the reduced sets of constraints $\+E^{\sigma}=\+F^{\tau}$, we have that $\mu_{\+E}^{\sigma}$ and $\mu_{\+F}^{\tau}$ are identically distributed over the set $V\setminus\Lambda$ of unassigned variables, 
and hence we can perfectly couple the unassigned variables (Lines~\ref{Line:couple-return-cond}-\ref{Line:couple-return}). 
Otherwise $\+E^{\sigma}\neq\+F^{\tau}$,  there must exist at least one (pinned) constraint $c\in\+E^{\sigma}\triangle\+F^{\tau}$, say $c\in \+F^{\tau}\setminus \+E^{\sigma}$ (Lines~\ref{Line:couple-swap-cond}-\ref{Line:couple-assign-return}).
The joint distribution $\mu_{\+E}^{\sigma}$ can be decomposed  based on the satisfaction and violation of $c$ as:
\[
\mu_{\+E}^{\sigma}=\mu^{\sigma}_{\+E}(c)\cdot \mu^{\sigma}_{\+E}(\cdot \mid c)+\mu^{\sigma}_{\+E}(\neg c)\cdot \mu^{\sigma}_{\+E\cup \{c\}}(\cdot \mid \neg c).
\]
The coupling algorithm then proceeds as follows according to the above decomposition:
\begin{itemize}
    \item With probability $\mu^{\sigma}_{\+E}(c)$, try to couple $\mu^{\sigma}_{\+E}(\cdot \mid c)$ and $\mu_{\+F}^{\tau}$ (Lines \ref{Line:couple-psat-cond}-\ref{Line:couple-psat-return}), which is solved recursively by the same coupling procedure on the new pair of sub-instances $(\+E\cup\{c\},\sigma)$ and $(\+F,\tau)$. 
    \item With probability $\mu^{\sigma}_{\+E}(\neg c)$, try to couple $\mu^{\sigma}_{\+E}(\cdot \mid \neg c)$ and $\mu_{\+F}^{\tau}$ (Lines \ref{Line:couple-else}-\ref{Line:couple-assign-return}). 
    Because $c$ is a (pinned) atomic constraint, $\mu^{\sigma}_{\+E}(\cdot \mid \neg c)$ fixes the assignment on $\vbl(c)$ to be $\pi=\vio{c}$. 
    On the other hand, 
    the joint distribution $\mu_{\+F}^{\tau}$ can be decomposes based on the assignment on $\vbl(c)$ as:
    \[
    \mu_{\+F}^{\tau}=\sum\limits_{\rho\in \+Q_{\vbl(c)}}\mu_{\+F}^{\tau}(\rho)\cdot \mu_{\+F}^{\tau}(\cdot \mid \rho).
    \]
    Then the coupling between $\mu^{\sigma}_{\+E}(\cdot \mid \neg c)$ and $\mu_{\+F}^{\tau}$ can be solved recursively on the new pair of sub-instances $(\+E,\sigma\land \pi)$ and $(\+F,\tau\land \rho)$, where $\pi=\vio{c}$ and $\rho\sim \mu^{\tau}_{\+F,\vbl(c)}$.
\end{itemize}
The mirrored case with $c\in\+E^{\sigma}\setminus\+F^{\tau}$ is processed symmetrically (Lines \ref{Line:couple-else-1}-\ref{Line:couple-assign-return-1}).

\Cref{Alg:couple} \emph{per se} does not give an efficient procedure for computing the output $(X,Y)$, 
since it relies on calculating some nontrivial marginal probabilities.
But it is useful for exhibiting the correlation decay property in the local lemma regime, 
which in turn is useful for the algorithmic implications.

\subsubsection*{{Correctness of the coupling.}}
%
%
Necessarily, the recursive procedure satisfies the invariant condition that the two joint distributions  $\mu_{\+E}^\sigma$ and $\mu_{\+F}^\tau$ are coupled correctly.


\begin{lemma}[soundness of coupling]\label{lemma:couple-well-defined}
Assume \Cref{condition:main-condition}.
For any constraint $c_0\in\+C$, the procedure $\Couple(\+C\setminus \{c_0\},\+C,\varnothing,\varnothing)$ terminates with probability 1 and returns a coupling of $\mu_{\+C}$ and $\mu_{\+C\setminus\{c_0\}}$. 
\end{lemma}
\begin{proof}
    We first prove that the coupling procedure is well-defined. 
    It suffices to show by structural induction in the top-down order of recursion that for each recursive call $\Couple(\+E,\+F,\sigma,\tau)$:
        \begin{equation}\label{eq:well-defined}
            \Lambda(\sigma)=\Lambda(\tau),\quad \P{\+E\land \sigma}>0, \quad \P{\+F\land \tau}>0.
        \end{equation} 
    The base case is $(\+E,\+F,\sigma,\tau)=(\+C\setminus \{c_0\},\+C,\varnothing,\varnothing)$. Then \eqref{eq:well-defined}  holds by \Cref{condition:main-condition} and \Cref{locallemma}.

    For the induction step, assume the current call is $\Couple(\+E,\+F,\sigma,\tau)$, then by the induction hypothesis, both $\+E^{\sigma}$ and $\+F^{\tau}$ are well-defined and $\Lambda(\sigma)=\Lambda(\tau)$. We then only prove the case when $\+F^{\tau}\not\subseteq \+E^{\sigma}$. The case when $\+F^{\tau}\subseteq \+E^{\sigma}$ follows analogously.

   Let $c$ be the smallest constraint in $\+F^{\tau}\setminus\+E^{\sigma}$. Note that by the induction hypothesis, $\mu_{\+E}^{\sigma}(c)$ is well-defined. Also, when $\mu_{\+E}^{\sigma}(c)>0$ we have
    \[
    \P{\left( \+E\cup \{c\}\right)\land \sigma}=\P{\+E\land \sigma}\cdot \mu_{\+E}^{\sigma}(c)>0,
    \]
    proving \eqref{eq:well-defined} for the direct recursive call $\Couple(\+E\cup \{c\},\+F,\sigma,\tau)$ on \Cref{Line:couple-psat-return}.
    
    For any direct recursive call $\Couple(\+E,\+F,\sigma\land \pi,\tau\land \rho)$ such that $\pi=\vio{c}$ and  $\rho\in \+Q_{\vbl(c)}$ on \Cref{Line:couple-assign-return}, $\Lambda(\sigma\land \pi)=\Lambda(\tau\land \rho)$ trivially holds. Also, it must hold that $\mu_{\+E}^{\sigma}(\neg c)>0, \mu_{\+F}^{\tau}(\rho)>0 $ by Lines \ref{Line:couple-else}-\Cref{Line:couple-sample} of \Cref{Alg:couple}, meaning $\pi$ implies $\neg c$. Therefore, we have
   \begin{align*}
    \P{\+E\land \sigma\land \pi}=\P{\+E\land \pi\land \neg c}=\mu_{\+E}^{\sigma}(\neg c)\cdot \P{\+E\land \sigma}>0,
   \end{align*}
   \[
   \P{\+F\land \tau\land \rho}=\mu_{\+F}^{\tau}(\rho)\cdot \P{\+F\land \tau}>0,
   \]
   finishing the proof of the induction step. This proves that the procedure is well-defined.

   Knowing the procedure is well-defined, notice that \Cref{Alg:couple} terminates when $\+E^{\sigma}=\+F^{\tau}$, and in each recursive step, either the size of $\+E^{\sigma}\triangle\+F^{\tau}$ reduces by one, or the number of unassigned variables in $\sigma$ and $\tau$ reduces by at least one. 
   Due to the finiteness of the number of (pinned) constraints and variables, the procedure $\Couple(\+C\setminus \{c_0\},\+C,\varnothing,\varnothing)$ terminates eventually.

   At last, we apply a structural induction in the bottom-up order of recursion, to prove that each possible recursive call of $\Couple(\+E,\+F,\sigma,\tau)$ produces a correct coupling of $\mu_{\+E}^{\sigma}$ and $\mu^{\tau}_{\+F}$.
   
   The base case is when $\+E^{\sigma}=\+F^{\tau}$, and the above claim holds by Lines \ref{Line:couple-return-cond}-\ref{Line:couple-return} of \Cref{Alg:couple}.

   For the induction step, by symmetry, we only prove for the case when $\+F^{\tau}\not\subseteq \+E^{\sigma}$. Let $c$ be the smallest constraint in $\+F^{\tau}\setminus\+E^{\sigma}$. From Lines \ref{Line:couple-choose}-\ref{Line:couple-assign-return} of \Cref{Alg:couple}, we have the following two cases:
\begin{itemize}
\item With probability $\mu_{\+E}^{\sigma}(c)$, $\Couple(\+E,\+F,\sigma,\tau)$ returns the output of $\Couple(\+E\cup \{c\},\+F,\sigma,\tau)$. By the induction hypothesis, the random pair of assignments $(X_1,Y_1)$ returned by $\Couple(\+E\cup \{c\},\+F,\sigma,\tau)$ at \Cref{Line:couple-psat-return} follow the marginal distributions
\[X_1\sim \mu_{\+E\cup \{c\}}^{\sigma}=\mu_{\+E}^{\sigma}(\cdot \mid c)\quad\text{ and }\quad Y_1\sim \mu^{\tau}_{\+F}.\]

\item With probability $\mu_{\+E}^{\sigma}(\neg c)$, $\Couple(\+E,\+F,\sigma,\tau)$ returns the output of $\Couple(\+E,\+F,\sigma\land \pi,\tau\land \rho)$ for $\pi=\vio{c}$ and $\rho\sim \mu^\tau_{\+F,\vbl(c)}$.  By the induction hypothesis, the random pair of assignments $(X_2,Y_2)$ returned by $\Couple(\+E,\+F,\sigma\land \pi,\tau\land \rho)$ at \Cref{Line:couple-assign-return} follow the marginal distributions:
\begin{align*}
X_2\sim  \mu_{\+E}^{\sigma}(\cdot \mid \pi)= \mu_{\+E}^{\sigma}(\cdot \mid \neg c) \quad\text{ and }\quad Y_2\sim\sum\limits_{\rho\in \+Q_{\vbl(c)}}\mu_{\+F}^{\tau}(\rho)\cdot \mu_{\+F}^{\tau}(\cdot \mid \rho)=\mu^{\tau}_{\+F}.
\end{align*}
\end{itemize}
Hence, the pair of assignments $(X,Y)$ returned at \Cref{Line:couple-assign-return} follow the marginal distributions:
\begin{equation*}
    \begin{aligned}
    X\sim & \mu_{\+E}^{\sigma}(c)\cdot \mu_{\+E}^{\sigma}(\cdot \mid c)+\mu_{\+E}^{\sigma}(\neg c)\cdot\mu_{\+E}^{\sigma}(\cdot \mid \neg c) &=\mu_{\+E}^{\sigma},\\
    Y\sim &\mu_{\+E}^{\sigma}(c)\cdot \mu^{\tau}_{\+F}+\mu_{\+E}^{\sigma}(\neg c)\cdot\mu^{\tau}_{\+F} &=\mu^{\tau}_{\+F}.
    \end{aligned}
\end{equation*}
This finishes the last case of the induction step and the proof of the lemma.
\end{proof}

\subsection{2-tree witness for discrepancy}
We need to bound the decay of correlation stated in \Cref{lemma:correlation-decay}. 
Throughout, we fix the input atomic CSP formula $\Phi=(V,\+Q,\+C)$ and an arbitrary constraint $c_0\subseteq \+C$ assumed in \Cref{lemma:correlation-decay}. 
Our goal is to bound the discrepancy of the coupling measured in Hamming distance between the random pair $(X,Y)$ produced by $\Couple(\+C\setminus \{c_0\},\+C,\varnothing,\varnothing)$.
To achieve this goal, we introduce a notion of witness for the discrepancy.
Inspired by the 2-tree structure~\cite{Alon91}, the witness is defined as a subset of independent constraints accessed by \Cref{Alg:couple}.

A technical aspect worth paying attention to is that \Cref{Alg:couple} deals with pinned constraints $c\in \+C^*$; 
however, independence between the original unpinned constraints $\unpin{c}\in \+C$ should be ensured to properly bound the discrepancy.
This is formalized through the following definition.

\begin{definition}[witness for discrepancy]\label{definition:bad-constraints}
Given a run of \Cref{Alg:couple} from 
$\Couple(\+C\setminus \{c_0\},\+C,\varnothing,\varnothing)$,
the \emph{witness set} $B=B_{c_0}\subseteq \+C$ is a set of (unpinned) constraints, constructed as follows:
\begin{itemize}
    \item Initially, $B=\emptyset$.
    \item Whenever $\Couple(\+E,\+F, \sigma\land \pi, \tau\land \rho)$ is recursively called at \Cref{Line:couple-assign-return} or \Cref{Line:couple-assign-return-1},
    add the original unpinned constraint $\unpin{c}\in \+C$ into $B$ if it does not share variables with any constraints already in $B$.
    Formally, let $B\gets B\Join c$, where the \emph{join} operator $B\Join c$ is defined as:
\begin{equation}\label{eq:definition-f}
    B\Join c\defeq \begin{cases} B & \exists c'\in B\text{ s.t. }\vbl\left(\unpin{c}\right)\cap \vbl(c')\neq \emptyset;\\
B\cup \left\{\unpin{c}\right\} & \text{otherwise}.\end{cases}
\end{equation}
\end{itemize}
%
\end{definition}

The set $B$ constructed in the definition above gives a witness for the discrepancy in the coupling produced by \Cref{Alg:couple}, 
in the sense which we will see later that 
the hamming distance $d_{\text{Ham}}(X,Y)$ is always upper bounded by $k\cdot (D+1)\cdot |B|$
for the output $(X,Y)$ of $\Couple(\+C\setminus \{c_0\},\+C,\varnothing,\varnothing)$.

Recall the definitions of dependency graph (\Cref{def:dependency-graph}) and 2-tree (\Cref{definition:2-tree}), we have the following observation.

\begin{lemma}\label{lemma:bad-constraint-2-tree}
During the process in \Cref{definition:bad-constraints}, 
the set $B$ is always a $2$-tree in the dependency graph $G_{\Phi}$.
\end{lemma}
\begin{proof}

Note that by \Cref{definition:bad-constraints} and \eqref{eq:definition-f}, we directly have the set $B$ always disjoint. It remains to show that $G^2_{\Phi}(B)$ is always connected. We then prove this claim by a forward induction.

The induction basis is when the current parameter for $\Couple$ is $(\+E,\+F,\sigma,\tau)=(\+C\setminus \{c_0\},\+C,\varnothing,\varnothing)$, and the claim holds by convention as we initialized $B$ as $\emptyset$.

 For the induction step, assume the current parameter $(\+E,\+F,\sigma,\tau)$ and the current witness set $B$. By the induction hypothesis, $G^2_{\Phi}(B)$ is connected. We then prove the hypothesis for all possible direct recursive calls. If $\+E^{\sigma}=\+F^{\tau}$, no further recursion will be called and we are already done. Otherwise $\+E^{\sigma}\neq\+F^{\tau}$, we then only prove the case when $\+F^{\tau}\not\subseteq \+E^{\sigma}$. The case when $\+F^{\tau}\subseteq \+E^{\sigma}$ follows analogously. Let $c$ be the smallest constraint in $\+F^{\tau}\setminus \+E^{\sigma}$, we have two cases:
\begin{itemize}
    \item We recurse to $\Couple(\+E\cup \{c\},\+F,\sigma,\tau)$ with $B$ remains unchanged, the claim trivially holds.
    \item We recurse to $\Couple(\+E,\+F,\sigma\land \pi,\tau\land \rho)$ for $\pi=\vio{c}$ and some $\rho\in \+Q_{\vbl(c)}$ and we update $B\gets B\Join c$. If $B\Join c=B$, then the claim trivially holds by the induction hypothesis, otherwise assume $B\Join c=B\cup \{\unpin{c}\}$, we then have two additional cases:
    \begin{itemize}
        \item $\unpin{c}$ is exactly $c_0$, in this case it must hold that $(\+E,\+F,\sigma,\tau)=(\+C\setminus \{c_0\},\+C,\varnothing,\varnothing)$, since otherwise after the first step $c_0$ is satisfied in both $\+E^{\sigma}$ and $\+F^{\tau}$, and any pinned constraint of it can never have been chosen by the algorithm. In this case, $B\Join c$ is simply $\{\unpin{c_0}\}$ and therefore the claim holds.
        \item Otherwise, it must satisfy that $\vbl(c)\neq\vbl(\unpin{c})$, since this is the only way that $c$ can be added into $\+E^{\sigma}\triangle\+F^{\tau}$. Take any variable $v\in \vbl(\unpin{c})\setminus \vbl(c)$, and let $c'\in \+C^*$ be the pinned constraint containing $v$ chosen by the procedure when $v$ is assigned. We must have $\unpin{(c')}\notin B$ since otherwise we have $B\Join \unpin{c}=B$ by \eqref{eq:definition-f}.  Hence, there must exist another constraint $c^*\in B$ such that $\vbl(c')\cap \vbl(c^*)\neq \emptyset$ by the time we update $B\gets B\Join c'$ according to \eqref{eq:definition-f}. Hence $c^*$ and $\unpin{c}$ are adjacent in $G^2_{\Phi}$ and $G^2_{\Phi}(B\Join c)=G^2_{\Phi}(B\cup \{\unpin{c}\})$ is also connected by the induction hypothesis.
    \end{itemize}
\end{itemize}
\end{proof}

%

To establish \Cref{lemma:correlation-decay}, it is intuitive to explore a truncated version of \Cref{Alg:couple}, limited to $|B|< K$ for an integer $K\ge 1$, and then bound the probability of truncation. 
This will be formalized by the random process introduced in \Cref{definition:procedure-root-to-leaf-path-1}. 
Here, we monitor the sequence of input arguments across all recursive calls of $\Couple$, 
along with the associated witness set $B$, truncating it once its size surpasses $K$. 
This formal approach elucidates the structure of the recursive execution of \Cref{Alg:couple}, thereby enhancing our understanding and paving the way for subsequent algorithmic insights.

\begin{definition}[random process simulating $K$-truncated \Cref{Alg:couple}]\label{definition:procedure-root-to-leaf-path-1}
Let  $\*X\sim \mu_{\+C\setminus \{c_0\}}$ and $\*Y\sim \mu_{\+C}$ be drawn independently beforehand. 
Define the random process $P^{\cp}=P_K^{\cp}=\left\{(\+E_t,\+F_t,\sigma_t,\tau_t,B_t)\right\}_{t\ge 0}$ starting from the initial state $(\+E_0,\+F_0,\sigma_0,\tau_0,B_0)=(\+C\setminus \{c_0\},\+C,\varnothing,\varnothing,\emptyset)$ as follows:

    \begin{enumerate}
        \item If $|B|= K$ or $\+E^{\sigma}=\+F^{\tau}$, the process stops and $(\+E_t,\+F_t,\sigma_t,\tau_t,B_t)$ is the outcome of the process.
       \item Otherwise, suppose $\+F_t^{\tau_t}\not\subseteq \+E_t^{\sigma_t}$. Let $c$ be the smallest pinned constraint in $\+F_t^{\tau_t}\setminus \+E_t^{\sigma_t}$.
       \[
        \left(\+E_{t+1},\+F_{t+1},\sigma_{t+1},\tau_{t+1},B_{t+1}\right)\gets
        \begin{cases}
        \left(\+E_{t}\cup \{c\},\+F_{t},\sigma_{t},\tau_{t},B_{t}\right)   & 
         \text{$c$ is satisfied by $\*X_{\vbl(c)}$};\\
             \left(\+E_{t},\+F_{t},\sigma_{t}\land \*X_{\vbl(c)},\tau_{t}\land \*Y_{\vbl(c)},B_{t}\Join c\right)  & \text{otherwise}.
        \end{cases}
        \]
        \item Otherwise $\+F_{t}^{\tau}\subseteq \+E_{t}^{\sigma}$. Let $c$ be the smallest pinned constraint in $\+E_{t}^{\sigma}\setminus \+F_{t}^{\tau}$. 
        \[
       \left(\+E_{t+1},\+F_{t+1},\sigma_{t+1},\tau_{t+1},B_{t+1}\right)\gets
        \begin{cases}
           \left(\+E_{t},\+F_{t}\cup \{c\},\sigma_{t},\tau_{t},B_{t}\right)   & 
         \text{$c$ is satisfied by $\*Y_{\vbl(c)}$};\\
         \left(\+E_{t},\+F_{t},\sigma_{t}\land \*X_{\vbl(c)},\tau_{t}\land \*Y_{\vbl(c)},B_{t}\Join c\right)   & \text{otherwise}.
        \end{cases}
        \]
    \end{enumerate}
Let $\mu^{\cp}=\mu_K^{\cp}$ denote the distribution of the outcome $(\+E_\infty,\+F_\infty,\sigma_\infty,\tau_\infty,B_\infty)$
of this process,
%
and let $\+L^{\cp}=\supp(\mu^{\cp})$ be its support.
Let $\+L^{\cp}_K$ be the set of ``truncated'' outcomes, i.e.:
\[
\+L^{\cp}_K=\left\{(\+E,\+F,\sigma,\tau,B)\mid \mu^{\cp}((\+E,\+F,\sigma,\tau,B))>0\land |B|= K\right\}.
\]
At last, let $\+V^{\cp}$ denote the set of all possible $(\+E,\+F,\sigma,\tau,B)$ with $\Pr{(\+E,\+F,\sigma,\tau,B)\in P^{\cp}}>0$.

\end{definition}

\begin{remark}
Note that for atomic CSP $\Phi$ satisfying \Cref{condition:main-condition}, due to \Cref{locallemma} and total probability,
the random process $P^{\cp}$ and the distribution $\mu^{\cp}$ in \Cref{definition:procedure-root-to-leaf-path-1}  are well-defined.
\end{remark}

\begin{remark}[explicitly identified randomness]\label{remark:explicit-randomness}
In \Cref{definition:procedure-root-to-leaf-path-1}, all randomness used by the process $P^\cp$ is identified with the two pre-generated random assignments $\*X\sim \mu_{\+C\setminus {c_0}}$ and $\*Y\sim \mu_{\+C}$. However, in \Cref{Alg:couple}, random choices (in \Cref{Line:couple-psat-cond,Line:couple-sample,Line:couple-psat-cond-1,Line:couple-sample-1}) are generated at the moment they are used. According to the principle of deferred decision, this distinction does not affect the definition of $P^\cp$, and the process faithfully simulates \Cref{Alg:couple}. Nevertheless, such explicit identification of randomness is crucial for our analyses of the coupling and its algorithmic implications.
\end{remark}

A key observation is that the correlation decay of the coupling is bounded by the probability of truncated outcomes of the process constructed in \Cref{definition:procedure-root-to-leaf-path-1}.


\begin{lemma}\label{lemma:hamming-transform}
Assume \Cref{condition:main-condition}. For the output $(X,Y)$ of $\Couple(\+C\setminus \{c_0\},\+C,\varnothing,\varnothing)$,
\[
\Pr{d_{\mathrm{Ham}}(X,Y)\geq k\cdot(D+1)\cdot K}\leq  \mu^{\cp}\left[\+L^{\cp}_{K}\right].
\]
\end{lemma}
\begin{proof}
Note that by \Cref{definition:bad-constraints},  each time we recursively call $\Couple(\+E,\+F,\sigma\land \pi, \tau\land \rho)$ at \Cref{Line:couple-assign-return} or \Cref{Line:couple-assign-return-1} of \Cref{Alg:couple}, we update $B\gets B\Join c$, at most $k$ variables is additionally assigned to $\sigma/\tau$ and any pinned constraint of $\unpin{c}$ will never been chosen by the procedure in future steps. Note that each constraint in $\+C$ shares constraint with at most $D$ other constraints, hence by \eqref{eq:definition-f}, after $D+1$ such operations the size of $B$ increases by at least one.  Also, by \Cref{Line:couple-return} of \Cref{Alg:couple} we have $d_{\mathrm{Ham}}(X,Y)$ is upper bounded by the total number of variables assigned in $\sigma/\tau$ during $\Couple(\+C\setminus \{c_0\},\+C,\varnothing,\varnothing)$. Therefore, we have
\[
d_{\mathrm{Ham}}(X,Y)\geq k\cdot (D+1)\cdot K\implies |B|\geq K.
\]
We claim that the process in \Cref{definition:procedure-root-to-leaf-path-1} can be coupled with \Cref{Alg:couple} initialized by $\Couple(\+C\setminus \{c_0\},\+C,\varnothing,\varnothing)$ so that each time we recursively call $\Couple(\+E,\+F,\sigma,\tau)$ with witness set $B$, we will also move to the same tuple $(\+E,\+F,\sigma,\tau,B)$ in \Cref{definition:procedure-root-to-leaf-path-1}. Under the coupling assumed by the claim, it follows that
\[
|B|\geq K\implies (\+E^{\cp},\+F^{\cp},\sigma^{\cp},\tau^{\cp},B^{\cp})\in \+L^{\cp}_K,
\]
and the lemma immediately follows.

We strengthen the claim that conditioning on we are currently at some tuple $(\+E,\+F,\sigma,\tau,B)$ in \Cref{definition:procedure-root-to-leaf-path-1}, it follows that
\begin{equation}\label{eq:initial-assignment-distribution}
\*X\sim \mu^{\sigma}_{\+E},\quad \*Y\sim \mu^{\tau}_{\+F}.
\end{equation}
and prove the strengthened claim using a structural induction in the top-down order. The base case is for the initial call $\Couple(\+C\setminus \{c_0\},\+C,\varnothing,\varnothing)$ with witness set $B=\emptyset$. Note that in \Cref{definition:procedure-root-to-leaf-path-1} we also initially have $(\+E,\+F,\sigma,\tau,B)=(\+C\setminus \{c_0\},\+C,\varnothing,\varnothing,\emptyset)$, and 
\[
\*X\sim \mu_{\+C\setminus \{c_0\}}=\mu^{\sigma}_{\+E},\quad \*Y\sim \mu_{\+C}=\mu^{\tau}_{\+F}.
\]

For the induction step, assume we are currently executing $\Couple(\+C\setminus \{c_0\},\+C,\varnothing,\varnothing)$ with witness set $B$. By the induction hypothesis, we are also at the same tuple $(\+E,\+F,\sigma,\tau,B)$ in \Cref{definition:procedure-root-to-leaf-path-1}, so that \eqref{eq:initial-assignment-distribution} holds.

Now it remains to compare \Cref{Alg:couple} with \Cref{definition:procedure-root-to-leaf-path-1} and notice that under \eqref{eq:initial-assignment-distribution}, the transition probabilities to the next tuples are the same in both. Therefore, the next tuple in both can still be perfectly coupled. Also, it can be verified that \eqref{eq:initial-assignment-distribution} also holds for each possible branch in \Cref{definition:procedure-root-to-leaf-path-1}. This finishes the proof of the claim and the lemma.
\end{proof}

The following tail bound for the measure of truncated outcomes is crucial for proving \Cref{lemma:correlation-decay}.
\begin{lemma}\label{lemma:process-tail-bound}
Assume \Cref{condition:main-condition}. For any integer $K\geq 1$,
\[
 \mu^{\cp}\left[\+L^{\cp}_{K}\right]\leq 2^{-K}.
\]
\end{lemma}

\Cref{lemma:correlation-decay} then follows directly from combining \Cref{lemma:couple-well-defined,lemma:hamming-transform,lemma:process-tail-bound}.
All we need to do now is prove \Cref{lemma:process-tail-bound}.

\subsection{Refutation of large witnesses}
By the above arguments, it suffices to bound the probability $\mu^{\cp}\left[\+L^{\cp}_K\right]$.
Consider the random outcome generated by the process in \Cref{definition:procedure-root-to-leaf-path-1}:
\[
\left(\+E^{\cp},\+F^{\cp},\sigma^{\cp},\tau^{\cp},B^{\cp}\right)\sim\mu^{\cp}.
\]
It then suffices to show that the witness set $B^{\cp}$ is very unlikely to become too large.
We first establish the following tail bound for any particular large enough subset of $B^{\cp}$.


\begin{lemma}\label{lemma:disjoint-prob}
Assume \Cref{condition:main-condition}. For any subset of disjoint constraints $T\subseteq \+C$,
\[
\Pr{T\subseteq B^{\cp}}\leq p^{\frac{2|T|}{2+\zeta}}\cdot (1-\mathrm{e}p)^{2(D+1)|T|}.
\]
\end{lemma}

\begin{proof}

 We claim that for each $c\in T$, the event $c\in B^{\cp}$ implies the following event 
 \begin{equation}\label{eq:definition-ec}
\+A_c: \forall v\in \vbl(c), \quad\*X(v)=\vio{c}(v) \lor \*Y(v)=\vio{c}(v).
\end{equation}

To prove the claim, by contradiction we suppose that there exists some $v\in \vbl(c)$ such that both 
$\*X(v)\neq \vio{c}(v)$ and $\*Y(v)\neq \vio{c}(v)$. Let $c'$ be the pinned constraint chosen by the procedure when $c$ is added into $B^{\cp}$ at some step in \Cref{definition:procedure-root-to-leaf-path-1}, then we have $c=\unpin{(c')}$. It must hold that $v\in \Lambda(c')$, as otherwise $c'$ is both satisfied in $\+E$ and $\+F$ conditioning on the assignment of $\*X$ and $\*Y$ on $\vbl(c)\setminus \Lambda(c')$, and could not have been chosen by the procedure. However, $v\in \Lambda(c')$ is also not possible as otherwise by the time $c'$ is chosen in \Cref{definition:procedure-root-to-leaf-path-1}, $c'$ will be satisfied by both $\*X_{\vbl(c')}$ and $\*Y_{\vbl(c')}$ and will not be added into $B^{\cp}$, a contradiction. The above claim is proved.
 
 Thus we have
 \begin{equation}\label{eq:disjoint-prob-1}
\begin{aligned}
\Pr{T\subseteq B^{\cp}}\leq & \Pr{\bigwedge\limits_{c\in T}\+A_c}\\
(\text{by disjointness between $c\in T$})\quad\leq &\prod\limits_{c\in T}\Pr{\+A_c}\\
\leq & \prod\limits_{c\in T}\left(\prod\limits_{v\in \vbl(c)}\left(\frac{2}{|Q_v|}-\frac{1}{|Q_v|^2}\right)\cdot (1-\mathrm{e}p)^{-2(D+1)}\right).
\end{aligned}
 \end{equation}

 Here, the last inequality is by interpreting the probability space for generating $\*X\sim \mu_{\+C}$ and $\*Y\sim \mu_{\+C\setminus \{c_0\}}$ as the product space over two copies of the distribution $\+P$, conditioning on that all constraints in $\+C$ are satisfied in the first copy, and all constraints in $\+C\setminus \{c_0\}$ are satisfied in the second copy. Note that this can be viewed as an LLL distribution with dependency degree at most $D$ and violation probability of each bad event at most $p$. Also, each event $\+A_c$ is mutually dependent with all but $2(D+1)$ bad events. Therefore, setting $x(c)=\mathrm{e}p$ for each bad event $c$ and applying \Cref{HSS} gives to the above last inequality. 

Note that since $\+C$ is a set of atomic constraints, we have for each $c\in \+C$,
\[
    \prod\limits_{v\in \vbl(c)}\frac{1}{|Q_v|}\leq p, 
\]
and for any $v\in \vbl(c)$,
\[
    \frac{\ln \left(2/|Q_v|-1/|Q_v|^2\right)}{\ln \left(1/|Q_v|\right)}=2-\frac{\ln(2/|Q_v|-1)}{\ln(1/|Q_v|)}\geq 2-\frac{\ln(2/q_{\min}-1)}{\ln q_{\min}}=\frac{2}{2+\zeta},
\]
where in above we use that $\frac{\ln(2x-1)}{\ln x}$ is a monotonically decreasing function for all $x>1$.

Hence, combining with \eqref{eq:disjoint-prob-1} we have
\[
\Pr{T\subseteq B^{\cp}}\leq \prod\limits_{c\in T}\left(p^{\frac{2}{2+\zeta}}\cdot (1-\mathrm{e}p)^{-2(D+1)}\right)=p^{\frac{2|T|}{2+\zeta}}\cdot (1-\mathrm{e}p)^{-2(D+1)|T|}. \qedhere
\]
\end{proof}

With the help of \Cref{lemma:disjoint-prob}, we are ready to prove \Cref{lemma:process-tail-bound}, completing the proof of \Cref{lemma:correlation-decay}.
\begin{proof}[Proof of \Cref{lemma:process-tail-bound}]
Let $\mathbb{T}_K^{c_0}$ be the set of $2$-trees in $G_{\Phi}$ of size $K$ containing $c_0$. Since $K\geq 1$ we have all $(\+E^{\cp},\+F^{\cp},\sigma^{\cp},\tau^{\cp},B^{\cp})\in \+V_{K}^{\cp}$ satisfy that $c_0\in B^{\cp}$ by \Cref{condition:main-condition} and \Cref{definition:procedure-root-to-leaf-path-1}. Then we have
\begin{align*}
&\mu^{\cp}[\+L^{\cp}_{K}]\\
(\text{by \Cref{definition:bad-constraints}})\quad\leq &\Pr{|B^{\cp}|\geq K}\\
(\text{by \Cref{lemma:bad-constraint-2-tree}})\quad\leq &\sum\limits_{T\in \mathbb{T}_K^{c_0}}\Pr{T\subseteq B^{\cp}}\\
(\text{by \Cref{lemma:disjoint-prob}})\quad\leq &\sum\limits_{T\in \mathbb{T}_K^{c_0}} p^{\frac{2K}{2+\zeta}}\cdot (1-\mathrm{e}p)^{-2(D+1)K}\\
(\text{by \Cref{lemma:num-2-tree}})\quad\leq &\frac{(\mathrm{e}D^2)^{K-1}}{2}\cdot p^{\frac{2L}{2+\zeta}}\cdot (1-\mathrm{e}p)^{-2(D+1)K}\\
\leq & \left(\mathrm{e}D^2\cdot p^\frac{2}{2+\zeta}\cdot (1-\mathrm{e}p)^{-2(D+1)}\right)^{K}\\
(\text{by \Cref{condition:main-condition}})\quad \leq &2^{-K}.
\end{align*}

\end{proof}

\section{Linear programming from constraint-wise coupling}\label{sec:algorithms}

In this section, we will present a linear program that is the central component for the counting and sampling algorithms within the local lemma regime, as stated in \Cref{theorem:main-counting} and \Cref{theorem:main-sampling}.


As in previous works for LP-based counting LLL~\cite{Moi19,guo2019counting,Vishesh21towards}, we will set up a linear program to mimic the transition probabilities in the coupling procedure (\Cref{definition:procedure-root-to-leaf-path-1}) with the goal of bootstrapping the constraint-wise marginal ratios. However, our construction of the coupling differs vastly from previous works. Additionally, our analysis of the coupling is near-critical in reaching the optimal threshold, making the design of the linear program challenging and requiring significant novelty. Nevertheless, we have managed to distill critical structures from the analysis in \Cref{sec:correlation-decay} that are sufficient for constructing the LP.

\subsection{Marginal probabilities from coupling procedure}


Fix an atomic CSP formula $\Phi=(V,\+Q,\+C)$, an arbitrary constraint $c_0\in \+C$, and an integer $K\ge 1$. 
Recall the random process $P^{\cp}=P^{\cp}_K=\{(\+E_t,\+F_t,\sigma_t,\tau_t,B_t)\}_{t\ge 0}$  and the distribution $\mu^{\cp}=\mu^{\cp}_K$ of its final outcome, both introduced in \Cref{definition:procedure-root-to-leaf-path-1}.
The following defines a family of probabilities arising from a natural one-sided sampling process induced from $P^{\cp}$.
These probabilities correspond to the variables of the linear program that will be introduced later.

\begin{definition}[one-sided sampler and marginal probabilities]\label{definition:imaginary-sampler-quantities}
Let ${X}$ (resp.~${Y}$) be drawn as:
\begin{itemize}
\item draw $(\bm{\+E},\bm{\+F},\bm{\sigma},\bm{\tau},\bm{B})\sim\mu^{\cp}$;
\item draw ${X}\sim \mu^{\bm{\sigma}}_{\bm{\+E}}$ (and resp. ${Y}\sim \mu^{\bm{\tau}}_{\bm{\+F}}$).
\end{itemize}
For any evaluation of $(\+E,\+F,\sigma,\tau,B)$, any $\bm{x}\in \+Q^{\+E\land\sigma}\triangleq\{\pi\in\+Q\mid \pi\text{ satisfies }\+E\land\sigma\}$ and $\bm{y}\in \+Q^{\+F\land\tau}$, define:
      \[p_{(\+E,\+F,\sigma,\tau,B)}^X\triangleq\Pr{(\+E,\+F,\sigma,\tau,B)\in P^{\cp}\mid {X}=\bm{x}},
      \]
      \[p_{(\+E,\+F,\sigma,\tau,B)}^Y\triangleq\Pr{(\+E,\+F,\sigma,\tau,B)\in P^{\cp}\mid {Y}=\bm{y}}.
    \]
\end{definition}

Note that in above definition, $p_{(\+E,\+F,\sigma,\tau,B)}^X=\Pr{(\+E,\+F,\sigma,\tau,B)\in P^{\cp}\mid {X}=\bm{x}}$ for \emph{any} $\bm{x}\in \+Q^{\+E\land\sigma}$ (and similarly $p_{(\+E,\+F,\sigma,\tau,B)}^Y=\Pr{(\+E,\+F,\sigma,\tau,B)\in P^{\cp}\mid {Y}=\bm{y}}$ for any $\bm{y}\in \+Q^{\+F\land\tau}$). This is actually well-defined. 
And furthermore, these probabilities are in fact marginal probabilities of some partial assignments in some joint distributions.
These are formally justified by the following proposition.

\begin{proposition}\label{lemma:quantities-well-defined}
Assume \Cref{condition:main-condition}. Fix any evaluation of $(\+E,\+F,\sigma,\tau,B)$. It holds that
\begin{align*}
\forall \bm{x},\bm{x}'\in \+Q^{\+E\land\sigma}:&&  \Pr{(\+E,\+F,\sigma,\tau,B)\in P^\cp\mid {X}=\bm{x}}&=\Pr{(\+E,\+F,\sigma,\tau,B)\in P^\cp\mid {X}=\bm{x}'},\\
\forall \bm{y},\bm{y}'\in \+Q^{\+F\land\tau}:&&  \Pr{(\+E,\+F,\sigma,\tau,B)\in P^\cp\mid {Y}=\bm{y}}&=\Pr{(\+E,\+F,\sigma,\tau,B)\in P^\cp\mid {Y}=\bm{y}'}.
\end{align*}
Moreover, it holds that
\begin{equation}\label{eq:real-quantities}
p_{(\+E,\+F,\sigma,\tau,B)}^{X}=\mu_{\+C}(\+F\land \tau) \quad\text{ and } \quad p_{(\+E,\+F,\sigma,\tau,B)}^{Y}=\mu_{\+C\setminus \{c_0\}}(\+E\land \sigma).
 \end{equation}
\end{proposition}

\begin{proof}
Following the same routine as the proof of \Cref{lemma:hamming-transform} (more precisely, following \eqref{eq:initial-assignment-distribution}), we can couple the one-sided sampler in \Cref{definition:imaginary-sampler-quantities} with the process $P^\cp$ so that the sample $X$ (resp. $Y$) generated by the one-sided sampler is identified with the random source $\*X\sim \mu_{\+C\setminus \{c_0\}}$ (resp. $\*Y\sim \mu_{\+C}$) used in the construction of $P^\cp$ in \Cref{definition:procedure-root-to-leaf-path-1}. 

We then only prove the first equality in \eqref{eq:real-quantities}, and the other follows analogously. Fix some $\bm{x}\in \+Q$, then under the coupling between $X$ and $\*X$ it immediately follows that
\[
\Pr{X=\bm{x}}=\mu_{\+C\setminus \{c_0\}}(\bm{x}).
\]
We claim that for any $(\+E,\+F,\sigma,\tau,B)\in \+V^{\cp}$,
\begin{equation}\label{eq:quantities-well-defined-claim}
    \Pr{(\+E,\+F,\sigma,\tau,B)\in P^{\cp}\land X= \bm{x}}=\begin{cases}0 & \bm{x}\notin \+Q^{\+E\land \sigma};\\ \mu_{\+C\setminus \{c_0\}}(\bm{x})\cdot \mu_{\+C}(\+F\land \tau)& \bm{x}\in \+Q^{\+E\land \sigma},\end{cases}
\end{equation}
and the lemma follows from the law of conditional probability.

We then prove \eqref{eq:quantities-well-defined-claim}. Note that the first case is immediate by \Cref{definition:imaginary-sampler-quantities}. 

Following the proof of  \Cref{lemma:hamming-transform}, we have that for each $(\+E,\+F,\sigma,\tau,B)\in \+V^{\cp}$,
\begin{equation}\label{eq:quantities-well-defined-property}
\Pr{(\+E,\+F,\sigma,\tau,B)\in P^{\cp}}=\mu_{\+C\setminus \{c_0\}}(\+E\land \sigma)\cdot \mu_{\+C}(\+F\land \tau).
\end{equation}

Therefore, for any $(\+E,\+F,\sigma,\tau,B)\in \+V^{\cp}$ and any $\bm{x}\in \+Q^{\+E\land \sigma}$,
\begin{align*}
&\Pr{(\+E,\+F,\sigma,\tau,B)\in P^{\cp}\land X=\bm{x}}\\
=& \Pr{(\+E,\+F,\sigma,\tau,B)\in P^{\cp}}\cdot \Pr{X=\bm{x}\mid (\+E,\+F,\sigma,\tau,B)\in P^{\cp}}\\
(\text{by \Cref{definition:imaginary-sampler-quantities}})\quad =& \Pr{(\+E,\+F,\sigma,\tau,B)\in P^{\cp}}\cdot \mu^{\sigma}_{\+E}(\bm{x})\\
(\text{by \eqref{eq:quantities-well-defined-property}})\quad=&\mu_{\+C\setminus \{c_0\}}(\+E\land \sigma)\cdot \mu_{\+C}(\+F\land \tau)\cdot \mu^{\sigma}_{\+E}(\bm{x})\\
(\text{by the chain rule})\quad=&\mu_{\+C\setminus \{c_0\}}(\bm{x})\cdot \mu_{\+C}(\+F\land \tau),
\end{align*}
finishing the proof of \eqref{eq:quantities-well-defined-claim}. Here, the last equality is additionally by that $\+E\land \sigma$ implies $\+C\setminus \{c_0\}$ for each $(\+E,\+F,\sigma,\tau,B)\in \+V^{\cp}$, an argument easily proven through an induction on \Cref{definition:procedure-root-to-leaf-path-1}.
\end{proof}

The following lists some basic properties of  $p_{(\+E,\+F,\sigma,\tau,B)}^X$ and $p_{(\+E,\+F,\sigma,\tau,B)}^Y$.

\begin{proposition}\label{lemma:quantities-properties}
    Assume \Cref{condition:main-condition}. The followings hold for the $p_{(\+E,\+F,\sigma,\tau,B)}^X$ and $p_{(\+E,\+F,\sigma,\tau,B)}^Y$:
\begin{enumerate}
    \item $p_{(\+E,\+F,\sigma,\tau,B)}^X,p_{(\+E,\+F,\sigma,\tau,B)}^Y\in [0,1]$. \label{item:quantities-properties-1}
    In particular, $p_{(\+C\setminus \{c_0\},\+C,\varnothing,\varnothing,\emptyset)}^X=p_{(\+C\setminus \{c_0\},\+C,\varnothing,\varnothing,\emptyset)}^Y=1$.
    \item For any $(\+E,\+F,\sigma,\tau,B) \in \+V^{\cp}\setminus \+L^{\cp}$, where $\+V^{\cp}$ and $\+L^\cp$ are defined in \Cref{definition:procedure-root-to-leaf-path-1},  \label{item:quantities-properties-2}
    \begin{enumerate}
        \item if $\+F^{\tau}\not\subseteq \+E^{\sigma}$, letting $c\in\+F^{\tau}\setminus \+E^{\sigma}$ be the smallest and $\pi=\vio{c}$,
        \begin{align*}
            p_{(\+E,\+F,\sigma,\tau,B)}^X =
            &p^X_{(\+E\cup \{c\},\+F,\sigma,\tau,B)}
            =
            \sum\limits_{\substack{\rho\in \+Q_{\vbl(c)}\\ (\+E,\+F,\sigma\land \pi,\tau\land \rho,B\Join c)\in \+V^{\cp}}}p^X_{(\+E,\+F,\sigma\land \pi,\tau\land \rho,B\Join c)};\\
             p_{(\+E,\+F,\sigma,\tau,B)}^Y =
             &p^Y_{(\+E\cup \{c\},\+F,\sigma,\tau,B)}+p^Y_{(\+E,\+F,\sigma\land \pi,\tau\land \rho,B\Join c)},\\
             &\text{for all }\rho\in \+Q_{\vbl(c)}\text{ and all }(\+E,\+F,\sigma\land \pi,\tau\land \rho,B\Join c)\in \+V^{\cp}.
        \end{align*}
       \label{item:quantities-properties-3-a}
        \item otherwise $\+F^{\tau}\subseteq \+E^{\sigma}$, letting $c\in\+E^{\sigma}\setminus \+F^{\tau}$ be the smallest and  $\rho=\vio{c}$,
        \begin{align*}
        p_{(\+E,\+F,\sigma,\tau,B)}^X
=
&p^X_{(\+E,\+F\cup \{c\},\sigma,\tau,B)}+p^X_{(\+E,\+F,\sigma\land \pi,\tau\land \rho,B\Join c)},\\
&\text{for all }\pi\in \+Q_{\vbl(c)}\text{ and all }(\+E,\+F,\sigma\land \pi,\tau\land \rho,B\Join c)\in \+V^{\cp};\\
p_{(\+E,\+F,\sigma,\tau,B)}^Y
            =
            &p^Y_{(\+E,\+F\cup \{c\},\sigma,\tau,B)}
            =
            \sum\limits_{\substack{\pi\in \+Q_{\vbl(c)}\\(\+E,\+F,\sigma\land \pi,\tau\land \rho,B\Join c)\in \+V^{\cp}}}p^Y_{(\+E,\+F,\sigma\land \pi,\tau\land \rho,B\Join c)}.
        \end{align*}\label{item:quantities-properties-3-b}
    \end{enumerate}
        \item For any $(\+E,\+F,\sigma,\tau,B)\in \+V^{\cp}$,
        \[
        p_{(\+E,\+F,\sigma,\tau,B)}^X\cdot \frac{ |\+Q^{\+E\land \sigma}|}{|\+Q^{\+C\setminus \{c_0\}}|}= p_{(\+E,\+F,\sigma,\tau,B)}^Y\cdot \frac{|\+Q^{\+F\land \tau}|}{|\+Q^{\+C}|}.
        \] \label{item:quantities-properties-4}
\end{enumerate}
\end{proposition}
These properties follow directly from \Cref{definition:imaginary-sampler-quantities} or from \Cref{lemma:quantities-well-defined}, and are easy to verify.

The following bounds for $p_{(\+E,\+F,\sigma,\tau,B)}^X$ and $p_{(\+E,\+F,\sigma,\tau,B)}^Y$ with 2-trees
follow from a similar argument as in our analysis of the correlation decay property (e.g.~the proof of \Cref{lemma:disjoint-prob}).

\begin{proposition}\label{lemma:quantities-property-bad-event}
     Assume \Cref{condition:main-condition}. Let $\mathbb{T}_K^{c_0}$ be the set of $2$-trees in $G_{\Phi}$ of size $K$ containing $c_0$. 
 Recall that $\+P$ is the product distribution.
 For each $c\in \+C$, let $\+A_c$ be the event defined as \eqref{eq:definition-ec}, that is,
       \[\+A_c: \forall v\in \vbl(c), \quad\*X(v)=\vio{c}(v) \lor \*Y(v)=\vio{c}(v) .
\]
Then for any $T\in \mathbb{T}_K^{c_0}$ and $\bm{x}\in \+Q^{\+C\setminus \{c_0\}}$,
\begin{align*}
         \sum\limits_{\substack{(\+E,\+F,\sigma,\tau,B)\in \+L^{\cp}:\\ B=T\land \bm{x}\in \+Q^{\+E\land \sigma}}} p^X_{(\+E,\+F,\sigma,\tau,B)}
         &\leq (1-\mathrm{e}p)^{(D+1)K}\Pr[\*Y\sim \+P]{\bigwedge\limits_{c\in T}\+A_c\mid \*X=\bm{x}},
\end{align*}
and for any $T\in \mathbb{T}_K^{c_0}$ and $\bm{y}\in \+Q^{\+C}$,
\begin{align*}
     \sum\limits_{\substack{(\+E,\+F,\sigma,\tau,B)\in \+L^{\cp}:\\ B=T\land \bm{y}\in \+Q^{\+F\land \tau}}} p^Y_{(\+E,\+F,\sigma,\tau,B)}
     &\leq (1-\mathrm{e}p)^{(D+1)K}\Pr[\*X\sim \+P]{\bigwedge\limits_{c\in T}\+A_c\mid \*Y=\bm{y}}.
\end{align*}
\end{proposition}

\begin{proof}
We only prove the first inequality, and the second one follows analogously.

Let $\left(\+E^{\cp},\+F^{\cp},\sigma^{\cp},\tau^{\cp},B^{\cp}\right)\sim\mu^{\cp}$ denote the random outcome of the process $P^\cp$ constructed in \Cref{definition:procedure-root-to-leaf-path-1}.
Following the same routine as the proof of \Cref{lemma:hamming-transform} (more precisely, following \eqref{eq:initial-assignment-distribution}), we can couple the one-sided sampler in \Cref{definition:imaginary-sampler-quantities} with the process $P^\cp$ so that the sample $X$ (resp. $Y$) generated by the one-sided sampler is identified with the random source $\*X\sim \mu_{\+C\setminus \{c_0\}}$ (resp. $\*Y\sim \mu_{\+C}$) used in the construction of $P^\cp$ in \Cref{definition:procedure-root-to-leaf-path-1}. 
Fix any $\bm{x}\in \+Q^{\+C\setminus \{c_0\}}$. Then we have
\begin{align*}
 \sum\limits_{\substack{(\+E,\+F,\sigma,\tau,B)\in \+L^{\cp}:\\ B=T\land \bm{x}\in \+Q^{\+E\land \sigma}}} p^X_{(\+E,\+F,\sigma,\tau,B)}
 =&\sum\limits_{\substack{(\+E,\+F,\sigma,\tau,B)\in \+L^{\cp}:\\ B=T\land \bm{x}\in \+Q^{\+E\land \sigma}}} \Pr{(\+E,\+F,\sigma,\tau,B)\in P^{\cp}\mid {X}=\bm{x}}\\
  =&\Pr{B^{\cp}=T\mid {X}=\bm{x}}\\
 (\text{by the coupling between $X$ and $\*X$})\qquad =& \Pr[\substack{\*X\sim\mu_{\+C\setminus \{c_0\}}\\\*Y\sim \mu_{\+C}}]{B^{\cp}=T\mid \*X=\bm{x}}\\
 \leq & \Pr[\substack{\*X\sim\mu_{\+C\setminus \{c_0\}}\\\*Y\sim \mu_{\+C}}]{\bigwedge\limits_{c\in T}\+A_c\mid \*X=\bm{x}}\\
(\text{by \Cref{HSS}})\qquad \leq & (1-\mathrm{e}p)^{(D+1)K}\Pr[\*Y\sim \+P]{\bigwedge\limits_{c\in T}\+A_c\mid \*X=\bm{x}}.
\end{align*}
Here, the first inequality follows the same argument in the proof of \Cref{lemma:disjoint-prob} that $B^{\cp}=T$ implies the event $\bigwedge\limits_{c\in T}\+A_c$.
\end{proof}

\Cref{lemma:quantities-property-bad-event} utilizes the independence between the random assignments $\*X,\*Y$,
which are generated as the random source used by the process $P^\cp$ constructed in \Cref{definition:procedure-root-to-leaf-path-1} for simulating \Cref{Alg:couple}.  
Such explicit identification of randomness helped in proving \Cref{lemma:disjoint-prob} and also will play a key role in constructing the linear program that will be introduced later.




\subsection{Setting up the linear program}





Next,  we  set up the linear program that mimicries the marginal probabilities $p_{(\+E,\+F,\sigma,\tau,B)}^X$ and $p_{(\+E,\+F,\sigma,\tau,B)}^Y$  introduced in \Cref{definition:imaginary-sampler-quantities}, in order to bootstrap the ratio between the numbers of assignments satisfying $\+C$ and those satisfying $\+C\setminus \{c_0\}$. 

\subsubsection{The coupling tree}
We define the following notion of the recursion tree for the coupling procedure $\Couple(\+C\setminus \{c_0\},\+C,\varnothing,\varnothing)$, truncated whenever the witness set $B$ has size $|B|\geq K$.


\begin{definition}[$K$-truncated coupling tree]\label{definition:coupling-tree}
The \emph{$K$-truncated coupling tree} $\+T=\+T_K(\Phi,c_0)$ is a finite rooted tree, 
with each tree node corresponding to a $(\+E,\+F,\sigma,\tau,B)$, which is the tuple appeared in \Cref{definition:procedure-root-to-leaf-path-1}.
The tree $\+T$ is inductively constructed as follows:
    \begin{enumerate}
    \item The root of $\+T$ corresponds to $(\+C\setminus \{c_0\},\+C,\varnothing,\varnothing,\emptyset)$, and has depth $0$.
    \label{item:coupling-tree-1}
    \item 
    For $i=0,1,\ldots$: for all existing tree nodes $(\+E,\+F,\sigma,\tau,B)\in V(\+T)$ of depth $i$ in the current $\+T$,\label{item:coupling-tree-2}
    \begin{enumerate}
        \item if $\sigma$ violates $\+E$ or $\tau$ violates $\+F$ or $\+E^{\sigma}=\+F^{\tau}$ or $|B|=K$, then do nothing and hence leave $(\+E,\+F,\sigma,\tau,B)$ as a leaf node in $\+T$; \label{item:coupling-tree-2-a}
        \item else if $\+F^{\tau}\not\subseteq \+E^{\sigma}$, then pick the smallest $c\in\+F^{\tau}\setminus \+E^{\sigma}$, add $(\+E\cup \{c\},\+F,\sigma,\tau,B)$ as a child of $(\+E,\+F,\sigma,\tau,B)$ in $\+T$, and for $\pi=\vio{c}$ and each $\rho\in \+Q_{\vbl(c)}$, add $(\+E,\+F,\sigma\land \pi,\tau\land \rho,B\Join c)$ as a child of $(\+E,\+F,\sigma,\tau,B)$ in $\+T$;\label{item:coupling-tree-2-b}
        \item else, pick the smallest $c\in\+E^{\sigma}\setminus \+F^{\tau}$, add $(\+E,\+F\cup \{c\},\sigma,\tau,B)$ as a child of $(\+E,\+F,\sigma,\tau,B)$ in $\+T$, and for each $\pi\in \+Q_{\vbl(c)}$ and $\rho=\vio{c}$, add $(\+E,\+F,\sigma\land \pi,\tau\land \rho,B\Join c)$ as a child of $(\+E,\+F,\sigma,\tau,B)$ in $\+T$.\label{item:coupling-tree-2-c}
    \end{enumerate}
\end{enumerate}
Let $\+L$ be the set of leaf nodes in $\+T$. 
We further define:
\begin{itemize}
    \item $\+L_{\coup}\defeq \{(\+E,\+F,\sigma,\tau,B)\in \+L\mid \+E^{\sigma}=\+F^{\tau}\}$ as the set of ``coupled'' leaf nodes in $\+T$;
    \item $\+L_{\trun}\defeq \{(\+E,\+F,\sigma,\tau,B)\in \+L\mid |B|=K\}$ as the set of ``truncated'' leaf nodes in $\+T$;
    \item $\+{L}_{\text{valid}}\defeq\+L_{\trun}\cup \+L_{\coup}$ as the set of ``valid'' leaf nodes in $\+T$;
    \item $\+L_{\text{invld}}\defeq \+L\setminus \+{L}_{\text{valid}}$ as the set of ``invalid'' leaf nodes in $\+T$.
\end{itemize}
\end{definition}

\begin{remark}[well-definedness of coupling tree]\label{remark:coupling-tree-well-defined}
For each internal node $(\+E,\+F,\sigma,\tau,B)$ in the coupling tree $\+T$, 
the event $\+E\land \sigma\land \+F\land \tau$ is partitioned in some way into mutually disjoint events, e.g.:
\[
(\+E\cup \{c\})\land \sigma\land \+F\land \tau,
\qquad 
\+E\land (\sigma\land \vio{c})\land \+F\land (\tau\land \rho), 
\quad \forall \rho\in \+Q_{\vbl(c)},
\]
each of which corresponds to a child.
Therefore, all nodes in $\+T$ are distinct and $\+T$ is well-defined.

Note that every leaf node $(\+E,\+F,\sigma,\tau,B)\in \+L$ satisfies exactly one of the followings:
 \begin{itemize}
     \item $\+E^{\sigma}=\+F^{\tau}$;
     \item $|B|=K$;
     \item $\sigma$ violates $\+E$ or $\tau$ violates $\+F$.
 \end{itemize}
Therefore, the classification of leaf nodes in \Cref{definition:coupling-tree} is also well defined.
\end{remark}

\begin{remark}[difference between the coupling tree and the original coupling process]\label{remark:coupling-tree}
Note that the coupling tree $\+T$ may include more nodes than the tuples $(\+E,\+F,\sigma,\tau,B)$  that can possibly be generated by the process in \Cref{definition:procedure-root-to-leaf-path-1}. 
This is because some branch in the coupling tree $\+T$ may have 0 probability in the random process.
Some properties are unaffected by such distinction and still hold for the coupling tree, for instance, \Cref{lemma:bad-constraint-2-tree}, while others need to be reproved in this stronger setting.
\end{remark}

\subsubsection{The linear program}

We now introduce the linear program.
The LP is constructed on an instance of $K$-truncated coupling tree introduced in \Cref{definition:coupling-tree}, 
where each tree node $(\+E,\+F,\sigma,\tau,B)$ is associated with a pair of variables, corresponding to $p_{(\+E,\+F,\sigma,\tau,B)}^X$ and $p_{(\+E,\+F,\sigma,\tau,B)}^Y$. 
The constraints of the LP are derived from the properties listed in \Cref{lemma:quantities-properties,lemma:quantities-property-bad-event}.

\begin{definition}[linear program induced by the coupling] \label{definition:linear-program}
Let $\+T=\+T_K(\Phi,c_0)$ be the $K$-truncated coupling tree. Let $0\leq r_- \leq r_+$ be two parameters.
The following linear program (for checking feasibility) is defined on the variables $\hat{p}_{(\+E,\+F,\sigma,\tau,B)}^X$ and $\hat{p}_{(\+E,\+F,\sigma,\tau,B)}^Y$ for all $(\+E,\+F,\sigma,\tau,B)\in V(\+T)$:
\begin{enumerate}[label=\Roman*.]
    \item \emph{Range constraints}: 
\begin{align*}
&\hat{p}_{(\+C\setminus \{c_0\},\+C,\varnothing,\varnothing,\emptyset)}^X
=\hat{p}_{(\+C\setminus \{c_0\},\+C,\varnothing,\varnothing,\emptyset)}^Y=1;\\
&\hat{p}_{(\+E,\+F,\sigma,\tau,B)}^X,\hat{p}_{(\+E,\+F,\sigma,\tau,B)}^Y\in [0,1], \qquad \forall(\+E,\+F,\sigma,\tau,B)\in V(\+T).
\end{align*}
\label{item:linear-program-range-constraints}
\item \emph{Non-leaf constraints}: 
For each non-leaf node $(\+E,\+F,\sigma,\tau,B) \in V(\+T)\setminus\+L$, \label{item:linear-program-non-leaf-constraints}
    \begin{enumerate}
        \item if $\+F^{\tau}\not\subseteq \+E^{\sigma}$, letting $c$ be the smallest constraint in $\+F^{\tau}\setminus \+E^{\sigma}$ and $\pi=\vio{c}$,
\begin{align*}
\hat{p}_{(\+E,\+F,\sigma,\tau,B)}^X
&=\hat{p}^X_{(\+E\cup \{c\},\+F,\sigma,\tau,B)}=\sum\limits_{\rho\in \+Q_{\vbl(c)}}\hat{p}^X_{(\+E,\+F,\sigma\land \pi,\tau\land \rho,B\Join c)};\\  
\hat{p}_{(\+E,\+F,\sigma,\tau,B)}^Y
&=\hat{p}^Y_{(\+E\cup \{c\},\+F,\sigma,\tau,B)}+\hat{p}^Y_{(\+E,\+F,\sigma\land \pi,\tau\land \rho,B\Join c)}, \qquad \forall\rho\in \+Q_{\vbl(c)}.
\end{align*}
\label{item:linear-program-2-a}
        \item otherwise $\+F^{\tau}\subseteq \+E^{\sigma}$, letting $c$ be the smallest constraint in $\+E^{\sigma}\setminus \+F^{\tau}$ and $\rho=\vio{c}$,
\begin{align*}
            \hat{p}_{(\+E,\+F,\sigma,\tau,B)}^X
            &=\hat{p}^X_{(\+E,\+F\cup \{c\},\sigma,\tau,B)}+\hat{p}^X_{(\+E,\+F,\sigma\land \pi,\tau\land \rho,B\Join c)}, \qquad \forall\pi\in \+Q_{\vbl(c)};\\
            \hat{p}_{(\+E,\+F,\sigma,\tau,B)}^Y
            &=\hat{p}^Y_{(\+E,\+F\cup \{c\},\sigma,\tau,B)}=\sum\limits_{\pi\in \+Q_{\vbl(c)}}\hat{p}^Y_{(\+E,\+F,\sigma\land \pi,\tau\land \rho,B\Join c)}.
\end{align*}
        \label{item:linear-program-2-b}
    \end{enumerate}
\item \emph{Leaf constraints}: \label{item:linear-program-leaf-constraints}
For each leaf node $(\+E,\+F,\sigma,\tau,B)\in \+L$,
    \begin{enumerate}
    \item if it is a coupled leaf $(\+E,\+F,\sigma,\tau,B)\in\+L_{\coup}$, 
        \[
       r_-\cdot \hat{p}_{(\+E,\+F,\sigma,\tau,B)}^Y\leq \hat{p}_{(\+E,\+F,\sigma,\tau,B)}^X\leq r_+\cdot \hat{p}_{(\+E,\+F,\sigma,\tau,B)}^Y.
        \] \label{item:linear-program-coupled}
  \item if it is an invalid leaf $(\+E,\+F,\sigma,\tau,B)\in \+L_{\text{invld}}$,
  \begin{align*}
      \hat{p}_{(\+E,\+F,\sigma,\tau,B)}^Y=0, \text{ if $\sigma$ violates $\+E$}; \\
      \hat{p}_{(\+E,\+F,\sigma,\tau,B)}^X=0, \text{ if $\tau$ violates $\+F$}.
  \end{align*}
  \label{item:linear-program-invalid}
    \end{enumerate}
        \item \emph{Overflow constraints (via 2-tree)}: 
        Let $\mathbb{T}_K^{c_0}$ be the set of $2$-trees in $G_{\Phi}$ of size $K$ containing $c_0$. 
\begin{align*}
         \sum\limits_{\substack{(\+E,\+F,\sigma,\tau,B)\in \+L_{\trun}:\\ B=T\land \bm{x}\in \+Q^{\+E\setminus (\+C\setminus \{c_0\})\land \sigma}}} \hat{p}^X_{(\+E,\+F,\sigma,\tau,B)}
         &\leq (1-\mathrm{e}p)^{(D+1)K}\Pr[\*Y\sim \+P]{\bigwedge\limits_{c\in T}\+A_c\mid \*X=\bm{x}};
         &\forall T\in \mathbb{T}_K^{c_0}, \bm{x}\in \+Q;\\
     \sum\limits_{\substack{(\+E,\+F,\sigma,\tau,B)\in \+L_{\trun}:\\ B=T\land \bm{y}\in \+Q^{\+F\setminus \+C\land \tau}}} \hat{p}^Y_{(\+E,\+F,\sigma,\tau,B)}
     &\leq (1-\mathrm{e}p)^{(D+1)K}\Pr[\*X\sim \+P]{\bigwedge\limits_{c\in T}\+A_c\mid \*Y=\bm{y}},
     &\forall T\in \mathbb{T}_K^{c_0}, \bm{y}\in \+Q,
\end{align*}
where recall that $\+P$ is the product distribution and $\+A_c$ denotes the event defined in \eqref{eq:definition-ec}, i.e.
       \[\+A_c: \forall v\in \vbl(c), \quad\*X(v)=\vio{c}(v) \lor \*Y(v)=\vio{c}(v).
\]
\label{item:linear-program-overflow-constraints}

\end{enumerate}
\end{definition}
\vspace{-2ex}
Due to the independence between the constraints within a 2-tree, the probabilities involving $\+A_c$ in the overflow constraints (\Cref{item:linear-program-overflow-constraints}) of \Cref{definition:linear-program}, can be calculated explicitly as:
\begin{align*}
    \Pr[\*Y\sim \+P]{\bigwedge\limits_{c\in T}\+A_c\mid \*X=\bm{x}}
    &=\prod\limits_{c\in T}\prod\limits_{v\in \vbl(c)}|Q_v|^{-I\left[\bm{x}_v\neq \vio{c}_v\right]},\\
\Pr[\*X\sim \+P]{\bigwedge\limits_{c\in T}\+A_c\mid \*Y=\bm{y}}
&=\prod\limits_{c\in T}\prod\limits_{v\in \vbl(c)}|Q_v|^{-I\left[\bm{y}_v\neq \vio{c}_v\right]}.
\end{align*}
The overflow constraints (\Cref{item:linear-program-overflow-constraints}) of \Cref{definition:linear-program} are derived from \Cref{lemma:quantities-property-bad-event}. 
These constraints encapsulate a critical ``independence property'' pivotal to our analysis of the correlation decay.
They stand as the primary novelty in our linear program design, distinguishing it from previous LP-based approaches~\cite{Moi19,galanis2019counting,Vishesh21towards}. Note, however, that the overflow constraints slightly differ from \Cref{lemma:quantities-property-bad-event}, as the domain of $\bm{x}/\bm{y}$ have been extended and the conditions under the summation have changed. These differences serve technical purposes for the linear program to be efficient.

While it is natural to question the number of overflow constraints described in \Cref{definition:linear-program}, 
we will address this concern later (in \Cref{lemma:lp-properties}). 
Specifically, we will show that there are indeed $\mathrm{exp}(K\cdot \poly(k,\log q_{\max},D))$ distinct overflow constraints,
and moreover, it takes approximately that amount of time to both  write down the linear program and verify its feasibility.

\subsection{Analysis of the linear program}
We now present several properties for the above linear program. 
First, we show that the feasibility of the linear program can be checked efficiently.

\begin{proposition}\label{lemma:lp-properties}
For any $0\leq r_{-}\leq r_{+}$, the feasibility of the linear program in \Cref{definition:linear-program} can be checked within $\mathrm{exp}(K\cdot \poly(k,\log q_{\max},D))$ time. 
\end{proposition}

\Cref{lemma:lp-properties} holds because there are at most $\mathrm{exp}(K\cdot \poly(k,\log q_{\max},D))$ constraints that need to be checked.
Its formal proof is deferred to \Cref{sec:lp-properties}.

Next, we show the soundness of this linear program:
The true values of the marginal probabilities satisfy all the linear constraints.

\begin{lemma}\label{lemma:lp-feasibility}
Assume \Cref{condition:main-condition}. The LP in \Cref{definition:linear-program} is feasible for $r_{-}=r_{+}=\frac{|\+Q^{\+C\setminus \{c_0\}}|}{|\+Q^{\+C}|}$.
\end{lemma}
\begin{proof}

For each $(\+E,\+F,\sigma,\tau,B)\in V(\+T)$, we let
\begin{equation}\label{eq:linear-program-solutions}
\hat{p}_{(\+E,\+F,\sigma,\tau,B)}^X=\mu_{\+C}(\+F\land \tau), \quad \hat{p}_{(\+E,\+F,\sigma,\tau,B)}^Y=\mu_{\+C\setminus \{c_0\}}(\+E\land \sigma),
\end{equation}
and show that this is a feasible solution for $r_{-}=r_{+}=\frac{|\+Q^{\+C\setminus \{c_0\}}|}{|\+Q^{\+C}|}$. We remark that following \Cref{lemma:couple-well-defined}, this choice of the solution in \eqref{eq:linear-program-solutions} coincides with the real values of $p_{(\+E,\+F,\sigma,\tau,B)}^{X},p_{(\+E,\+F,\sigma,\tau,B)}^{Y}$ defined in \Cref{definition:imaginary-sampler-quantities}, but extends to all nodes in $V(\+T)$ instead of just $\+V^{\cp}$.

\Cref{item:linear-program-range-constraints} directly follows from \Cref{item:quantities-properties-1} of \Cref{lemma:quantities-properties}. Also, \Cref{item:linear-program-non-leaf-constraints} directly follows from \Cref{item:quantities-properties-2} of \Cref{lemma:quantities-properties}.

 \Cref{item:linear-program-invalid} directly follows from \eqref{eq:linear-program-solutions}. To prove \Cref{item:linear-program-coupled} , note that for $(\+E,\+F,\sigma,\tau,B)\in \+L_{\coup}$, $\+E^\sigma=\+F^{\tau}$ and hence $|\+Q^{\+E\land \sigma}|=|\+Q^{\+F\land \tau}|$ . Therefore,
\begin{itemize}
    \item if $|\+Q^{\+E\land \sigma}|=|\+Q^{\+F\land \tau}|>0$, then \Cref{item:linear-program-coupled}  follows directly from \Cref{item:quantities-properties-4} of \Cref{lemma:quantities-properties};
    \item otherwise $|\+Q^{\+E\land \sigma}|=|\+Q^{\+F\land \tau}|=0$, then $\mu_{\+C}(\+F\land \tau)=\mu_{\+C\setminus \{c_0\}}(\+E\land \sigma)=0$ and \Cref{item:linear-program-coupled} still holds.
\end{itemize}
 
The proof of \Cref{item:linear-program-overflow-constraints} falls into a similar vibe as the proof for \Cref{lemma:quantities-property-bad-event}. However, since now the domain for $\bm{x}$ and $\bm{y}$ has changed, we need to define another procedure similar to \Cref{definition:procedure-root-to-leaf-path-1} for explicitly identifying randomness to interpret the left-hand side in the inequality of \Cref{item:linear-program-overflow-constraints}. The formal proof of \Cref{item:linear-program-overflow-constraints} is deferred to \Cref{sec:lp-bounds}.
\end{proof}

At last, we show that the feasibility of the linear program implies that $r_{-}$ and $r_{+}$ provide respective 
lower and upper bound approximations for $\frac{|\+Q^{\+C\setminus \{c_0\}}|}{|\+Q^{\+C}|}$ with bounded multiplicative error. 
With this, we can apply a binary search to approximate the true value of $\frac{|\+Q^{\+C\setminus \{c_0\}}|}{|\+Q^{\+C}|}$.

\begin{lemma}\label{lemma:linear-program-error-bound}
Assume \Cref{condition:main-condition}. 
If the LP in \Cref{definition:linear-program} is feasible for parameters $0\leq r_-\leq r_+$, then
\[
\left(1-2\cdot 2^{-K}\right)r_-\leq \frac{|\+Q^{\+C\setminus \{c_0\}}|}{|\+Q^{\+C}|}\leq \left(1+2\cdot 2^{-K}\right)r_+.
\]
\end{lemma}
\begin{proof}
First, we claim the following property for the feasible solutions.
   \begin{claim}\label{lemma:lp-sum-to-1}
\begin{align*}
    \sum\limits_{\substack{(\+E,\+F,\sigma,\tau,B)\in \+L_{\textnormal{valid}}:\\\bm{x}\in \+Q^{\+E\land \sigma}}}\hat{p}_{(\+E,\+F,\sigma,\tau,B)}^X
    &=1, \quad\textnormal{for all }\bm{x}\in \+Q^{\+C\setminus \{c_0\}},\\
     \sum\limits_{\substack{(\+E,\+F,\sigma,\tau,B)\in \+L_{\textnormal{valid}}:\\\bm{y}\in \+Q^{\+F\land \tau}}}\hat{p}_{(\+E,\+F,\sigma,\tau,B)}^Y
     &=1, \quad\textnormal{for all }\bm{y}\in \+Q^{\+C}.
\end{align*}   
    \end{claim}
This claim can be proved by routinely verifying \Cref{definition:linear-program}. The  proof is deferred to \Cref{sec:lp-properties}.
    By summing the equations in \Cref{lemma:lp-sum-to-1} over all $c\in\+Q^{\+C\setminus \{c_0\}}$ and respectively all $c\in \+Q^{\+C}$, we obtain:
    \begin{equation}\label{eq:ss-sum}
    \begin{aligned}
                  |\+Q^{\+C\setminus \{c_0\}}|
                  &=\sum\limits_{\bm{x}\in \+Q^{\+C\setminus \{c_0\}}}\sum\limits_{\substack{(\+E,\+F,\sigma,\tau,B)\in \+L_{\text{valid}}:\\ \bm{x}\in \+Q^{\+E\land \sigma}}}\hat{p}_{(\+E,\+F,\sigma,\tau,B)}^X,\\
                      |\+Q^{\+C}|
                      &=\sum\limits_{\bm{y}\in \+Q^{\+C}}\sum\limits_{\substack{(\+E,\+F,\sigma,\tau,B)\in \+L_{\text{valid}}:\\\bm{y}\in \+Q^{\+F\land \tau}}}\hat{p}_{(\+E,\+F,\sigma,\tau,B)}^Y.
    \end{aligned}
    \end{equation}
\begin{claim}\label{lemma:linear-program-bad-leaf-loss}
Assume \Cref{condition:main-condition}. The followings hold. 
\begin{align*}
    \frac{1}{|\+Q^{\+C\setminus \{c_0\}}|}\sum\limits_{\bm{x}\in \+Q^{\+C\setminus \{c_0\}}}\sum\limits_{\substack{(\+E,\+F,\sigma,\tau,B)\in \+L_{\trun}:\\\bm{x}\in \+Q^{\+E\land \sigma}}}\hat{p}_{(\+E,\+F,\sigma,\tau,B)}^X 
    &\leq 2^{-K},\\
    \frac{1}{|\+Q^{\+C}|}\sum\limits_{\bm{y}\in \+Q^{\+C}}\sum\limits_{\substack{(\+E,\+F,\sigma,\tau,B)\in \+L_{\trun}:\\\bm{y}\in \+Q^{\+F\land \tau}}}\hat{p}_{(\+E,\+F,\sigma,\tau,B)}^Y
    &\leq 2^{-K}.
\end{align*}
\end{claim}
This claim can be proved by a similar argument to the analysis of correlation decay in \Cref{sec:correlation-decay}.  The proof is deferred to \Cref{sec:lp-bounds}.

We then have
\begin{align*}
(\text{by \eqref{eq:ss-sum}})&&
|\+Q^{\+C\setminus \{c_0\}}|
 =&\sum\limits_{\bm{x}\in \+Q^{\+C\setminus \{c_0\}}}\sum\limits_{\substack{(\+E,\+F,\sigma,\tau,B)\in \+L_{\text{valid}}:\\\bm{x}\in \+Q^{\+E\land \sigma}}}\hat{p}_{(\+E,\+F,\sigma,\tau,B)}^X\\
&&=&\sum\limits_{\bm{x}\in \+Q^{\+C\setminus \{c_0\}}}\sum\limits_{\substack{(\+E,\+F,\sigma,\tau,B)\in \+L_{\coup}:\\\bm{x}\in \+Q^{\+E\land \sigma}}}\hat{p}_{(\+E,\+F,\sigma,\tau,B)}^X\\
&&&+\sum\limits_{\bm{x}\in \+Q^{\+C\setminus \{c_0\}}}\sum\limits_{\substack{(\+E,\+F,\sigma,\tau,B)\in \+L_{\trun}:\\\bm{x}\in \+Q^{\+E\land \sigma}}}\hat{p}_{(\+E,\+F,\sigma,\tau,B)}^X\\
(\text{by \Cref{lemma:linear-program-bad-leaf-loss}})
&&\leq & \sum\limits_{(\+E,\+F,\sigma,\tau,B)\in \+L_{\coup}}|\+Q^{\+E\land \sigma}|\cdot \hat{p}_{(\+E,\+F,\sigma,\tau,B)}^X+2^{-K}\cdot |\+Q^{\+C\setminus \{c_0\}}|. 
\end{align*}
Therefore, we have
\begin{align*}
    \left|\+Q^{\+C\setminus \{c_0\}}\right|
    \in &\left[\hat{z}^X,\,\,\left(1+2\cdot 2^{-K}\right)\hat{z}^X \right],\\
    &\text{where }\hat{z}^X\triangleq \sum\limits_{(\+E,\+F,\sigma,\tau,B)\in \+L_{\coup}}|\+Q^{\+E\land \sigma}|\cdot \hat{p}_{(\+E,\+F,\sigma,\tau,B)}^X;\\
    \left|\+Q^{\+C}\right|
    \in &\left[\hat{z}^Y,\,\,\left(1+2\cdot 2^{-K}\right)\hat{z}^Y \right],\\
    &\text{where }\hat{z}^Y\triangleq \sum\limits_{(\+E,\+F,\sigma,\tau,B)\in \+L_{\coup}}|\+Q^{\+F\land \tau}|\cdot \hat{p}_{(\+E,\+F,\sigma,\tau,B)}^Y.
\end{align*}
Therefore, 
\begin{align*}
\frac{|\+Q^{\+C\setminus \{c_0\}}|}{|\+Q^{\+C}|}
&\leq 
\left(1+2\cdot 2^{-K}\right)\frac{\sum\limits_{(\+E,\+F,\sigma,\tau,B)\in \+L_{\coup}}|\+Q^{\+E\land \sigma}|\cdot \hat{p}_{(\+E,\+F,\sigma,\tau,B)}^X}{\sum\limits_{(\+E,\+F,\sigma,\tau,B)\in \+L_{\coup}}|\+Q^{\+F\land \tau}|\cdot \hat{p}_{(\+E,\+F,\sigma,\tau,B)}^Y}\\
&\leq \left(1+2\cdot 2^{-K}\right)r_{+},
\end{align*}
where the last inequality is due to the leaf constraints of the LP (\Cref{item:linear-program-leaf-constraints} of \Cref{definition:linear-program}).
Symmetrically, we have
\begin{align*}
\frac{|\+Q^{\+C\setminus \{c_0\}}|}{|\+Q^{\+C}|}
&\geq \left(1-2\cdot 2^{-K}\right)\frac{\sum\limits_{(\+E,\+F,\sigma,\tau,B)\in \+L_{\coup}}|\+Q^{\+E\land \sigma}|\cdot \hat{p}_{(\+E,\+F,\sigma,\tau,B)}^X}{\sum\limits_{(\+E,\+F,\sigma,\tau,B)\in \+L_{\coup}}|\+Q^{\+F\land \tau}|\cdot \hat{p}_{(\+E,\+F,\sigma,\tau,B)}^Y}\\
&\geq \left(1-2\cdot 2^{-K}\right)r_{-}.
\end{align*}
Together, we have
\[
\left(1-2\cdot 2^{-K}\right)r_-\leq \frac{|\+Q^{\+C\setminus \{c_0\}}|}{|\+Q^{\+C}|}\leq \left(1+2\cdot 2^{-K}\right)r_+.\qedhere
\]
\end{proof}

\section{Counting and sampling via linear programming} 
In this section, we utilize the linear program introduced in \Cref{sec:algorithms} to prove \Cref{theorem:main-counting,theorem:main-sampling}.
\subsection{Constraint-wise self-reducibility in the local lemma regime}
To obtain the algorithms for approximate counting and sampling required in \Cref{theorem:main-counting,theorem:main-sampling}, we will first provide a \emph{marginal estimator} and a \emph{dynamic sampler}.
Within the local lemma regime characterized by \Cref{condition:main-condition},
the marginal estimator can efficiently estimate the ratio between the numbers of satisfying assignments for two sets of constraints differing in exactly one constraint;
and the dynamic sampler can dynamically update a sample for the uniform satisfying assignment when a new constraint is added.
Such a marginal estimator and a dynamic sampler are stated in the following two theorems.

\begin{theorem}[marginal estimator]\label{lemma:marginal-estimator}
There exists a deterministic algorithm such that for any $\Phi=(V,\+Q,\+C)$ satisfying \Cref{condition:main-condition}, any $c_0\in\+C$ and $\varepsilon\in (0,1)$, 
within time $O\left(\left(\frac{1}{\varepsilon}\right)^{\poly(k,D,\log q_{\max})}\right)$ the algorithm outputs an estimate $\hat{r}$ satisfying
\[
(1-\varepsilon)\frac{Z_{\Phi}}{Z_{\Phi^{c_0}}}\leq \hat{r}\leq (1+\varepsilon)\frac{Z_{\Phi}}{Z_{\Phi^{c_0}}},
\]
where $Z_{\Phi}$ and $Z_{\Phi^{c_0}}$ respectively denote the numbers of solutions to $\Phi$ and $\Phi^{c_0}=(V,\+Q,\+C\setminus \{c_0\})$.
\end{theorem}

\begin{theorem}[dynamic sampler]\label{lemma:marginal-sampler}
There exists a randomized algorithm such that for any $\Phi=(V,\+Q,\+C)$ satisfying \Cref{condition:main-condition}, any $c_0\in\+C$
and $\varepsilon\in (0,1)$, given access to a sample $\sigma\sim\mu_{\Phi^{c_0}}$ of the uniform distribution $\mu_{\Phi^{c_0}}$ over all solutions to $\Phi^{c_0}=(V,\+Q,\+C\setminus \{c_0\})$, 
within time $O\left(\left(\frac{1}{\varepsilon}\right)^{\poly(k,D,\log q_{\max})}\right)$ the algorithm updates the $\sigma$ to an assignment $\tau\in\+Q$ satisfying
\[
\dtv(\tau,\mu_{\Phi})\leq \varepsilon.
\]
\end{theorem}

Through a constraint-wise self-reduction that maintains the local lemma regime, 
\Cref{theorem:main-counting} is implied by \Cref{lemma:marginal-estimator} and \Cref{theorem:main-sampling} is implied by \Cref{lemma:marginal-sampler}.


\begin{proof}[Proof of \Cref{theorem:main-counting}]
Let $\+C=\{c_1,c_2,\dots,c_m\}$. Note that we have $m\leq n(D+1)$.  For each $0\leq i\leq m$, we let $\+C_i=\{c_1,c_2,\dots,c_i\}$, $\Phi_i=(V,\+Q,\+C_i)$ and $Z_{\Phi_i}$ denote the number of solutions to $\Phi_i$. Note that $\Phi_m=\Phi$. We can then decompose $Z_{\Phi_m}$ in the following way:
\begin{equation}
    Z_{\Phi_m}=Z_{\Phi_0}\cdot\frac{Z_{\Phi_{m}}}{Z_{\Phi_0}}=Z_{\Phi_0}\cdot \prod\limits_{i=1}^{m}\frac{Z_{\Phi_i}}{Z_{\Phi_{i-1}}}.
\end{equation}
Note that $Z_{\Phi_0}=\prod\limits_{v\in V}|Q_v|$ is easily computable in $\poly(n,q_{\max})$ time. We can then apply \Cref{lemma:marginal-estimator} to approximate each $\frac{Z_{\Phi_i}}{Z_{\Phi_{i-1}}}$up to multiplicative error $\frac{\varepsilon}{4m}$ and multiply them together, by that $\Phi_i$ satisfies \Cref{condition:main-condition} for each $0\leq i\leq m$. In this way we obtain an estimate up to error $\varepsilon$ within time $O\left(\left(\frac{n}{\varepsilon}\right)^{\poly(k,D,\log q_{\max})}\right)$, as desired.
\end{proof}

\begin{proof}[Proof of \Cref{theorem:main-sampling}]
Let $\+C=\{c_1,c_2,\dots,c_m\}$. Note that we have $m\leq n(D+1)$.  For each $0\leq i\leq m$, we let $\+C_i=\{c_1,c_2,\dots,c_i\}$ and $\Phi_i=(V,\+Q,\+C_i)$. We build a sequence of assignments $\sigma_0,\sigma_1,\dots,\sigma_m\in \+Q$ 
as follows: 
\begin{itemize}
    \item Sample $\sigma_0\sim \+P$.
    \item For each $1\leq i\leq m$, use the dynamic sampler in \Cref{lemma:marginal-sampler} to obtain $\sigma_i$ from $\sigma_{i-1}$ with input formula $(V,\+Q,\+C_{i})$, input constraint $c_{i}$,  input parameter $\frac{\varepsilon}{m}$, and input sample $\sigma_{i-1}$. We output an arbitrary final assignment instead if $\sigma_i$ isn't a satisfying solution to $\Phi_i$ for some $0\leq i\leq m$.
\end{itemize}

By \Cref{lemma:marginal-sampler} and induction, it easily follows from maximally coupling $\sigma_{i}$ with a uniform solution of $\Phi_i$ and applying the coupling lemma (\Cref{lemma:coupling-lemma}) that
\[
\forall 0\leq i\leq m, \dtv(\sigma_i,\mu_{\+C_{i}})\leq \frac{i}{m}\cdot \varepsilon.
\]
Also, the overall running time is bounded by $O\left(\left(\frac{n}{\varepsilon}\right)^{\poly(k,D,\log q_{\max})}\right)$. Therefore, the theorem is proved.
\end{proof}

\subsection{Construction of the marginal estimator/sampler}

It remains to prove \Cref{lemma:marginal-estimator,lemma:marginal-sampler}. 
\Cref{lemma:marginal-estimator} is proved  by directly applying the linear program constructed and analyzed in \Cref{sec:algorithms}.

\begin{proof}[Proof of \Cref{lemma:marginal-estimator}]
We set the truncation threshold as.
\[
K=1+\log \varepsilon^{-1}.
\]
Applying a binary search on the parameters $r_-$ and $r_+$, 
one can obtain an estimate of $\frac{|\+Q^{\+C\setminus \{c_0\}}|}{|\+Q^{\+C}|}$ within multiplicative error $\varepsilon$ by \Cref{lemma:linear-program-error-bound}.
The running time bound follows from \Cref{lemma:lp-properties}.
\end{proof}

We then prove \Cref{lemma:marginal-sampler}.  The dynamic sampler claimed in \Cref{lemma:marginal-sampler} is constructed as follows.

\begin{definition}[dynamic sampler]\label{definition:marginal-sampling-algorithm}
The algorithm takes as input a CSP $\Phi=(V,\+Q,\+C)$, a constraint $c_0\in \+C$ and $\varepsilon\in (0,1)$, 
and is given access to an assignment $\sigma^{\text{in}}\in \+Q$ which satisfies $\Phi^{c_0}=(V,\+Q,\+C\setminus \{c_0\})$. 

    The algorithm proceeds as follows to update  $\sigma^{\text{in}}$ to a new assignment $\sigma^{\text{out}}\in \+Q$:
    \begin{itemize}
        \item Set the truncation threshold as
        \[
            K=1+\log \left(\frac{\varepsilon}{4}\right)^{-1}.
        \]
        Construct the LP in \Cref{definition:linear-program} on the $K$-truncated coupling tree $\+T=\+T_K(\Phi,c_0)$.
        Apply a binary search to narrow down the interval $[r_{-},r_{+}]$ so that $r_{-}\geq \frac{4+\varepsilon}{4+2\varepsilon}r_{+}$ and the LP is still feasible with parameters $r_{-}$ and $r_{+}$. 
        By \Cref{lemma:linear-program-error-bound}, it is guaranteed that
        \[
        \left(1-\frac{\varepsilon}{4}\right)r_{-}\leq \frac{|\+Q^{\+C\setminus \{c_0\}}|}{|\+Q^{\+C}|}\leq \left(1+\frac{\varepsilon}{4}\right)r_+.
        \]
        Let $\left\{\hat{p}_{(\+E,\+F,\sigma,\tau,B)}^X,\hat{p}_{(\+E,\+F,\sigma,\tau,B)}^Y\right\}_{(\+E,\+F,\sigma,\tau,B)\in V(\+T)}$ be the corresponding LP feasible solution. 
       
        \item Use the following random process to generate a root-to-leaf path in $\+T$. 
        The process starts from the root node $(\+C\setminus \{c_0\},\+C,\varnothing,\varnothing,\emptyset)$. 
        In each step, at a non-leaf node $(\+E,\+F,\sigma,\tau,B)$, it does:
\begin{enumerate}
        \item if $\+F^{\tau}\not\subseteq \+E^{\sigma}$, denoted by $c$ be the smallest pinned constraint in $\+F^{\tau}\setminus \+E^{\sigma}$, 
        \begin{itemize}
            \item move to $(\+E\cup \{c\},\+F,\sigma,\tau,B)$ if $c$ is satisfied by $\sigma^{\text{in}}_{\vbl(c)}$;
            \item otherwise, move to $(\+E,\+F,\sigma\land \sigma^{\text{in}}_{\vbl(c)},\tau\land \rho,B\Join c)$  for each $\rho\in \+Q_{\vbl(c)}$ w.p. 
            \[
            \frac{\hat{p}^X_{\+E,\+F,\sigma\land \sigma^{\text{in}}_{\vbl(c)},\tau\land \rho,B\Join c}}{\hat{p}_{(\+E,\+F,\sigma,\tau,B)}^X};
            \]
        \end{itemize}
        \item otherwise $\+F^{\tau}\subseteq \+E^{\sigma}$, denoted by $c$ be the smallest pinned constraint in $\+E^{\sigma}\setminus \+F^{\tau}$,
        \begin{itemize}
            \item move to $(\+E,\+F\cup \{c\},\sigma,\tau,B)$ with probability
            \[
            \frac{\hat{p}^X_{(\+E,\+F\cup \{c\},\sigma,\tau,B)}}{\hat{p}_{(\+E,\+F,\sigma,\tau,B)}^X};
            \]
            \item move to $(\+E,\+F,\sigma\land \sigma^{\text{in}}_{\vbl(c)},\tau\land \rho,B\Join c)$ for $\rho=\vio{c}$  with probability   \[
            \frac{\hat{p}^X_{(\+E,\+F,\sigma\land \sigma^{\text{in}}_{\vbl(c)},\tau\land \rho,B\Join c)}}{\hat{p}_{(\+E,\+F,\sigma,\tau,B)}^X}.
            \]
        \end{itemize}

        
       \end{enumerate}
In the end, the random process arrives at a random leaf $(\bm{\+E},\bm{\+F},\bm{\sigma},\bm{\tau},\bm{B})\in \+L$.
Let $P^{\samp}$ denote the root-to-leaf path generated by this process.
          \item  If the sampled leaf node $(\bm{\+E},\bm{\+F},\bm{\sigma},\bm{\tau},\bm{B})\notin \+L_{\coup}$, 
          then let $\sigma^{\text{out}}\in \+Q$ be an arbitrary unsatisfying assignment of $\Phi$; otherwise, update  $\sigma^{\text{in}}$  by changing the values of the assigned variables to the corresponding values as in $\bm{\tau}$, i.e.~let $\sigma^{\text{out}}\gets \bm{\tau}\land \sigma^{\text{in}}_{V\setminus\Lambda(\bm{\tau})}$.
    \end{itemize}
\end{definition}

We then finish the section by proving the following lemma, which directly implies \Cref{lemma:marginal-sampler}.

\begin{lemma}\label{lemma:marginal-sampler-tvd-bound}
Assume \Cref{condition:main-condition} and $\sigma^{\textnormal{in}}\sim \mu_{\+C\setminus \{c_0\}}$. 
Then the $\sigma^{\textnormal{out}}$ produced by the dynamic sampler in \Cref{definition:marginal-sampling-algorithm} satisfies
\[
\dtv(\sigma^{\textnormal{out}},\mu_{\+C})\leq \varepsilon.
\]
\end{lemma}
\begin{proof}
We need the following lemma, whose proof is deferred to \Cref{sec:lp-bounds}.

\begin{lemma}\label{lemma:sampling-alg-path-probability-bound}
Assume \Cref{condition:main-condition}. Let $\sigma^{\textnormal{in}}\sim \mu_{\+C\setminus \{c_0\}}$. For any $(\+E,\+F,\sigma,\tau,B)\in V(\+T)$,  we have
\begin{align*}
\Pr{(\+E,\+F,\sigma,\tau,B)\in P^{\samp}}=&\mu_{\+C\setminus \{c_0\}}(\+E\land \sigma)\cdot \hat{p}_{(\+E,\+F,\sigma,\tau,B)}^X.
\end{align*}
Moreover, conditioning on $(\+E,\+F,\sigma,\tau,B)\in P^{\samp}$, it follows that
\[
\sigma^{\textnormal{in}}\sim \mu_{\+E}^{\sigma},
\]
for each $(\+E,\+F,\sigma,\tau,B)$ such that 
$\Pr{(\+E,\+F,\sigma,\tau,B)\in P^{\samp}}>0$.
\end{lemma}
   Then for each $\bm{y}\in \+Q^{\+C}$, we have
  \begin{align*}
\Pr{\sigma^{\text{out}}=\bm{y}}= &\sum\limits_{\substack{(\+E,\+F,\sigma,\tau,B)\in \+L_{\coup}:\\\bm{y}\in \+Q^{\+F\land \tau}}}\mu_{\+C\setminus \{c_0\}}(\+E\land \sigma)\cdot \hat{p}_{(\+E,\+F,\sigma,\tau,B)}^X\cdot \frac{1}{|\+Q^{\+E\land \sigma}|}\\
(\text{by \Cref{item:linear-program-coupled} of \Cref{definition:linear-program}})\quad\geq & r_{-}\cdot \sum\limits_{\substack{(\+E,\+F,\sigma,\tau,B)\in \+L_{\coup}:\\\bm{y}\in \+Q^{\+F\land \tau}}}\mu_{\+C\setminus \{c_0\}}(\+E\land \sigma)\cdot \hat{p}_{(\+E,\+F,\sigma,\tau,B)}^Y\cdot \frac{1}{|\+Q^{\+E\land \sigma}|}\\
(\star)\quad = &r_{-}\cdot \frac{1}{|\+Q^{\+C\setminus \{c_0\}}|}\cdot\sum\limits_{\substack{(\+E,\+F,\sigma,\tau,B)\in \+L_{\coup}:\\\bm{y}\in \+Q^{\+F\land \tau}}} \hat{p}_{(\+E,\+F,\sigma,\tau,B)}^Y.
\end{align*}
 Here, the first equality is by that $\+E^{\sigma}=\+F^{\tau}$ for each $(\+E,\+F,\sigma,\tau,B)\in \+L_{\coup}$, and hence each solution in $\+Q^{\+F\land \tau}$ is generated with equal probability $\frac{1}{\abs{\+Q^{\+F\land \tau}}}=\frac{1}{\abs{\+Q^{\+E\land \sigma}}}$ by \Cref{lemma:sampling-alg-path-probability-bound}. The $\star$ equality is by the chain rule and that $\+E\land \sigma\implies \+C\setminus \{c_0\}$ for each $(\+E,\+F,\sigma,\tau,B)\in V(\+T)$, following the argument in \Cref{remark:coupling-tree-well-defined}.

 Consider a distribution $\nu$ defined over $\+Q$ such that for each $\bm{y}\in \+Q^{\+C}$, we let
\begin{equation}\label{eq:definition-nu}
 \nu[\bm{y}]=\frac{1}{|\+Q^{\+C}|}\cdot\sum\limits_{\substack{(\+E,\+F,\sigma,\tau,B)\in \+L_{\coup}:\\\bm{y}\in \+Q^{\+F\land \tau}}}\hat{p}_{(\+E,\+F,\sigma,\tau,B)}^Y,
 \end{equation}
 and that distribution of $\nu$ over unsatisfying assignments of $\Phi$ is carefully designed in a way that
\begin{equation}\label{eq:nu-design}
 \dtv\left(\nu,\sigma^{\text{out}}\right)=\abs{\frac{1}{|\+Q^{\+C}|}-r_{-}\cdot \frac{1}{|\+Q^{\+C\setminus \{c_0\}}|}}\sum\limits_{\bm{y}\in \+Q^{\+C}}\sum\limits_{\substack{(\+E,\+F,\sigma,\tau,B)\in \+L_{\coup}:\\\bm{y}\in \+Q^{\+F\land \tau}}}\hat{p}_{(\+E,\+F,\sigma,\tau,B)}^Y.
\end{equation}

By  $\left(1-\frac{\varepsilon}{4}\right)r_{-}\cdot \frac{|\+Q^{\+C\setminus \{c_0\}}|}{|\+Q^{\+C}|}\leq \left(1+\frac{\varepsilon}{4}\right)r_+$ and $r_{-}\geq \frac{4+\varepsilon}{4+2\varepsilon}r_{+}$ ensured by the process, we then have
\begin{equation}\label{eq:dtv-transform}
\left(1-\frac{\varepsilon}{2}\right)\cdot \frac{|\+Q^{\+C\setminus \{c_0\}}|}{|\+Q^{\+C}|}\leq r_{-}\leq \left(1+\frac{\varepsilon}{2}\right)\cdot \frac{|\+Q^{\+C\setminus \{c_0\}}|}{|\+Q^{\+C}|}.
\end{equation}
Combining \eqref{eq:nu-design} and \eqref{eq:dtv-transform}, we obtain
\begin{equation}\label{eq:dtv-1}
 \dtv(\nu,\sigma^{\text{out}})\leq \frac{\varepsilon}{2}\cdot \frac{1}{|\+Q^{\+C}|}\sum\limits_{\bm{y}\in \+Q^{\+C}}\sum\limits_{\substack{(\+E,\+F,\sigma,\tau,B)\in \+L_{\coup}:\\\bm{y}\in \+Q^{\+F\land \tau}}}\hat{p}_{(\+E,\+F,\sigma,\tau,B)}^Y\leq \frac{\varepsilon}{2}.
 \end{equation}

Now that by \eqref{eq:definition-nu} and \Cref{lemma:lp-sum-to-1}, for each $\bm{y} \in \+Q^{\+C}$ we have 
  $ \nu[\bm{y}]\leq \frac{1}{|\+Q^{\+C}|}=\mu_{\+C}(\bm{y})$.
Hence, 
\begin{align*}
\dtv(\mu_{\+C},\nu)=&\sum\limits_{\bm{y}\in \+Q^{\+C}}\left(\frac{1}{|\+Q^{\+C}|}-\frac{1}{|\+Q^{\+C}|}\cdot\sum\limits_{\substack{(\+E,\+F,\sigma,\tau,B)\in \+L_{\coup}:\\\bm{y}\in \+Q^{\+F\land \tau}}} \hat{p}_{(\+E,\+F,\sigma,\tau,B)}^Y\right)\\
(\text{by \Cref{lemma:lp-sum-to-1}})\quad= &\frac{1}{|\+Q^{\+C}|}\sum\limits_{\bm{y}\in \+Q^{\+C}}  \sum\limits_{\substack{(\+E,\+F,\sigma,\tau,B)\in \+L_{\trun}:\\\bm{y}\in \+Q^{\+F\land \tau}}} \hat{p}_{(\+E,\+F,\sigma,\tau,B)}^Y\\
(\text{by \Cref{lemma:linear-program-bad-leaf-loss}})\quad\leq & \frac{\varepsilon}{2}.
\end{align*}
Combining with \eqref{eq:dtv-1} and using the triangle inequality for the total variation distance, we obtain
\[
\dtv(\mu_{\+C},\sigma^{\text{out}})\leq \dtv(\mu_{\+C},\nu)+\dtv(\nu,\sigma^{\text{out}})=\varepsilon,
\]
finishing the proof of the lemma. 
\end{proof}

\section{Conclusion and open problems}\label{sec:conclusions}
In this paper, we present polynomial-time algorithms for approximate counting/almost uniform sampling atomic constraint satisfaction solutions in the regime of $pD^{2+o(1)}\lesssim 1$ (\Cref{condition:main-condition}), where the $o(1)$ factor approaches zero as the minimum domain size grows. This regime almost matches the lower bound $pD^2\lesssim 1$ given by ~\cite{BGG19,galanis2021inapproximability}. Even for the worst-case of Boolean domains, our regime $pD^{4.82}\lesssim 1$ still improves over the previous best regime of $pD^5\lesssim 1$~\cite{he2022counting}.

At the heart of our approach is a novel constraint-wise coupling for CSPs. This coupling, along with its analyses, sharply captures the critical phenomenon for counting/sampling LLL and may be of independent interest, helping pave the way for future improvements.



Here, we outline several future research directions that our findings may serve as a basis for:
\begin{itemize}
    \item 
    An obvious open problem is establishing the optimal counting/sampling LLL for small domain sizes $q$, especially for counting/sampling $k$-CNF, where the domain size $q=2$.
    \item 
Another open direction is to generalize the current approach to general CSPs, possibly with non-atomic constraints.
    \item 
    Last but not the least,
    an open aspect is to devise \emph{fast} algorithms for sampling LLL that runs in polynomial time when $k$ and $D$ are allowed to be unbounded or in near-linear time when  $k$ and $D$ are bounded. Such fast algorithms should avoid relying on exhaustive enumerations of local structures, unlike the deterministic approximate counting algorithms to this day.
    
\end{itemize}

\appendix
\bibliographystyle{alpha}
\bibliography{references} 
\clearpage

\section{Basic properties of the linear program}\label{sec:lp-properties}

In this section, we will prove two basic properties of the linear program constructed in \Cref{sec:algorithms}, specifically \Cref{lemma:lp-properties} and \Cref{lemma:lp-sum-to-1}. 

Recall that we have fixed a CSP $\Phi=(V,\+Q,\+C)$ and a constraint $c_0\in \+C$. Also, recall the notations in \Cref{definition:imaginary-sampler-quantities}, the $K$-truncated coupling tree $\+T=\+T_K(\Phi,c_0)$ in \Cref{definition:coupling-tree}, and the linear program in \Cref{definition:linear-program}.

\subsection{Proof of \Cref{lemma:lp-properties}}

Before proving \Cref{lemma:lp-properties}, we need to first prove the following lemma, which bounds the depth of the coupling tree $\+T$.

\begin{lemma}\label{lemma:coupling-tree-depth-bound}
$\+T$ has depth at most $KD(D+1)+1$. 
\end{lemma}
\begin{proof}
Suppose to the contrary that there exists a node $(\+E',
\+F',\sigma',\tau',B')\in V(\+T)$ of depth $KD(D+1)+2$. We track the size of $\+E^{\sigma}\triangle\+F^{\tau}$ when going down the path between the root and $(\+E', \+F',\sigma',\tau',B')$ and denote it by $t$. Initially, by \Cref{condition:main-condition}, we have $t=1$. Note that each time we either add some $c$ into $\+E$ or $\+F$, and let $t$ decreases by one; or we assign some values to $\sigma$ and $\tau$ on $\vbl(c)$, and let $t$ increases by at most $D-1$ (note that at most $D$ new elements can be added into $\+E^{\sigma}\triangle\+F^{\tau}$ and that $c$ is removed from $\+E^{\sigma}\triangle\+F^{\tau})$. We denote the number of times the latter operation being executed by $i$. 

By $(\+E', \+F',\sigma',\tau',B')$ is of depth $KD(D+1)+2$, the above step is repeated for $KD(D+1)+2$ times. 
Note that at $(\+E', \+F',\sigma',\tau',B')$ we have $t\geq 0$, then
\[
(D-1)\cdot i+1-( KD(D+1)+2-i) \geq 0,
\]
and hence $i>K(D+1)$. Note that by \Cref{definition:coupling-tree} and \eqref{eq:definition-f}, we have $|B'|>K$, contradicting the truncation condition at \Cref{item:coupling-tree-2-a} of \Cref{definition:coupling-tree}. Therefore, the lemma is proved.
\end{proof}

We are now ready to prove \Cref{lemma:lp-properties}.

\begin{proof}[Proof of \Cref{lemma:lp-properties}]
    Note that the branching number of the coupling tree is at most $(q_{\max})^{2k}$ . Hence by \Cref{lemma:coupling-tree-depth-bound}, the size of the coupling tree is bounded by $\mathrm{exp}(K\cdot \poly(k,\log q_{\max},D))$.

    Therefore, the total number of constraints of Items \ref{item:linear-program-range-constraints} through \ref{item:linear-program-leaf-constraints} in \Cref{definition:linear-program} is also bounded by $\mathrm{exp}(K\cdot \poly(k,\log q,D))$. For \Cref{item:linear-program-overflow-constraints} of \Cref{definition:linear-program}, we only show how to efficiently check the feasibility for the first inequality, and the other one follows analogously. We note that after fixing some $T\in \mathbb{T}_K^{c_0}$, for all leaf nodes $(\+E,\+F,\sigma,\tau,B)\in \+L_{\trun}$ satisfying $B=T$, the following properties hold:
    \begin{itemize}
        \item for any $v\in \Lambda(\sigma)$, there exist $c\in \+C,c'\in T$ such that $v\in \vbl(c)$ and $\vbl(c)\cap \vbl(c')\neq \emptyset$;
        \item for any $c\in \+E\setminus (\+C\setminus \{c_0\})$, we have $\Lambda(\sigma)\cap \vbl(\unpin{c})\neq \emptyset$. 
    \end{itemize}
    Here, the first property can be shown using proof by contradiction: Suppose that such $c\in \+C, c'\in T$ cannot be found. by the time $v$ is assigned in $\sigma/\tau$ through the root-to-leaf path in $\+T$, let $c^*$ be the (pinned) constraint chosen at the time, then $\unpin{(c^*)}$ doesn't share variable with any constraint in $B$ at the time according to the assumption, and thus $\unpin{(c^*)}$ must be added into $B$ by \Cref{definition:linear-program} and \eqref{eq:definition-f}. Then we have a contradiction by letting $c=c'=c^*$. The second property follows from that each (pinned) constraint $c$ chosen by the procedure and added into $\+E$ in \Cref{definition:linear-program} lies in $\+E^{\sigma}\triangle\+F^{\tau}$ at the time, and $c$ can added into  $\+E^{\sigma}\triangle\+F^{\tau}$ only by either
    \begin{itemize}
        \item $c$ is exactly $c_0$;
        \item or $c$ is pinned by $\sigma/\tau$, i.e. $\vbl(\unpin{c})\neq \vbl(c)$.
    \end{itemize}
    
    By the above properties, the feasibility of the first equation can be directly verified by enumerating the assignment of $\bm{x}$ on variables of constraints within distance $2$ of $T$ in $G_{\Phi}$. This says it suffices to only check $(q_{\max})^{Kk(D+1)^2}$ (partial) assignments of $\bm{x}$ to verify the first inequality in \Cref{item:linear-program-overflow-constraints} of \Cref{definition:linear-program}. Combining with the bound on the number of $2$-trees in \Cref{lemma:num-2-tree}, we have the number of constraints in \Cref{item:linear-program-overflow-constraints} of \Cref{definition:linear-program} can be bounded by $\mathrm{exp}(K\cdot \poly(k,\log q_{\max},D))$.
    
     Hence, we have both the number of variables and constraints of the linear program defined in \Cref{definition:linear-program} is at most $\mathrm{exp}(K\cdot \poly(k,\log q_{\max},D))$. Then the lemma follows from the standard running time guarantees for linear programming (e.g., \cite{khachiyan1980polynomial}).
\end{proof}

\subsection{Proof of \Cref{lemma:lp-sum-to-1}}

The proof of \Cref{lemma:lp-sum-to-1} comes from verifying several basic definitions of the linear program.

\begin{proof}[Proof of \Cref{lemma:lp-sum-to-1}]
  We only prove the first equality of \Cref{lemma:lp-sum-to-1}, then the second equality follows analogously. By \Cref{item:linear-program-non-leaf-constraints} of \Cref{definition:linear-program} and an induction on the increasing depth of $\+T$, we have
    \[
   \sum\limits_{\substack{(\+E,\+F,\sigma,\tau,B)\in \+L:\\\bm{x}\in \+Q^{\+E\land \sigma}}}\hat{p}_{(\+E,\+F,\sigma,\tau,B)}^X=1, \quad\text{for all }\bm{x}\in \+Q^{\+C\setminus \{c_0\}}.
    \]

     Note that for any $x\in \+Q^{\+C\setminus \{c_0\}}$, any $(\+E,\+F,\sigma,\tau,B)\in \+L_{\textnormal{invld}}$ such that $\bm{x}\in \+Q^{\+E\land \sigma}$, it must hold that $\tau$ violates $\+F$. Then, the lemma holds immediately by applying \Cref{item:linear-program-invalid} of \Cref{definition:linear-program}. 
\end{proof}

\section{Probabilistic properties of the linear program}\label{sec:lp-bounds}

In this section, we prove several important probabilistic properties related to the linear program. Specifically, we will finish the proof of \Cref{lemma:lp-feasibility} regarding \Cref{item:linear-program-overflow-constraints} of \Cref{definition:linear-program}, also prove \Cref{lemma:linear-program-bad-leaf-loss} and \Cref{lemma:sampling-alg-path-probability-bound}. The key ideas of the proofs will heavily resemble the proofs of the correlation decay property in \Cref{sec:correlation-decay}.

Recall that we have fixed a CSP $\Phi=(V,\+Q,\+C)$ and a constraint $c_0\in \+C$.  Also, recall the notations in \Cref{definition:imaginary-sampler-quantities}, the $K$-truncated coupling tree $\+T=\+T_K(\Phi,c_0)$ in \Cref{definition:coupling-tree}, and the linear program in \Cref{definition:linear-program}.

\subsection{Remaining proof of \Cref{lemma:lp-feasibility}}

Recall that for each
$(\+E,\+F,\sigma,\tau,B)\in V(\+T)$, we let
\begin{equation}\label{eq:recall}
\hat{p}_{(\+E,\+F,\sigma,\tau,B)}^X=\mu_{\+C}(\+F\land \tau), \quad \hat{p}_{(\+E,\+F,\sigma,\tau,B)}^Y=\mu_{\+C\setminus \{c_0\}}(\+E\land \sigma).
\end{equation}

We then finish the proof of \Cref{lemma:lp-feasibility} by showing the above choice satisfies \Cref{item:linear-program-overflow-constraints} of \Cref{definition:linear-program}. We only prove the first inequality, and the second inequality follows analogously.

We need the definition of a random process, which is almost the same process as \Cref{definition:procedure-root-to-leaf-path-1}, except that now $\*X$ follows the distribution of $\+P$ instead of $\mu_{\+C\setminus \{c_0\}}$. Also, the random process now can be interpreted as generating a root-to-leaf path in $\+T$. Here in this random process, it is possible to reach those nodes $(\+E,\+F,\sigma,\tau,B)\in \+L_{\text{invld}}$ due to the different distribution of $\*X$.

\begin{definition}[random process associated with the true marginal probabilities]\label{definition:procedure-root-to-leaf-path-2}
Let  $\*X^{\mg}\sim \+P$ and $\*Y^{\mg}\sim \mu_{\+C}$ be drawn independently beforehand. 
Define the random process $P^{\mg}=P_K^{\mg}=\left\{(\+E_t,\+F_t,\sigma_t,\tau_t,B_t)\right\}_{t\ge 0}$ starting from the initial state $(\+E_0,\+F_0,\sigma_0,\tau_0,B_0)=(\+C\setminus \{c_0\},\+C,\varnothing,\varnothing,\emptyset)$ as follows:
    \begin{enumerate}
        \item If $\sigma$ violates $\+E$ or $\tau$ violates $\+F$ or or $\+E^{\sigma}=\+F^{\tau}$ or $|B|=K$, the process stops and $(\+E_t,\+F_t,\sigma_t,\tau_t,B_t)$ is the outcome of the process.
       \item Otherwise, suppose $\+F_t^{\tau_t}\not\subseteq \+E_t^{\sigma_t}$. Let $c$ be the smallest pinned constraint in $\+F_t^{\tau_t}\setminus \+E_t^{\sigma_t}$.
       \[
    \left(\+E_{t+1},\+F_{t+1},\sigma_{t+1},\tau_{t+1},B_{t+1}\right)\gets
        \begin{cases}
        \left(\+E_{t}\cup \{c\},\+F_{t},\sigma_{t},\tau_{t},B_{t}\right)   & 
         \text{$c$ is satisfied by $\*X^{\mg}_{\vbl(c)}$};\\
             \left(\+E_{t},\+F_{t},\sigma_{t}\land \*X^{\mg}_{\vbl(c)},\tau_{t}\land \*Y^{\mg}_{\vbl(c)},B_{t}\Join c\right)  & \text{otherwise}.
        \end{cases}
        \]
        \item Otherwise $\+F_{t}^{\tau}\subseteq \+E_{t}^{\sigma}$. Let $c$ be the smallest pinned constraint in $\+E_{t}^{\sigma}\setminus \+F_{t}^{\tau}$. 
        \[
       \left(\+E_{t+1},\+F_{t+1},\sigma_{t+1},\tau_{t+1},B_{t+1}\right)\gets
        \begin{cases}
           \left(\+E_{t},\+F_{t}\cup \{c\},\sigma_{t},\tau_{t},B_{t}\right)   & 
         \text{$c$ is satisfied by $\*Y_{\vbl(c)}$};\\
         \left(\+E_{t},\+F_{t},\sigma_{t}\land \*X^{\mg}_{\vbl(c)},\tau_{t}\land \*Y^{\mg}_{\vbl(c)},B_{t}\Join c\right)   & \text{otherwise}.
        \end{cases}
        \]
    \end{enumerate}
Let $\mu^{\mg}$ denote the distribution of the outcome $(\+E_\infty,\+F_\infty,\sigma_\infty,\tau_\infty,B_\infty)$ of this process.
\end{definition}

By comparing with the definition of $\+T=\+T_K(\Phi,c_0)$ in \Cref{definition:coupling-tree}. We have all tuples possibly visited by \Cref{definition:procedure-root-to-leaf-path} are in $V(\+T)$.

The following lemma holds for the process in \Cref{definition:procedure-root-to-leaf-path-2}. 

\begin{lemma}\label{lemma:true-marginal-property}
Assume \Cref{condition:main-condition}. For each $(\+E,\+F,\sigma,\tau,B)\in V(\+T)$,
\[
(\+E,\+F,\sigma,\tau,B)\in P^{\mg}\Longleftrightarrow \*X^{\mg}\in \+Q^{\+E\setminus (\+C\setminus \{c_0\})\land \sigma}\land \*Y^{\mg}\in \+Q^{\+F\land \tau}.
\]
\end{lemma}

The following corollary is then direct by combining \Cref{lemma:true-marginal-property} with the law of conditional probability.

\begin{corollary}\label{corollary:true-marginal-property}
Assume \Cref{condition:main-condition}. For each $\bm{x}\in \+Q$ and each $(\+E,\+F,\sigma,\tau,B)\in V(\+T)$, 
\[
\Pr{(\+E,\+F,\sigma,\tau,B)\in P^{\mg}\mid \*X^{\mg}=\bm{x}}=\begin{cases}
0 & \bm{x}\notin \+Q^{\+E\setminus (\+C\setminus \{c_0\})\land \sigma};\\
\mu_{\+C}(\+F\land \tau) & \bm{x}\in \+Q^{\+E\setminus (\+C\setminus \{c_0\})\land \sigma}.
\end{cases}
\]
\end{corollary}

\begin{proof}
For each $\bm{x}\in \+Q$ and each $(\+E,\+F,\sigma,\tau,B)\in \+L$ we have
\begin{equation}\label{eq:true-marginal-transform}
\begin{aligned}
    \Pr{(\+E,\+F,\sigma,\tau,B)\in P^{\mg}\mid \*X^{\mg}=\bm{x}}=&\frac{\Pr{(\+E,\+F,\sigma,\tau,B)\in P^{\mg}\land \*X^{\mg}=\bm{x}}}{\Pr{\*X^{\mg}=\bm{x}}}\\
    (\text{by \Cref{definition:procedure-root-to-leaf-path-2}})\quad =&|\+Q|\cdot \Pr{(\+E,\+F,\sigma,\tau,B)\in P^{\mg}\land \*X^{\mg}=\bm{x}}.
\end{aligned}
\end{equation}
By \Cref{lemma:true-marginal-property}, if $\bm{x}\notin  \+Q^{\+E\setminus (\+C\setminus \{c_0\})\land \sigma}$, then the above term clearly equals $0$. We then assume $\bm{x}\in  \+Q^{\+E\setminus (\+C\setminus \{c_0\})\land \sigma}$. Note that by \Cref{definition:procedure-root-to-leaf-path-2} and \Cref{lemma:true-marginal-property}, conditioning on the event $(\+E,\+F,\sigma,\tau,B)\in P^{\mg}$, the distribution of $\*X^{\mg}$ is uniform over all assignments in $\+Q^{\+E\setminus (\+C\setminus \{c_0\})\land \sigma}$. Therefore,
\begin{align*}
\Pr{(\+E,\+F,\sigma,\tau,B)\in P^{\mg}\land \*X^{\mg}=\bm{x}}=& \Pr{(\+E,\+F,\sigma,\tau,B)\in P^{\mg}}\cdot \frac{1}{|\+Q^{\+E\setminus (\+C\setminus \{c_0\})\land \sigma}|}\\
(\text{by \Cref{lemma:true-marginal-property}})\quad=& \frac{|\+Q^{\+E\setminus (\+C\setminus \{c_0\})\land \sigma}|}{|\+Q|}\cdot \frac{|\+Q^{\+F\land \tau}|}{|\+Q^{\+C}|}\cdot \frac{1}{|\+Q^{\+E\setminus (\+C\setminus \{c_0\})\land\sigma}|}\\
(\star)\quad=& \frac{1}{|\+Q|}\cdot \mu_{\+C}(\+F\land \tau),
\end{align*}
combining with \eqref{eq:true-marginal-transform} proves the corollary. Here, the $\star$ equality is by that $\+F\land \tau\implies \+C$ for each $(\+E,\+F,\sigma,\tau,B)\in V(\+T)$, following the argument in \Cref{remark:coupling-tree-well-defined}.
\end{proof}

We are now ready to finish the proof of \Cref{lemma:lp-feasibility}.

\begin{proof}[Remaining proof of \Cref{lemma:lp-feasibility}]
Let $\left(\+E^{\mg},\+F^{\mg},\sigma^{\mg},\tau^{\mg},B^{\mg}\right)\sim\mu^{\mg}$ denote the random outcome of the process $P^\mg$ constructed in \Cref{definition:procedure-root-to-leaf-path-2}.

Recall that $\mathbb{T}_K^{c_0}$ is the set of $2$-trees in $G_{\Phi}$ of size $K$ containing $c_0$. By \Cref{corollary:true-marginal-property}, we have for each $\bm{x}\in \+Q$ and each $T\in \mathbb{T}_K^{c_0}$,

\begin{equation}\label{eq:mg-transform}
    \begin{aligned}
    \Pr{B^{\mg}=T\mid \*X^{\mg}=\bm{x}}= &\sum\limits_{\substack{(\+E,\+F,\sigma,\tau,B)\in \+L_{\trun}:\\ B=T}}\Pr{(\+E,\+F,\sigma,\tau,B)\in P^{\mg}\mid \*X^{\mg}=\bm{x}}\\
    (\text{by \Cref{corollary:true-marginal-property} and \eqref{eq:recall}})\quad =&\sum\limits_{\substack{(\+E,\+F,\sigma,\tau,B)\in \+L_{\trun}:\\ B=T\land \bm{x} \in \+Q^{\+E\setminus (\+C\setminus \{c_0\})}}}\hat{p}_{(\+E,\+F,\sigma,\tau,B)}^X,
\end{aligned}
\end{equation}
which is exactly the expression in the left-hand side of the first inequality in 
\Cref{item:linear-program-overflow-constraints} of \Cref{definition:linear-program}.

Fix some $T\in \mathbb{T}_K^{c_0}$. Now it remains to bound $\Pr{B^{\mg}=T\mid \*X^{\mg}=\bm{x}}$. Here, we leverage the analysis in \Cref{lemma:disjoint-prob}.  Following the argument in \Cref{lemma:disjoint-prob}, we have that for each $c\in T$, $c\in B^{\mg}$ implies 
$$\+A^{\mg}_c: \forall v\in \vbl(c), \quad\*X^{\mg}(v)=\vio{c}(v) \lor \*Y^{\mg}(v)=\vio{c}(v).$$ 

Therefore, we have 
\begin{align*}
\Pr{B^{\mg}=T\mid \*X^{\mg}=\bm{x}}\leq &\Pr{\bigwedge\limits_{c\in T}\+A^{\mg}_c\mid \*X^{\mg}=\bm{x}}\\
(\text{by \Cref{definition:procedure-root-to-leaf-path-2}})\quad= &\Pr[\*Y\sim \mu_{\+C}]{\bigwedge\limits_{c\in T}\+A_c\mid \*X=\bm{x}}\\
(\text{by \Cref{HSS}})\quad \leq & (1-\mathrm{e}p)^{K(D+1)}\Pr[\*Y\sim \+P]{\bigwedge\limits_{c\in T}\+A_c\mid \*X=\bm{x}}.
\end{align*}

Combining with \eqref{eq:mg-transform}, \Cref{lemma:lp-feasibility} is proved.
\end{proof}

We finish this subsection by proving \Cref{lemma:true-marginal-property}.

\begin{proof}[Proof of \Cref{lemma:true-marginal-property}]

By $\Phi$ satisfy \Cref{condition:main-condition}, $\*Y^{\mg}$ and the procedure in \Cref{definition:procedure-root-to-leaf-path-2} are well-defined.  We then prove by structural induction in the top-down order that for any $(\+E,\+F,\sigma,\tau,B)$ such that $\Pr{(\+E,\+F,\sigma,\tau,B)\in P^{\mg}}>0$,
\begin{equation}\label{eq:true-marginal-property-1}
(\+E,\+F,\sigma,\tau,B)\in P^{\mg}\Longleftrightarrow \*X^{\mg}\in  \+Q^{\+E\setminus (\+C\setminus \{c_0\})\land \sigma} \land \*Y^{\mg}\in  \+Q^{\+F\setminus \+C\land\tau},
\end{equation}
then \Cref{lemma:true-marginal-property} directly holds.

 The base case is when $(\+E,\+F,\sigma,\tau,B)=(\+C\setminus \{c_0\},\+C,\varnothing,\varnothing,\emptyset)$. In this case, $(\+E,\+F,\sigma,\tau,B)\in P^{\mg}$ and  $\+Q^{\+E\setminus (\+C\setminus \{c_0\})\land \sigma}=\+Q^{\+F\setminus \+C\land\tau}=\+Q$. So \eqref{eq:true-marginal-property-1} holds, proving the base case.

For the induction step, we assume for the current $(\+E,\+F,\sigma,\tau,B)$ that $\+E^{\sigma}\neq \+F^{\tau}$, $|B|<K$, $\sigma$ doesn't violate $\+E$ and $\tau$ doesn't violate $\+F$. We then only prove the case for $\+F^{\tau}\not\subseteq \+E^{\sigma}$. The case when $\+F^{\tau}\subseteq \+E^{\sigma}$ follows analogously.

let $c$ be the smallest constraint in $\+F^{\tau}\setminus \+E^{\sigma}$. By \Cref{definition:procedure-root-to-leaf-path-2}, we have the following two cases: 
\begin{itemize}
    \item The next tuple is $(\+E\cup \{c\},\+F,\sigma,\tau,B)$, then we have
    \begin{align*}
    &(\+E\cup \{c\},\+F,\sigma,\tau,B)\in P^{\mg}\\
    (\text{by \Cref{definition:procedure-root-to-leaf-path-1}})\quad\Longleftrightarrow \quad& (\+E,\+F,\sigma,\tau,B)\in P^{\mg}\land \*X^{\mg}_{\vbl(c_0)} \text{ satisfies }c\\
    (\text{by I.H.})\quad\Longleftrightarrow \quad& \*X^{\mg}\in  \+Q^{\+E\setminus (\+C\setminus \{c_0\})\land \sigma} \land \*Y^{\mg}\in  \+Q^{\+F\land \tau}\land \*X^{\mg}_{\vbl(c_0)} \text{ satisfies }c\\
    \Longleftrightarrow \quad& \*X^{\mg}\in  \+Q^{\+E\setminus (\+C\setminus \{c_0\})\cup \{c\}\land\sigma} \land \*Y^{\mg}\in  \+Q^{\+F\land \tau}\\
    (\star)\quad  \Longleftrightarrow \quad& \*X^{\mg}\in  \+Q^{(\+E\cup \{c\})\setminus (\+C\setminus \{c_0\})\land\sigma} \land \*Y^{\mg}\in  \+Q^{\+F\land \tau}.
    \end{align*}
   Here, the $\star$ implication is by that each $c\in \+F^{\tau}\setminus \+E^{\sigma}$ chosen by the random process in \Cref{definition:procedure-root-to-leaf-path-2} must satisfy either
   \begin{itemize}
       \item is exactly $c_0$;
       \item or is pinned by $\sigma/\tau$, i.e. $\vbl(c)\neq \vbl(\unpin{c})$,
   \end{itemize}
   and hence cannot be in $\+C\setminus \{c_0\}$.
    \item The next tuple is $\left(\+E,\+F,\sigma\land \pi,\tau\land \rho,B\Join c\right)$ for $\pi=\vio{c}$ and some $\rho\in\+Q_{\vbl(c)}$, then we have
     \begin{align*}
    &(\+E,\+F,\sigma\land \pi,\tau\land \rho,B\Join c)\in P^{\mg}\\
    (\text{by \Cref{definition:procedure-root-to-leaf-path-1}})\quad\Longleftrightarrow \quad& (\+E,\+F,\sigma,\tau,B)\in P^{\mg}\land \*X^{\mg}\text{ violates } c\land \*X^{\mg}_{\vbl(c)}=\pi\land \*Y^{\mg}_{\vbl(c)}=\rho\\
     (\text{by $\pi=\vio{c}$})\quad\Longleftrightarrow \quad& (\+E,\+F,\sigma,\tau,B)\in P^{\mg}\land \*X^{\mg}_{\vbl(c)}=\pi\land \*Y^{\mg}_{\vbl(c)}=\rho\\
    (\text{by I.H.})\quad\Longleftrightarrow \quad& \*X^{\mg}\in  \+Q^{\+E\setminus (\+C\setminus \{c_0\})\land \sigma} \land \*Y^{\mg}\in  \+Q^{\+F\land \tau}
    \land \*X^{\mg}_{\vbl(c)}=\pi\land \*Y^{\mg}_{\vbl(c)}=\rho\\
    \Longleftrightarrow \quad& \*X^{\mg}\in  \+Q^{\+E\setminus (\+C\setminus \{c_0\})\land\sigma\land \pi} \land \*Y^{\mg}\in  \+Q^{\+F\land\tau\land \rho},
    \end{align*}
    finishing the induction steps and the proof of the lemma.
\end{itemize}
\end{proof}

\subsection{Proof of \Cref{lemma:linear-program-bad-leaf-loss} and \Cref{lemma:sampling-alg-path-probability-bound}}

In this subsection, we simultaneously prove \Cref{lemma:linear-program-bad-leaf-loss} and \Cref{lemma:sampling-alg-path-probability-bound}. Similar as the introduction of \Cref{definition:procedure-root-to-leaf-path-1} in \Cref{sec:correlation-decay} and \Cref{definition:procedure-root-to-leaf-path-2} in the previous subsection, we need to define the following random process for explicitly identified randomness:

  \begin{definition}[two random processes associated with the linear program]\label{definition:procedure-root-to-leaf-path}
We consider the following way of generating root-to-leaf paths of $\+T$.  Initiate a random assignment $\*X^{\lp}\in \+Q$ distributed as either
\begin{itemize}
    \item $\*X^{\lp}\sim \+P$\; (type $1$ initialization)
    \item or $\*X^{\lp}\sim \mu_{\+C\setminus \{c_0\}}$\; (type $2$ initialization)
\end{itemize}
    The process starts from the root node $(\+C\setminus \{c_0\},\+C,\varnothing,\varnothing,\emptyset)$ of $\+T$. 
        In each step, at a non-leaf node $(\+E,\+F,\sigma,\tau,B)$, it does:
    
\begin{enumerate}
        \item if $\+F^{\tau}\not\subseteq \+E^{\sigma}$, denoted by $c$ be the smallest pinned constraint in $\+F^{\tau}\setminus \+E^{\sigma}$, 
        \begin{itemize}
            \item move to $(\+E\cup \{c\},\+F,\sigma,\tau,B)$ if $c$ is satisfied by $\*X^{\lp}_{\vbl(c)}$;
            \item otherwise, move to $(\+E,\+F,\sigma\land \*X^{\lp}_{\vbl(c)},\tau\land \rho,B\Join c)$  for each $\rho\in \+Q_{\vbl(c)}$ w.p. 
            \[
            \frac{\hat{p}^X_{\+E,\+F,\sigma\land \*X^{\lp}_{\vbl(c)},\tau\land \rho,B\Join c}}{\hat{p}_{(\+E,\+F,\sigma,\tau,B)}^X};
            \]
        \end{itemize}
        \item otherwise $\+F^{\tau}\subseteq \+E^{\sigma}$, denoted by $c$ be the smallest pinned constraint in $\+E^{\sigma}\setminus \+F^{\tau}$,
        \begin{itemize}
            \item move to $(\+E,\+F\cup \{c\},\sigma,\tau,B)$ with probability
            \[
            \frac{\hat{p}^X_{(\+E,\+F\cup \{c\},\sigma,\tau,B)}}{\hat{p}_{(\+E,\+F,\sigma,\tau,B)}^X};
            \]
            \item move to $(\+E,\+F,\sigma\land \*X^{\lp}_{\vbl(c)},\tau\land \rho,B\Join c)$ for $\rho=\vio{c}$  with probability   \[
            \frac{\hat{p}^X_{(\+E,\+F,\sigma\land \*X^{\lp}_{\vbl(c)},\tau\land \rho,B\Join c)}}{\hat{p}_{(\+E,\+F,\sigma,\tau,B)}^X}.
            \]
        \end{itemize}

       \end{enumerate}
       \sloppy
            By \Cref{item:linear-program-non-leaf-constraints} of \Cref{definition:linear-program}, it can be verified that for both types of initialization of $\*X^{\lp}$, the above process generates a probability distribution over root-to-leaf paths of $\+T$. For the two processes with $\*X^{\lp}$ initialized as type $1$/type $2$, we refer to the $\*X^{\lp}$ as $\*X^{\lpone}/\*X^{\lptwo}$, respectively. We let the root-to-leaf path generated be $P^{\lpone}/P^{\lp2}$, and the distribution on the leaf node induced by the process be $\mu^{\lpone}/\mu^{\lptwo}$, respectively.
\end{definition}

We then present probability bounds for \Cref{definition:procedure-root-to-leaf-path} with different initialization, concluded by the two following lemmas. 

\begin{lemma}\label{lemma:lp-procedure-1-probability-bound}
Assume \Cref{condition:main-condition}. For each $(\+E,\+F,\sigma,\tau,B)\in V(\+T)$, it holds that
\[
\Pr{(\+E,\+F,\sigma,\tau,B)\in P^{\lpone}}=\P{(\+E\setminus (\+C\setminus \{c_0\}))\land \sigma}\cdot \hat{p}_{(\+E,\+F,\sigma,\tau,B)}^X.
\]
Moreover, conditioning on $(\+E,\+F,\sigma,\tau,B)\in P^\lpone$, it follows that
\[
\*X^{\lpone}\sim \mu_{\+E\setminus (\+C\setminus \{c_0\})}^{\sigma},
\]
for each $(\+E,\+F,\sigma,\tau,B)$ such that 
$\Pr{(\+E,\+F,\sigma,\tau,B)\in P^{\lpone}}>0$.
\end{lemma}

\begin{lemma}\label{lemma:lp-procedure-2-probability-bound}
Assume \Cref{condition:main-condition}. For each $(\+E,\+F,\sigma,\tau,B)\in V(\+T)$, it holds that
\[
\Pr{(\+E,\+F,\sigma,\tau,B)\in P^{\lptwo}}=\mu_{\+C\setminus \{c_0\}}(\+E\land \sigma)\cdot \hat{p}_{(\+E,\+F,\sigma,\tau,B)}^X.
\]
Moreover, conditioning on $(\+E,\+F,\sigma,\tau,B)\in P^\lptwo$, it follows that
\[
\*X^{\lptwo}\sim \mu_{\+E}^{\sigma},
\]
for each $(\+E,\+F,\sigma,\tau,B)$ such that 
$\Pr{(\+E,\+F,\sigma,\tau,B)\in P^{\lptwo}}>0$.
\end{lemma}

By \Cref{lemma:lp-procedure-2-probability-bound} and comparing \Cref{definition:procedure-root-to-leaf-path} with \Cref{definition:marginal-sampling-algorithm},  \Cref{lemma:sampling-alg-path-probability-bound} is immediately proved. Also, we have the following corollary by combining \Cref{lemma:lp-procedure-1-probability-bound} with the law of conditional probability.

\begin{corollary}\label{corollary:lp-procedure-1-probability-bound}
Assume \Cref{condition:main-condition}. For each $\bm{x}\in \+Q$ and each $(\+E,\+F,\sigma,\tau,B)\in V(\+T)$, 
\[
\Pr{(\+E,\+F,\sigma,\tau,B)\in P^{\lpone}\mid \*X^{\lpone}=\bm{x}}=\begin{cases}
0 & \bm{x}\notin \+Q^{\+E\setminus (\+C\setminus \{c_0\})\land \sigma};\\
\hat{p}_{(\+E,\+F,\sigma,\tau,B)}^X & \bm{x}\in \+Q^{\+E\setminus (\+C\setminus \{c_0\})\land \sigma}.
\end{cases}
\]
\end{corollary}

\begin{proof}
For each $\bm{x}\in \+Q$ and each $(\+E,\+F,\sigma,\tau,B)\in V(\+T)$ we have
\begin{equation}\label{eq:lp-procedure-transform}
\begin{aligned}
    \Pr{(\+E,\+F,\sigma,\tau,B)\in P^\lpone\mid \*X^{\lpone}=\bm{x}}=&\frac{\Pr{(\+E,\+F,\sigma,\tau,B)\in P^\lpone\land \*X^{\lpone}=\bm{x}}}{\Pr{\*X^{\lpone}=\bm{x}}}\\
    (\text{by \Cref{definition:procedure-root-to-leaf-path}})\quad =&|\+Q|\cdot \Pr{(\+E,\+F,\sigma,\tau,B)\in P^\lpone\land \*X^{\lpone}=\bm{x}}.
\end{aligned}
\end{equation}
Note that by \Cref{lemma:lp-procedure-1-probability-bound}, if $\bm{x}\notin  \+Q^{\+E\setminus (\+C\setminus \{c_0\})\land \sigma}$, then the above term clearly equals to $0$. We then assume $\bm{x}\in  \+Q^{\+E\setminus (\+C\setminus \{c_0\})\land \sigma}$. Note that by \Cref{lemma:lp-procedure-1-probability-bound}, conditioning on the event $ (\+E,\+F,\sigma,\tau,B)$, the distribution of $\sigma$ is uniform over all assignments in $\+Q^{\+E\setminus (\+C\setminus \{c_0\})\land \sigma}$. Therefore,
\begin{align*}
\Pr{(\+E,\+F,\sigma,\tau,B)\in P^{\lpone}\land \*X^{\lpone}=\bm{x}}=& \Pr{(\+E,\+F,\sigma,\tau,B)\in P^{\lpone}}\cdot \frac{1}{|\+Q^{\+E\setminus (\+C\setminus \{c_0\})\land \sigma}|}\\
(\text{by \Cref{lemma:lp-procedure-1-probability-bound}})\quad=& \frac{|\+Q^{\+E\setminus (\+C\setminus \{c_0\})\land \sigma}|}{|\+Q|}\cdot \hat{p}_{(\+E,\+F,\sigma,\tau,B)}^X\cdot \frac{1}{|\+Q^{\+E\setminus (\+C\setminus \{c_0\})\land \sigma}|}\\
=& \frac{1}{|\+Q|}\cdot \hat{p}_{(\+E,\+F,\sigma,\tau,B)}^X,
\end{align*}
combining with \eqref{eq:lp-procedure-transform} proves the corollary. 
\end{proof}

The proofs of Lemmas \ref{lemma:lp-procedure-1-probability-bound} and \ref{lemma:lp-procedure-2-probability-bound} follow a direct induction on the coupling tree. We provide the proofs of Lemmas \ref{lemma:lp-procedure-1-probability-bound} and \ref{lemma:lp-procedure-2-probability-bound} at the end of this subsection for completeness.

We are now ready to prove \Cref{lemma:linear-program-bad-leaf-loss}.
\begin{proof}[Proof of \Cref{lemma:linear-program-bad-leaf-loss}]
We only prove the first inequality, and the second inequality follows analogously.
  
 Note that  we have
\begin{equation}\label{eq:linear-program-leaf-loss-1}
\begin{aligned}
&\frac{1}{|\+Q^{\+C\setminus \{c_0\}}|}\sum\limits_{\bm{x}\in \+Q^{\+C\setminus \{c_0\}}}\sum\limits_{\substack{(\+E,\+F,\sigma,\tau,B)\in \+L_{\trun}:\\\bm{x}\in \+Q^{\+E\land \sigma}}}\hat{p}_{(\+E,\+F,\sigma,\tau,B)}^X\\
(\star)\quad=&\sum\limits_{\bm{x}\in \+Q^{\+C\setminus \{c_0\}}}\sum\limits_{\substack{(\+E,\+F,\sigma,\tau,B)\in \+L_{\trun}:\\\bm{x}\in \+Q^{\+E\land \sigma}}}\frac{\Pr{(\+E,\+F,\sigma,\tau,B)\in P^\lptwo}}{|\+Q^{\+E\land \sigma}|}\\
(\blacktriangle)\quad=&\sum\limits_{(\+E,\+F,\sigma,\tau,B)\in \+L_{\trun}}\Pr{(\+E,\+F,\sigma,\tau,B)\in P^\lptwo},
\end{aligned}
\end{equation}
where the $\star$ equality is by \Cref{lemma:lp-procedure-2-probability-bound} and that $\+E\land \sigma\implies \+C\setminus \{c_0\}$ for each $(\+E,\+F,\sigma,\tau,B)\in V(\+T)$, following the argument in \Cref{remark:coupling-tree-well-defined}. The $\blacktriangle$ equality is still by that $\+E\land \sigma\implies \+C\setminus \{c_0\}$ and exchanging the order of summation.

Let $\left(\+E^{\lpone},\+F^{\lpone},\sigma^{\lpone},\tau^{\lpone},B^{\lpone}\right)\sim\mu^{\lpone}$ and $\left(\+E^{\lptwo},\+F^{\lptwo},\sigma^{\lptwo},\tau^{\lptwo},B^{\lptwo}\right)\sim\mu^{\lptwo}$  denote the random leaf of $V(\+T)$ generated by the two processes in \Cref{definition:procedure-root-to-leaf-path}, respectively.

Let $\mathbb{T}_K^{c_0}$ be the set of $2$-trees in $G_{\Phi}$ of size $K$ containing $c_0$. Note that following the proof of \Cref{lemma:bad-constraint-2-tree} we have $B\subseteq \mathbb{T}_K^{c_0}$ for all $(\+E,\+F,\sigma,\tau,B)\in \+L_{\trun}$. We then have

\begin{equation}\label{eq:linear-program-bad-leaf-loss-1}
\begin{aligned}
&\frac{1}{|\+Q^{\+C\setminus \{c_0\}}|}\sum\limits_{\bm{x}\in \+Q^{\+C\setminus \{c_0\}}}\sum\limits_{\substack{(\+E,\+F,\sigma,\tau,B)\in \+L_{\trun}:\\\bm{x}\in \+Q^{\+E\land \sigma}}}\hat{p}_{(\+E,\+F,\sigma,\tau,B)}^X\\
(\text{by \eqref{eq:linear-program-leaf-loss-1}})\quad = & \sum\limits_{(\+E,\+F,\sigma,\tau,B)\in \+L_{\trun}}\Pr{(\+E,\+F,\sigma,\tau,B)\in P^\lptwo}\\
= & \sum\limits_{T\in \mathbb{T}_K^{c_0}}\Pr{ B^{\lptwo}=T} \\
= & \sum\limits_{T\in \mathbb{T}_K^{c_0}}\sum\limits_{\bm{x}\in \+Q^{\+C\setminus \{c_0\}}}\Pr{\*X^{\lptwo}=\bm{x}}\cdot \Pr{B^{\lptwo}=T\mid \*X^{\lptwo}=\bm{x}}\\
(\star)\quad= & \sum\limits_{T\in \mathbb{T}_K^{c_0}}\sum\limits_{\bm{x}\in \+Q^{\+C\setminus \{c_0\}}}\Pr{\*X^{\lptwo}=\bm{x}}\cdot \Pr{B^{\lpone}=T\mid \*X^{\lpone}=\bm{x}}\\
(\text{by \Cref{corollary:lp-procedure-1-probability-bound}})\quad = & \sum\limits_{T\in \mathbb{T}_K^{c_0}}\sum\limits_{\bm{x}\in \+Q^{\+C\setminus \{c_0\}}}\Pr{\*X^{\lptwo}=\bm{x}}\cdot \sum\limits_{\substack{(\+E,\+F,\sigma,\tau,B)\in \+L_{\trun}:\\
\bm{x}\in \+Q^{\+E\setminus (\+C\setminus \{c_0\})\land \sigma}\land B=T}}\hat{p}_{(\+E,\+F,\sigma,\tau,B)}^X.
\end{aligned}
\end{equation}

 Here, the $\star$ equality is by noting that we can perfectly couple the random processes in \Cref{definition:procedure-root-to-leaf-path} with two types of initialization when $\*X^{\lpone}$ and $\*X^{\lptwo}$ are the same value in $\+Q^{\+C\setminus \{c_0\}}$.

Recall $\+A_c$ denotes the event defined in \eqref{eq:definition-ec}, i.e.
       \[
       \+A_c: \forall v\in \vbl(c), \quad\*X(v)=\vio{c}(v) \lor \*Y(v)=\vio{c}(v).
       \]
       We further note that
 \begin{align*}
&\sum\limits_{T\in \mathbb{T}_K^{c_0}}\sum\limits_{\bm{x}\in \+Q^{\+C\setminus \{c_0\}}}\Pr{\*X^{\lptwo}=\bm{x}}\cdot \sum\limits_{\substack{(\+E,\+F,\sigma,\tau,B)\in \+L_{\trun}:\\
\bm{x}\in \+Q^{\+E\setminus (\+C\setminus \{c_0\})\land \sigma}\land B=T}}\hat{p}_{(\+E,\+F,\sigma,\tau,B)}^X\\
(\star)\quad \leq &  (1-\mathrm{e}p)^{(D+1)K}\sum\limits_{T\in \mathbb{T}_K^{c_0}}\sum\limits_{\bm{x}\in \+Q^{\+C\setminus \{c_0\}}} \Pr{\*X^{\lptwo}=\bm{x}}\Pr[\*Y\sim \+P]{\bigwedge\limits_{c\in T}\+A_c\mid \*X=\bm{x}}\\
(\blacktriangle)\quad= &  (1-\mathrm{e}p)^{(D+1)K}\sum\limits_{T\in \mathbb{T}_K^{c_0}}\sum\limits_{\bm{x}\in \+Q_{\vbl(T)}} \Pr{\*X^{\lptwo}_{\vbl(T)}=\bm{x}}\cdot \Pr[\*Y\sim \+P]{\bigwedge\limits_{c\in T}\+A_c\mid \*X_{\vbl(T)}=\bm{x}}\\
(\text{by \Cref{definition:procedure-root-to-leaf-path}})\quad\leq &(1-\mathrm{e}p)^{2(D+1)K}\sum\limits_{T\in \mathbb{T}_K^{c_0}}\sum\limits_{\bm{x}\in \+Q_{\vbl(T)}} \Pr[\*X\sim \+P]{\*X_{\vbl(T)}=\bm{x}}\cdot \Pr[\*Y\sim \+P]{\bigwedge\limits_{c\in T}\+A_c\mid \*X_{\vbl(T)}=\bm{x}}\\\
\quad = &(1-\mathrm{e}p)^{2(D+1)K}\sum\limits_{T\in \mathbb{T}_K^{c_0}} \Pr[\*X,\*Y\sim \+P]{\bigwedge\limits_{c\in T}\+A_c}\\
(\blacksquare)\quad\leq &(1-\mathrm{e}p)^{2(D+1)K}\sum\limits_{T\in \mathbb{T}_K^{c_0}} p^{\frac{2}{2+\zeta}}\\
(\text{by \Cref{lemma:num-2-tree}})\quad\leq &\left(\mathrm{e}D^2\cdot p^\frac{2}{2+\zeta}\cdot (1-\mathrm{e}p)^{-2(D+1)}\right)^{K}\\
(\text{by \Cref{condition:main-condition}})\quad \leq &2^{-K},
 \end{align*}
which combining with \eqref{eq:linear-program-bad-leaf-loss-1} finishes the proof of the lemma. 

Here, the $\star$ inequality is by \Cref{item:linear-program-overflow-constraints} of \Cref{definition:linear-program}. The $\blacktriangle$ equality is by the definition of $\+A_{c}$ that the expression
\[
\Pr[\*Y\sim \+P]{\bigwedge\limits_{c\in T}\+A_c}
\]
only depends on $\*X_{\vbl(T)}$. The $\blacksquare$ inequality is by following the argument in the proof of \Cref{lemma:disjoint-prob}.
\end{proof}

We finish this subsection and the overall section by proving Lemmas \ref{lemma:lp-procedure-1-probability-bound} and \ref{lemma:lp-procedure-2-probability-bound}.

\begin{proof}[Proof of \Cref{lemma:lp-procedure-1-probability-bound}]
    We prove the lemma by a structural induction in a top-down order.
    
     The base case is when at the root node, i.e., $(\+E,\+F,\sigma,\tau,B)=(\+C\setminus \{c_0\},\+C,\varnothing,\varnothing,\emptyset)$. Note that in this case we have $\Pr{(\+E,\+F,\sigma,\tau,B)\in P^{\lpone}}=1$ and
    \[
    \P{(\+E\setminus (\+C\setminus \{c_0\}))\land \sigma}=1,\quad \hat{p}_{(\+E,\+F,\sigma,\tau,B)}^X=1,
    \]
    where the second equality is by \Cref{item:coupling-tree-2} of \Cref{definition:linear-program}. Also, it follows from \Cref{definition:procedure-root-to-leaf-path} that
    \[
        \*X^{\lpone}\sim \+P=\mu_{\+E\setminus (\+C\setminus \{c_0\})}^{\sigma}.
    \]
    The base case is proved.
    
     For the induction step, we only prove the case when $\+F^{\tau}\not\subseteq \+E^{\sigma}$. The case when $\+F^{\tau}\subseteq \+E^{\sigma}$ follows analogously.  Let $c$ be the smallest constraint in $ \+F^{\tau}\setminus \+E^{\sigma}$. From \Cref{definition:procedure-root-to-leaf-path}, we have the following cases:
     \begin{itemize}
        \item $c$ is satisfied by $\*X^{\lpone}$. In this case, we go to $(\+E\cup \{c\},\+F,\sigma,\tau,B)$. By the induction hypothesis, we have $\*X^{\lpone}\sim \mu_{\+E\setminus (\+C\setminus \{c_0\})}^{\sigma}$ and therefore this event ($c$ is satisfied by $\*X^{\lpone}$) happens with probability $\mu_{\+E\setminus (\+C\setminus \{c_0\})}^{\sigma}(c)$. Then
        \begin{align*}
     &\Pr{(\+E\cup \{c\},\+F,\sigma,\tau,B)\in P^{\lpone}}\\
     (\text{by \Cref{definition:procedure-root-to-leaf-path}})\quad=&  \Pr{(\+E,\+F,\sigma,\tau,B)\in P^{\lpone}}\cdot  \mu_{\+E\setminus (\+C\setminus \{c_0\})}^{\sigma}(c)\\
     (\text{by I.H.})\quad=&  \P{(\+E\setminus (\+C\setminus \{c_0\}))\land \sigma}\cdot \hat{p}_{(\+E,\+F,\sigma,\tau,B)}^X\cdot \mu_{\+E\setminus (\+C\setminus \{c_0\})}^{\sigma}(c)\\
     (\text{by the chain rule})\quad=& \P{((\+E\setminus (\+C\setminus \{c_0\}))\cup \{c\})\land \sigma}\cdot \hat{p}_{(\+E,\+F,\sigma,\tau,B)}^X\\
      (\star)\quad =& \P{((\+E\cup \{c\})\setminus (\+C\setminus \{c_0\}))\land \sigma}\cdot \hat{p}_{(\+E,\+F,\sigma,\tau,B)}^X\\
     (\text{by \Cref{item:linear-program-non-leaf-constraints} of \Cref{definition:linear-program}})\quad=& \P{((\+E\cup \{c\})\setminus (\+C\setminus \{c_0\}))\land \sigma}\cdot \hat{p}^X_{(\+E\cup \{c\},\+F,\sigma,\tau,B)}.
         \end{align*}
         Here, the $\star$ equality is by that each $c\in \+F^{\tau}\setminus \+E^{\sigma}$ chosen by the random process in \Cref{definition:procedure-root-to-leaf-path-2} must satisfy either
   \begin{itemize}
       \item is exactly $c_0$; 
       \item or is pinned by $\sigma$, i.e., $\vbl(c)\neq \vbl(\unpin{c})$,
   \end{itemize}
   and hence cannot be in $\+C\setminus \{c_0\}$.
   Also, conditioning on going to $(\+E\cup \{c\},\+F,\sigma,\tau,B)$, it follows that
    \[
        \*X^{\lpone}\sim \mu_{\+E\setminus (\+C\setminus \{c_0\})\cup \{c\}}^{\sigma}=\mu_{(\+E\cup \{c\})\setminus (\+C\setminus \{c_0\})}^{\sigma},        \]
        finishing the proof of this case.
         
         \item $c$ is violated by $\*X^{\lpone}$. Let $\pi=\*X^{\lpone}_{\vbl(c)}$, then we have $\pi=\vio{c}$.  By the induction hypothesis, we have $\*X^{\lpone}\sim \mu_{\+E\setminus (\+C\setminus \{c_0\})}^{\sigma}$ and this event ($c$ is violated by $\*X^{\lpone}$) happens with probability \[
         \mu_{\+E\setminus (\+C\setminus \{c_0\})}^{\sigma}(\neg c) = \mu_{\+E\setminus (\+C\setminus \{c_0\})}^{\sigma}(\pi).
         \]
         In this case, for each $\rho\in \+Q_{\vbl(c)}$ we go to $(\+E,\+F,\sigma\land \pi,\tau\land \rho,B\Join c)$ with probability $\frac{\hat{p}^X_{(\+E,\+F,\sigma\land \pi,\tau\land \rho,B\Join c)}}{\hat{p}_{(\+E,\+F,\sigma,\tau,B)}^X}$.
        Hence for each $\rho\in \+Q_{\vbl(c)}$,
          \begin{align*}
     &\Pr{(\+E,\+F,\sigma\land \pi,\tau\land \rho,B\Join c)\in P^{\lpone}}\\
     (\star)\quad=&  \Pr{(\+E,\+F,\sigma,\tau,B)\in P^{\lpone}}\cdot \mu_{\+E\setminus (\+C\setminus \{c_0\})}^{\sigma}(\pi)\cdot \frac{\hat{p}^X_{(\+E,\+F,\sigma\land \pi,\tau\land \rho,B\Join c)}}{\hat{p}_{(\+E,\+F,\sigma,\tau,B)}^X}\\
     (\text{by I.H.})\quad=&  \P{(\+E\setminus (\+C\setminus \{c_0\}))\land \sigma}\cdot \hat{p}_{(\+E,\+F,\sigma,\tau,B)}^X\cdot \mu_{\+E\setminus (\+C\setminus \{c_0\})}^{\sigma}(\pi)\cdot \frac{\hat{p}^X_{(\+E,\+F,\sigma\land \pi,\tau\land \rho,B\Join c)}}{\hat{p}_{(\+E,\+F,\sigma,\tau,B)}^X}\\
     (\blacktriangle)\quad=& \P{(\+E\setminus (\+C\setminus \{c_0\}))\land (\sigma\land \pi)}\cdot \hat{p}^X_{(\+E,\+F,\sigma\land \pi,\tau\land \rho,B\Join c)},
         \end{align*}

         where the $\star$ equality is by \Cref{definition:procedure-root-to-leaf-path} and the $\blacktriangle$ equality is by the chain rule.
    Also, conditioning on going to $(\+E,\+F,\sigma\land \pi,\tau\land \rho,B\Join c)$, it follows that
    \[
        \*X^{\lpone}\sim  \mu_{\+E\setminus (\+C\setminus \{c_0\})}^{\sigma\land \pi},
        \]
        finishing the proof of this case and the lemma.
     \end{itemize}
    \end{proof}

    \begin{proof}[Proof of \Cref{lemma:lp-procedure-2-probability-bound}]
   We prove the lemma by a structural induction in a top-down order.
    
     The base case is when at the root node, i.e., $(\+E,\+F,\sigma,\tau,B)=(\+C\setminus \{c_0\},\+C,\varnothing,\varnothing,\emptyset)$. Note that in this case we have $\Pr{(\+E,\+F,\sigma,\tau,B)\in P^{\lptwo}}=1$ and
    \[
   \mu_{\+C\setminus \{c_0\}}(\+E\land \sigma)=1,\quad \hat{p}_{(\+E,\+F,\sigma,\tau,B)}^X=1,
    \]
    where the second equality is by \Cref{item:coupling-tree-2} of \Cref{definition:linear-program}. Also, it follows from \Cref{definition:procedure-root-to-leaf-path} that
    \[
         \*X^{\lptwo}\sim \mu_{\+C\setminus \{c_0\}}=\mu_{\+E}^{\sigma}.
    \]
    The base case is proved.
    
     For the induction step, we only prove the case when $\+F^{\tau}\not\subseteq \+E^{\sigma}$. The case when $\+F^{\tau}\subseteq \+E^{\sigma}$ follows analogously.  Let $c$ 
be the smallest constraint in $\+F^{\tau}\setminus \+E^{\sigma}$. From \Cref{definition:procedure-root-to-leaf-path}, we have the following cases:
     \begin{itemize}
        \item $c$ is satisfied by $\*X^{\lptwo}$. In this case, we go to $(\+E\cup \{c\},\+F,\sigma,\tau,B)$. By the induction hypothesis, we have $\*X^{\lptwo}\sim \mu_{\+E}^{\sigma}$ and this event ($c$ is satisfied by $\*X^{\lptwo}$) happens with probability $\mu_{\+E}^{\sigma}(c)$. Then
        \begin{align*}
     &\Pr{(\+E\cup \{c\},\+F,\sigma,\tau,B)\in P^{\lptwo}}\\
     (\text{by \Cref{definition:procedure-root-to-leaf-path}})\quad=&  \Pr{(\+E,\+F,\sigma,\tau,B)\in P^{\lptwo}}\cdot \mu_{\+E}^{\sigma}(c)\\
     (\text{by I.H.})\quad=&  \mu_{\+C\setminus \{c_0\}}(\+E\land \sigma)\cdot \hat{p}_{(\+E,\+F,\sigma,\tau,B)}^X\cdot \mu_{\+E}^{\sigma}(c)\\
     (\star)\quad=& \mu_{\+C\setminus \{c_0\}}((\+E\cup \{c\})\land \sigma)\cdot \hat{p}_{(\+E,\+F,\sigma,\tau,B)}^X\\
     (\text{by \Cref{item:linear-program-non-leaf-constraints} of \Cref{definition:linear-program}})\quad=& \mu_{\+C\setminus \{c_0\}}((\+E\cup \{c\})\land \sigma)\cdot \hat{p}^X_{(\+E\cup \{c\},\+F,\sigma,\tau,B)},
         \end{align*}
     where the $\star$ equality is by the chain rule and that $\+E\land \sigma\implies \+C\setminus \{c_0\}$ for each $(\+E,\+F,\sigma,\tau,B)\in V(\+T)$, following the argument in \Cref{remark:coupling-tree-well-defined}.
   Also, conditioning on going to $(\+E\cup \{c\},\+F,\sigma,\tau,B)$, it follows that
    \[
        \*X^{\lptwo}\sim \mu_{\+E\cup \{ c\}}^{\sigma},
        \]
        finishing the proof of this case.
         
         \item $c$ is violated by $\*X^{\lptwo}$. Let $\pi=\*X^{\lptwo}_{\vbl(c)}$. By the induction hypothesis, we have $\*X^{\lptwo}\sim \mu_{\+E}^{\sigma}$ and this event ($c$ is violated by $\*X^{\lptwo}$) happens with probability \[
         \mu_{\+E}^{\sigma}(\neg c)=\mu_{\+E}^{\sigma}(\pi).
         \]
         In this case, for each $\rho\in \+Q_{\vbl(c)}$ we go to $(\+E,\+F,\sigma\land \pi,\tau\land \rho,B\Join c)$ with probability $\frac{\hat{p}^X_{(\+E,\+F,\sigma\land \pi,\tau\land \rho,B\Join c)}}{\hat{p}_{(\+E,\+F,\sigma,\tau,B)}^X}$.
        Hence, for each $\rho\in \+Q_{\vbl(c)}$,
          \begin{align*}
     &\Pr{(\+E,\+F,\sigma\land \pi,\tau\land \rho,B\Join c)\in P^{\lptwo}}\\
     (\text{by \Cref{definition:procedure-root-to-leaf-path}})\quad=&  \Pr{(\+E,\+F,\sigma,\tau,B)\in P^{\lptwo}}\cdot \mu_{\+E}^{\sigma}(\pi)\cdot \frac{\hat{p}^X_{(\+E,\+F,\sigma\land \pi,\tau\land \rho,B\Join c)}}{\hat{p}_{(\+E,\+F,\sigma,\tau,B)}^X}\\
     (\text{by I.H.})\quad=&  \mu_{\+C\setminus \{c_0\}}(\+E\land \sigma)\cdot \hat{p}_{(\+E,\+F,\sigma,\tau,B)}^X\cdot \mu_{\+E}^{\sigma}(\pi)\cdot \frac{\hat{p}^X_{(\+E,\+F,\sigma\land \pi,\tau\land \rho,B\Join c)}}{\hat{p}_{(\+E,\+F,\sigma,\tau,B)}^X}\\
     (\star)\quad =& \mu_{\+C\setminus \{c_0\}}(\+E\land (\sigma\land \pi))\cdot \hat{p}^X_{(\+E,\+F,\sigma\land \pi,\tau\land \rho,B\Join c)},
         \end{align*}
             where the $\star$ equality is by the chain rule and that $\+E\land \sigma\implies \+C\setminus \{c_0\}$ for each $(\+E,\+F,\sigma,\tau,B)\in V(\+T)$, following the argument in \Cref{remark:coupling-tree-well-defined}.
    Also, conditioning on going to $(\+E,\+F,\sigma\land \pi,\tau\land \rho,B\Join c)$, it follows that
    \[
        \*X^{\lp2}\sim \mu_{\+E}^{\sigma\land \pi},
        \]
        finishing the proof of this case and the lemma.
     \end{itemize}
    \end{proof}

\end{document}